\documentclass{article}[12pt]
\usepackage{graphicx}
\usepackage[english]{babel}
\usepackage[autostyle]{csquotes}
\usepackage{amsmath,amsfonts}
\usepackage{amsthm}
\usepackage{geometry}
\usepackage{subcaption}
\usepackage[toc,page]{appendix}
\usepackage{hyperref}
\hypersetup{
	colorlinks=true,
	linkcolor=blue,
	filecolor=blue,      
	urlcolor=blue,
	citecolor = blue,
	anchorcolor = blue 
}
\newtheorem{theorem}{Theorem}[section]

\newtheorem{algorithm}[theorem]{Algorithm}

\newtheorem{remark}{Remark}[section]

\providecommand{\keywords}[1]{\textbf{\textbf{Key Words:}} #1}

\newtheorem{manuallemmainner}{Lemma}
\newenvironment{manuallemma}[1]{%
	\renewcommand\themanuallemmainner{#1}%
	\manuallemmainner
}{\endmanuallemmainner}

\begin{document}

\title{\bf Robust Clustering with Normal Mixture Models: A Pseudo $\beta$-Likelihood Approach}
\author{Soumya Chakraborty\footnote{Joint Affiliation at Bethune College (Assistant Professor) and Indian Statistical Institute (Research Fellow)}, Ayanendranath Basu\footnote{ayanbasu@isical.ac.in (Corresponding Author)} ~and Abhik Ghosh
\\
Indian Statistical Institute, Kolkata, India. \\
  }

\maketitle

\begin{abstract}
	As in other estimation scenarios, likelihood based estimation in the normal mixture set-up is highly non-robust against model misspecification and presence of outliers (apart from being an ill-posed optimization problem). A robust alternative to the ordinary likelihood approach for this estimation problem is proposed which performs simultaneous estimation and data clustering and leads to subsequent anomaly detection. To invoke robustness, the methodology based on the minimization of the density power divergence (or alternatively, the maximization of the $\beta$-likelihood) is utilized under suitable constraints. An iteratively reweighted least squares approach has been followed in order to compute the proposed estimators for the component means (or equivalently cluster centers) and component dispersion matrices which leads to simultaneous data clustering. Some exploratory techniques are also suggested for anomaly detection, a problem of great importance in the domain of statistics and machine learning. The proposed method is validated with simulation studies under different set-ups; it performs competitively or better compared to the popular existing methods like K-medoids, TCLUST, trimmed K-means and MCLUST, especially when the mixture components (i.e., the clusters) share regions with significant overlap or outlying clusters exist with small but non-negligible weights (particularly in higher dimensions). Two real datasets are also used to illustrate the performance of the newly proposed method in comparison with others along with an application in image processing. The proposed method detects the clusters with lower misclassification rates and successfully points out the outlying (anomalous) observations from these datasets.
\\
\\

\end{abstract}
\section{Introduction}

Mixture distributions arise in many common practical situations. 
In particular, when the population is not homogeneous due to the presence of different categorical attributes, the variable of interest has different behaviour over distinct homogeneous subgroups which come together to generate an overall heterogeneous mixture system. To draw statistical inference based on this kind of heterogeneous datasets, a single probability distribution may not be adequate to model the data; mixture distributions provide the appropriate structure in these situations. Mixtures of many different distributions with varying shapes have been used in the literature to model datasets 
coming from different disciplines ranging over astronomy, clinical psychology, economics, finance, DNA sequencing, 
image processing, voice recognition, criminology, species counting and many others. 
In parametric mixture modelling, one uses a model probability distribution constructed as a mixture (convex combination) of 
several probability distributions from a particular parametric class with different parameter values. 
Mathematically, given a parametric family of densities  $\mathcal{F}_{\boldsymbol{\Theta}}=\{f_{\boldsymbol{\theta}}:\boldsymbol{\theta} \in \boldsymbol{\Theta}\subset {\mathbb R}^k\}$ 
indexed by the parameter vector $\boldsymbol{\theta} \in \boldsymbol{\Theta}  $, a mixture from this class can be described in terms of 
the mixture probability density function (PDF),
$
f_{\boldsymbol{\boldsymbol{\theta}}}=\sum\limits_{j=1}^{k}\pi_{j}f_{\boldsymbol{\boldsymbol{\theta}}_{j}}
$    
where $\pi_{j}$ is the weight given to the $j$-th component of the mixture having PDF $f_{\boldsymbol{\theta}_{j}}$ for $j=1,..,k$ with $\sum_{j=1}^{k}\pi_{j}=1$, 
and $\boldsymbol{\theta}=(\boldsymbol{\theta}_{1},\boldsymbol{\theta}_{2},...,\boldsymbol{\theta}_{k})$ is the parameter vector of interest. 
In practice, the weights $\pi_{j}$s are unknown and hence they also have to be estimated along with the parameters ($\boldsymbol{\theta}_{j}$s) of the component distributions. 
Normal, Cauchy and Laplace are perhaps the most common symmetric mixture examples. The normal mixture models are flexible, can fit a large variety of shapes, 
and are among the  most popular and most used statistical tools in practice. 
Moreover, gamma scale mixtures of normal components cover a very large class of symmetric mixture distributions 
(e.g., \cite{sclaemixnorm}). For non-symmetric or skewed mixture distributions, 
uniform mixtures of normal distributions are used (e.g., \cite{sclaemix2}). For general references covering different areas of mixture models (including non-normal mixtures), 
see \cite{titterington}, \cite{lindsay} and \cite{peel}.

Motivated by their huge applicability, here we develop a robust estimation procedure for normal mixture models. However, likelihood based inference, which is asymptotically the most efficient under the model, is strongly affected by outliers and model misspecification. To address the robustness issue, we have taken a minimum distance approach in the spirit of the density power divergence (DPD). The DPD class was originally proposed (under a different name) by \cite{basupaper}. We are going to view the problem of minimizing the density power divergence as a maximization of a generalized likelihood function. One of our primary objectives in performing robust inference in normal mixtures is its subsequent use in robust clustering. Clustering is an active area of research and has many real life applications. Some of the existing clustering methods available in the literature are K-means, K-medoids and standard likelihood discrimination. The problem with the K-means method is that it tends to find only spherical or elliptical clusters with \enquote{roughly} equal cluster sizes and equal component covariance matrices. Also the method is based on traditional cluster mean estimates which makes the algorithm non-robust in the presence of \enquote{anomalous} observations. The presence of anomaly is not a rare thing in practice. 
Subjective deletion of outliers, inadvisable as it is, cannot be done in high dimensions where the data cannot be visualized.  
Hence, some of the small components of the cluster may be misspecified and the observations coming from these irregular clusters become disturbing.

Various modifications of the K-means clustering algorithm have been proposed in the literature with the aim of robustifying the algorithm. 
The trimmed K-means method (\cite{trimmedkmeans}) is one such example
where the same objective function as the K-means is used but only using a subsample after trimming the extreme observations from the whole sample.  
%
%
To further generalize the K-means and trimmed K-means algorithms beyond spherical or elliptical clusters, 
the concept of heterogeneous clustering has been considered in the literature. 
\cite{gallegosritter} proposed a normal mixture set-up in this context under the \enquote{spurious-outliers model}
and developed an algorithm for estimation which is a \textit{naive extension} of the fast MCD algorithm \cite{reviewpaper}.  
But the estimation procedure under this model is too difficult because 
the algorithm often ends up finding clusters made up of observations lying on a low dimensional subspace. 
Moreover the likelihood function can be unbounded; when one of the observations is equal to one of the component means, the likelihood tends to infinity as the determinant of the dispersion matrix of that component goes to zero. 
To bypass this difficulty, \cite{Gallegos} and \cite{gallegosritter} 
proposed two additional methods based on the structure of the component covariance matrices. 
The first one proposed an algorithm assuming similar scales of the components 
while the second one assumes that the dispersion matrices have an unknown but same covariance structure. 
Later, \cite{tclust} proposed the TCLUST approach 
which performs a likelihood based discrimination with trimming a certain proportion of outlying observations;
it is still based on the maximum likelihood estimates of the parameters but obtained only from non-trimmed observations. 
This method is very popular and heavily used for robust clustering under normal mixture data. 

Some of the other existing robust clustering algorithms based on finite gaussian mixture models include \cite{punzo1}, \cite{mclust}, \cite{raftery} which mainly model the contaminated datasets with finite mixture Gaussian models (non-Gaussian models are also used, eg., \cite{raftery}) and perform parameter estimation with subsequent data clustering and outlier detection. 


Divergence based clustering procedures are also used to derive robust and efficient parametric clustering literature. Here, the clustering methods are based on robust estimation of the model parameters by minimizing suitable divergence measures (between probability density functions). The $\gamma$-divergence is one of those divergences which are utilized to build robust clustering tools. The $\gamma$-divergence was originally proposed by \cite{gamma2001} and was motivated later by \cite{dpdjapan2008} through the $\gamma$-cross entropy. Two interesting robust clustering algorithms were proposed by utilizing the $\gamma$-divergence, viz., \cite{gammaclust2} and \cite{gammaclust1}. The former developed a robust clustering algorithm by combining the $q$-Gaussian mixture model and minimum $\gamma$-divergence estimation procedure and applied the algorithm on cryo-EM data where the latter developed a hierarchical robust clustering algorithm by minimizing the $\gamma$-divergence.

In the present paper, we have developed a fully parametric robust approach for the estimation of the parameters under the normal mixture model 
which leads to a subsequent robust clustering strategy. 
It is useful in the same estimation and clustering set-up for which TCLUST is one of the existing robust procedures, 
but the source of robustness of our estimators and our algorithm lies in suitable density power downweighting of the observations 
rather than through invoking likelihood based trimming. 
A section of our assumptions are similar to those necessary for TCLUST, 
although there are also some assumptions specific to the form of the DPD. The contributions of the present work are highlighted in the following paragraphs. 

Firstly, we present a new robust method for model based clustering under the Gaussian mixture set-up developed in the spirit of the density power divergence, which contains the Kullback-Leibler divergence as a special (extreme) case. Thus, the proposed set of techniques represent a general class of estimation methods which contains likelihood based methods as a particular case. A single scalar tuning parameter controls the trade-off between efficiency and robustness and we demonstrate how positive but small values of the tuning parameter provide more stable performance under noisy data with little loss in efficiency compared to the likelihood based results. 

Secondly, we develop an approximate EM-like algorithm to solve the problem efficiently even in higher dimensions. This is important since a straightforward optimization of the proposed objective function is difficult particularly in high dimensions. The proposed algorithm performs parameter estimation as a precursor to the detection of the clusters and helps to detect anomalous observations (if present) in the dataset. Subsequently, the algorithm leads to robust detection of clusters on the basis of the robust parameter estimates obtained earlier.

Thirdly, under noisy data, the proposed method provides improved results in terms of the estimated misclassification rates and outlier detection compared to the TCLUST (which possibly represents the most popular robust and clustering algorithm currently in use in this area which has an easily implementable software) and the trimmed K-means methods. Along with the robust methods like the TCLUST, trimmed K-means and K-medoids algorithms, we have also included the MCLUST method (\cite{mclust}) in our analysis which optimally selects the number of clusters unlike the others. However, we use the MCLUST method with fixed number of clusters (with uniform noise component) in this work.

Finally, we illustrate the usefulness of the proposed clustering procedure in image processing to identify 
differently colored (anomalous) regions from a colour image. With appropriately proposed additional refinements,
our methodology is seen to outperform the aforesaid procedures also for this special application
as illustrated through analysis of a real satellite image.  

The rest of the paper is organized as follows. 
In Section \ref{section2}, we propose our algorithm along with the underlying theoretical formulations. In Section \ref{SEC:theoretical_result}, we present some of the theoretical properties of our estimators (discussed in a detailed fashion in the Appendices) and the behaviour of the influence functions is explored to justify the claimed robustness. A large scale comparative simulation study is presented in Section \ref{secsim}. 
Analysis of two real life datasets are considered in Section \ref{secrealdata} 
while an application of our method to image processing is provided in Section \ref{secimage}. Conclusions and future plans are discussed in Section \ref{secconfu}. For brevity, technical proofs and derivations are presented in the Appendices.

\section{Proposed Parameter Estimation and Clustering Procedure}
\label{section2}
Let $\boldsymbol{X}_{1}, \boldsymbol{X}_{2},...,\boldsymbol{X}_{n}$ be a random sample drawn from a $p$-dimensional multivariate normal mixture distribution with $k$ components having an unknown PDF which is modelled by the model density $f_{\boldsymbol{\theta}}(\boldsymbol{x})=\sum_{j=1}^{k}\pi_{j}\phi_{p}(\boldsymbol{x},\boldsymbol{\mu}_{j},\boldsymbol{\Sigma}_{j})$ for $\boldsymbol{x} \in \mathbb{R}^{p}$, where $\phi_{p}(\cdot,\boldsymbol{\mu},\boldsymbol{\Sigma})$ denotes the PDF of a $p$-dimensional normal with mean $\boldsymbol{\mu}$ and dispersion matrix $\boldsymbol{\Sigma}$. The parameter $\boldsymbol{\theta}$ is given by $\boldsymbol{\theta} = (\pi_{1}, \pi_{2},...., \pi_{k}, \boldsymbol{\mu}_{1}, \boldsymbol{\mu}_{2},...,\boldsymbol{\mu}_{k}, \boldsymbol{\Sigma}_{1}, \boldsymbol{\Sigma}_{2},..., \boldsymbol{\Sigma}_{k})$ where $\boldsymbol{\mu}_{j} \in  \mathbb{R}^{p}$, $\boldsymbol{\Sigma}_{j}$ is a real symmetric, positive definite $p \times p$ matrix and $0 \leq \pi_{j} \leq 1$ is the weight of the $j$-th component, $j=1,2,...,k$, with $\sum_{j=1}^{k}\pi_{j}=1$. Our objective is to estimate the parameter $\boldsymbol{\theta}$ robustly and to detect the true clusters. Instead of the ordinary likelihood method, we propose a generalized likelihood approach which is motivated by the minimum DPD methodology of \cite{basupaper}, and subsumes the ordinary maximum likelihood approach.  
\subsection{Theoretical Formulation}
\label{SEC:theoretical_method}

To provide a brief background on the DPD, let us consider a simple random sample $\boldsymbol{Z}_{1},\boldsymbol{Z}_{2},...,\boldsymbol{Z}_{n}$ 
from an unknown probability distribution with PDF $g$ having cumulative distribution function (CDF) $G$. We model this unknown PDF $g$  by a parametric family of densities $\mathcal{F}_{\boldsymbol{\Theta}}=\{f_{\boldsymbol{\theta}}: \boldsymbol{\theta} \in \boldsymbol{\Theta}\} $. In the minimum distance approach, the best approximation of $g$ in the model family is given by the member of $\mathcal{F}_{\boldsymbol{\Theta}}$ which is \enquote{closest} to $g$, in terms of an appropriate statistical distance measure. For this purpose we will utilize the DPD given by
\begin{equation}
\centering
\label{Eq1}
D_{\beta}(g,f_{\boldsymbol{\theta}})=\int \left[f_{\boldsymbol{\theta}}^{1+\beta}(\boldsymbol{x})-\left(1+\frac{1}{\beta}\right)g(\boldsymbol{x})f_{\boldsymbol{\theta}}^{\beta}(\boldsymbol{x})+\frac{1}{\beta}g^{1+\beta}(\boldsymbol{x})\right]\;d\boldsymbol{x},
\end{equation}
where $\beta$ is a non-negative tuning parameter. From efficiency considerations, the practical range of the parameter $\beta$ is usually taken to be the interval $[0, 1]$. 
It may observed that, 
\begin{equation}
\label{Eq1.1}
D_{\beta}(g,f_{\boldsymbol{\theta}}) =\int f_{\boldsymbol{\theta}}^{1+\beta}(\boldsymbol{x}) \;d\boldsymbol{x} - \left(1+\frac{1}{\beta}\right)E_{G}(f_{\boldsymbol{\theta}}^{\beta}(\boldsymbol{X})) + \frac{1}{\beta}\int g^{1+\beta}(\boldsymbol{x})\;d\boldsymbol{x}.
\end{equation}
The last term in Equation (\ref{Eq1.1}) is independent of $\boldsymbol{\theta}$. Hence it is enough to minimize the objective function
\begin{equation}
\label{Eq1.11}
\int f_{\boldsymbol{\theta}}^{1+\beta}(\boldsymbol{x}) \;d\boldsymbol{x} - \left(1+\frac{1}{\beta}\right)E_{G}(f_{\boldsymbol{\theta}}^{\beta}(\boldsymbol{X})) 
\end{equation}
in order to find the model element $f_{\boldsymbol{\theta}}$ which is closest to $g$ in the minimum DPD sense. 

From the Glivenko-Cantelli Theorem, the empirical CDF $G_{n}$ based on $\boldsymbol{Z}_{1},\boldsymbol{Z}_{2},...,\boldsymbol{Z}_{n}$ is strongly consistent for $G$ and as a consequence, we can approximate Equation (\ref{Eq1.11}) by its empirical version
\begin{align*}
&\int f_{\boldsymbol{\theta}}^{1+\beta}(\boldsymbol{x}) \;d\boldsymbol{x} - \left(1+\frac{1}{\beta}\right)\frac{1}{n}\sum_{i=1}^{n}f_{\boldsymbol{\theta}}^{\beta}(\boldsymbol{Z}_{i}) 
\end{align*}
and minimize this with respect to $\boldsymbol{\theta}$ in order to obtain the minimum DPD estimators (MDPDEs) of the unknown parameters.
Note that the aforesaid minimization problem can also be viewed as the maximization of,
\begin{equation}
\centering
\label{Eq2}
l_{\beta}(\boldsymbol{\theta})= \left(1+\frac{1}{\beta}\right)\frac{1}{n}\sum_{i=1}^{n}f_{\boldsymbol{\theta}}^{\beta}(\boldsymbol{Z}_{i})-\int f_{\boldsymbol{\theta}}^{1+\beta}(\boldsymbol{x}) \;d\boldsymbol{x}.
\end{equation}
As $\beta\rightarrow 0$, the MDPDE coincides with the maximum likelihood estimator (MLE) as the limit $l_0(\boldsymbol{\theta})$  coincides with the log-likelihood function. For a general $\beta>0$, the quantity in the right hand side of Equation (\ref{Eq2}) has been referred to as the $\beta$-likelihood in \cite{dpdjapan} and its maximizer is termed as the maximum $\beta$-likelihood estimator. Also see \cite{alpha_div} for a description of $\alpha$, $\beta$ and $\gamma$ divergences (leading to the corresponding generalized likelihoods) as robust methods of similarities. For the $\beta$-likelihood, the tuning parameter $\beta \in [0,1]$ controls the trade off between the degree of robustness and asymptotic efficiency of the MDPDE, or, equivalently, the maximum $\beta$-likelihood estimator of $\boldsymbol{\theta}$. The asymptotic properties and the robustness properties of this estimator have been discussed in \cite{basupaper} and \cite{basubook}. The relationship between the MDPD and the minimum $q$-entropy estimation was established by \cite{q_entropy}. The same quantity (in the right hand side of Equation (\ref{Eq2})) is indeed also equal to the $L_q$-likelihood function (\cite{q_entropy})  with $q = 1-\beta$; its use in robust statistical inferences has been explored by \cite{q_entropy_finance, q_entropy_tail, L_q likelihood, qent3, qent1, qent2}. 

Returning to the original problem of normal mixture models, the joint likelihood is given by
\begin{equation}
\centering
\label{Eq3}
L_{M}(\boldsymbol{\theta},F_{n})= \prod_{i=1}^{n}f_{\boldsymbol{\theta}}(\boldsymbol{X}_{i}), ~~~~\mbox{ with } ~f_{\boldsymbol{\theta}}(\boldsymbol{x})=\sum_{j=1}^{k}\pi_{j}\phi_{p}(\boldsymbol{x},\boldsymbol{\mu}_{j},\boldsymbol{\Sigma}_{j}),
\end{equation}
where 
$F_{n}=\frac{1}{n}\sum_{i=1}^{n}\delta_{\boldsymbol{X}_i}$ is the empirical probability measure of the data. 
The classical MLE is defined as the maximizer of $ L_{M}(\boldsymbol{\theta},F_{n})$ with respect to $\boldsymbol{\theta}\in\boldsymbol{\Theta}$.
The corresponding $\beta$-likelihood can again be defined by Equation (\ref{Eq2}),  
but now $f_{\boldsymbol{\theta}}$ is as given in (\ref{Eq3}). This may be maximized  with respect to $\boldsymbol{\theta}$ to obtain the maximum $\beta$-likelihood estimator (or the MDPDE) of $\boldsymbol{\theta}$  for a given value of $\beta$. 
However, due to the presence of a summation term in $f_{\boldsymbol{\theta}}$ and the integral of its power in the objective function, the associated optimization problem becomes extremely difficult.
\cite{dpdjapan} have proposed an algorithm for this particular optimization problem to obtain the MDPDEs
for the univariate ($p=1$) normal mixture models. 
But the algorithm is very difficult to implement in higher dimensions
and hence the computation of the maximum $\beta$-likelihood estimator in a normal mixture model with $p>1$ still remains a challenging problem. 

We hereby propose an alternative EM like algorithm for robust parameter estimation for the normal mixture models
without directly maximizing the DPD objective function (the $\beta$-likelihood) as done by \cite{dpdjapan}.
In particular, we consider an alternative version of the likelihood using the $\beta$-likelihood of the individual component densities, as described below,  
rather than considering the $\beta$-likelihood for the overall mixture density $f_{\boldsymbol{\theta}}$ as in \cite{dpdjapan}.
Our approach leads to a valid objective function which has a much simpler form that is fairly straightforward to maximize through EM type iterative algorithms 
even for higher dimensions ($p>1$). 
As a result, our algorithm also leads to clustering and outlier detection and is structurally similar to the TCLUST algorithm 
(\cite{tclust}) but the source of robustness is different. 
Instead of performing a likelihood based trimming, we invoke the $\beta$-likelihood from the minimum DPD approach. 
The motivation comes from the fact that outliers (if present) may also provide useful information about the system; so they should be further scrutinized rather than be eliminated by trimming. 

In order to describe our proposed algorithm, let us note that even the likelihood function of the normal mixture model given in (\ref{Eq3}) 
is difficult to maximize directly with respect to $\boldsymbol{\theta}$ and a different expression for the likelihood function is used for the computation of MLE via EM algorithms. 
Consider the missing assignment functions 
\[
Z_{j}(\boldsymbol{X}_{i},\boldsymbol{\theta})= 
\begin{cases}
1,& \text{if } \boldsymbol{X}_{i} \in C_{j}\\
0,              & \text{otherwise}
\end{cases}
\]
with $C_{j}$ as the $j$-th cluster, $j=1,2,...,k$. Using these assignment functions, the likelihood function can also be presented as,
\begin{align*}
L(\boldsymbol{\theta},F_{n})=\prod_{j=1}^{k}\prod_{i \in C_{j}}\pi_{j}\phi_{p}(\boldsymbol{X}_{i},\boldsymbol{\mu}_{j},\boldsymbol{\Sigma}_{j})= \prod_{i=1}^{n}\prod_{j=1}^{k}\pi_{j}^{Z_{j}(\boldsymbol{X}_{i},\boldsymbol{\theta})}{\phi_{p}(\boldsymbol{X}_{i},\boldsymbol{\mu}_{j},\boldsymbol{\Sigma}_{j})}^{Z_{j}(\boldsymbol{X}_{i},\boldsymbol{\theta})}.
\end{align*}
Instead of maximizing $L(\boldsymbol{\theta},F_{n})$, it is mathematically equivalent and convenient to maximize,
\begin{align*}
l(\boldsymbol{\theta},F_{n})&=
\;\text{log}\;L(\boldsymbol{\theta},F_{n})\\
&= \sum_{i=1}^{n}\sum_{j=1}^{k}Z_{j}(\boldsymbol{X}_{i},\boldsymbol{\theta})[\log\;\pi_{j}+\log\;\phi_{p}(\boldsymbol{X}_{i},\boldsymbol{\mu}_{j},\boldsymbol{\Sigma}_{j})]= \sum_{j=1}^{k}[n_{j}\log\;\pi_{j}+\sum_{i \in C_{j}} \log\;\phi_{p}(\boldsymbol{X}_{i},\boldsymbol{\mu}_{j},\boldsymbol{\Sigma}_{j})],
\end{align*}
where $n_{j}=\sum_{i=1}^{n}Z_{j}(\boldsymbol{X}_{i},\boldsymbol{\theta})$ represents the number of observation in the $j$-th cluster for $j=1, \ldots, k$.

Now our goal is to use the $\beta$-likelihood instead of the ordinary log-likelihood for the estimation of the parameter $\boldsymbol{\theta}$. Hence, for $j=1, \ldots, k$,  we replace separately the individual term  $\sum_{i \in C_{j}} \log\;\phi_{p}(\boldsymbol{X}_{i},\boldsymbol{\mu}_{j},\boldsymbol{\Sigma}_{j})$ by $\dfrac{n_j}{1+\beta}l_\beta^{(j)}(\boldsymbol{\theta})$, an appropriate constant multiple of the $\beta$-likelihood function $l_\beta^{(j)}(\boldsymbol{\theta})$ of the $j$-th component density given by,
\begin{equation}
\label{Eq1.2}
l_{\beta}^{(j)}(\boldsymbol{\theta})=(1+\frac{1}{\beta})\frac{1}{n_j}\sum_{i \in C_{j}}\phi_{p}^{\beta}(\boldsymbol{X}_{i},\boldsymbol{\mu}_{j},\boldsymbol{\Sigma}_{j})-\int \phi_{p}^{1+\beta}(\boldsymbol{x},\boldsymbol{\mu}_{j},\boldsymbol{\Sigma}_{j}) \;d\boldsymbol{x}.
\end{equation}
Under suitable modification with constants independent of $\boldsymbol{\theta}$, $\dfrac{n_j}{1+\beta}l_\beta^{(j)}(\boldsymbol{\theta})$ indeed converges to \\ $\sum_{i \in C_{j}} \log\;\phi_{p}(\boldsymbol{X}_{i},\boldsymbol{\mu}_{j},\boldsymbol{\Sigma}_{j})$ as $\beta$ tends to $0$. Our $\beta$ modified objective function thus becomes,\\ $ \sum_{j=1}^{k}\left[n_{j}\log\;\pi_{j}+\frac{n_{j}}{1+\beta} l_{\beta}^{(j)}(\boldsymbol{\theta})\right]$, which after some algebra simplifies to 
\begin{equation}
\label{Eq1.3}
nE_{F_{n}}\left[\sum_{j=1}^{k}Z_{j}(\boldsymbol{X},\boldsymbol{\theta})\log\;\pi_{j}+\frac{1}{\beta}\sum_{j=1}^{k}Z_{j}(\boldsymbol{X},\boldsymbol{\theta})\phi_{p}^{\beta}(\boldsymbol{X},\boldsymbol{\mu}_{j},\boldsymbol{\Sigma}_{j})-\frac{1}{1+\beta}\sum_{j=1}^{k}Z_{j}(\boldsymbol{X},\boldsymbol{\theta})\int \phi_{p}^{1+\beta}(\boldsymbol{x},\boldsymbol{\mu}_{j},\boldsymbol{\Sigma}_{j}) \;d\boldsymbol{x}\right].
\end{equation}

\noindent Hence, it is enough to maximize,
\begin{equation}
\centering
\label{Eq4}
L_{\beta}(\boldsymbol{\theta},F_{n})=E_{F_{n}}\left[\sum_{j=1}^{k}Z_{j}(\boldsymbol{X},\boldsymbol{\theta})\left[\log\;\pi_{j}+\frac{1}{\beta}\phi_{p}^{\beta}(\boldsymbol{X},\boldsymbol{\mu}_{j},\boldsymbol{\Sigma}_{j})-\frac{1}{1+\beta}\int \phi_{p}^{1+\beta}(\boldsymbol{x},\boldsymbol{\mu}_{j},\boldsymbol{\Sigma}_{j}) \;d\boldsymbol{x}\right]\right].
\end{equation}
Equation (\ref{Eq4}) represents the empirical optimization problem. 
Assuming $F$ to be the true probability measure of $\boldsymbol{X}_{1}$, the corresponding theoretical objective function is given by,
\begin{equation}
\centering
\label{Eq5}
L_{\beta}(\boldsymbol{\theta},F)=E_{F}\left[\sum_{j=1}^{k}Z_{j}(\boldsymbol{X},\boldsymbol{\theta})\left[\log\;\pi_{j}+\frac{1}{\beta}\phi_{p}^{\beta}(\boldsymbol{X},\boldsymbol{\mu}_{j},\boldsymbol{\Sigma}_{j})-\frac{1}{1+\beta}\int \phi_{p}^{1+\beta}(\boldsymbol{x},\boldsymbol{\mu}_{j},\boldsymbol{\Sigma}_{j}) \;d\boldsymbol{x}\right]\right].
\end{equation}
To solve the aforesaid estimation problem, we need a specific algebraic form of $Z_{j}(\boldsymbol{X}_{i},\boldsymbol{\theta})$. 
That is, we need a discrimination rule which can assign a particular observation $\boldsymbol{X}_{i}$ to a cluster $C_{j}$ systematically. 
The most well-known discrimination rule is based on likelihood method, originally proposed by R.A. Fisher, and was used in TCLUST method by  \cite{tclust}. We are also going to use the likelihood based discrimination rule which is defined below. 

\noindent\textbf{Discriminant Function:}  Given $\boldsymbol{\theta} \in \boldsymbol{\Theta}_{C}$, we define the discriminant functions 
\begin{align*}
D_{j}(\boldsymbol{X},\boldsymbol{\theta}) &= \pi_{j}\phi_{p}(\boldsymbol{X},\boldsymbol{\mu}_{j},\boldsymbol{\Sigma}_{j})  \;\; \text{and}\;D(\boldsymbol{X},\boldsymbol{\theta}) =\underset{1 \leq j \leq k}{\text{max}}D_{j}(\boldsymbol{X},\boldsymbol{\theta})
\end{align*}
and we include a particular observation $\boldsymbol{X}_{i}$ to the $j$-th cluster $C_{j}$ if $D(\boldsymbol{X}_{i},\boldsymbol{\theta})=D_{j}(\boldsymbol{X}_{i},\boldsymbol{\theta})$. It is easy to observe that $D_{j}(\boldsymbol{X},\boldsymbol{\theta})$ is proportional to the posterior probability of $C_j$ given $\boldsymbol{X}$ which equals $c(\boldsymbol{\theta})\pi_{j}\phi_{p}(\boldsymbol{X};\boldsymbol{\mu}_j,\boldsymbol{\Sigma}_j)$, where $c(\boldsymbol{\theta})=(\sum_{j=1}^{k}\pi_{j}\phi_{p}(\boldsymbol{X}_i;\boldsymbol{\mu}_j,\boldsymbol{\Sigma}_j))^{-1}$).

Note that although the discrimination is based on the likelihood, we compute its empirical values 
by substituting the robust parameter estimates obtained through the maximum pseudo $\beta$-likelihood method, 
which guarantees proper stability. 	
In terms of this discriminant functions, the assignment functions can be written as $Z_{j}(\boldsymbol{X}_{i},\boldsymbol{\theta})=I\left[D(\boldsymbol{X}_{i},\boldsymbol{\theta})=D_{j}(\boldsymbol{X}_{i},\boldsymbol{\theta})\right].$

This discrimination functions will also be used for outlier detection at the end of our algorithm. A small value of the discriminant function is a good indicator of possible anomaly in respect of a particular observation in relation to the presumed cluster.

We refer to the right hand side of Equation (\ref{Eq4}) as the empirical pseudo $\beta$-likelihood function and the right hand side of Equation (\ref{Eq5}) as the theoretical pseudo $\beta$-likelihood ($PL_\beta$) function. We define the maximizers of these empirical and theoretical pseudo $\beta$-likelihood functions, with respect to $\boldsymbol{\theta}$,  as the maximum pseudo $\beta$-likelihood estimator (MPLE$_{\beta}$) and the maximum pseudo $\beta$-likelihood functional (MPLF$_{\beta}$), respectively.

It may be noted that after using the alternative version of the $\beta$-likelihood for the individual component densities, our objective function in Equation (\ref{Eq4}) is no longer the objective function of the actual MDPDE of the normal mixture model. In this way we differ from the \cite{dpdjapan} approach, although the two approaches coincide for $\beta = 0$ (the case of the ordinary likelihood). The source of robustness of our procedure as well as  our motivation and philosophy are, however, strictly in line with those of the MDPDEs. Accordingly we feel that the  \enquote{pseudo $\beta$-likelihood} and the \enquote{maximum pseudo $\beta$-likelihood  estimator} represent logical nomenclature for our method and our estimator.

However as we have already noted in the previous section, the mixture normal likelihood is unbounded as a function of
the parameters so that its direct maximization is not a well defined problem. The same difficulty also arises in case of the pseudo $\beta$-likelihood of
the mixture normal model leading to singularities in the estimates of covariance matrices. To circumvent this problem in one dimension, \cite{hathway} proposed a constraint on the ratios of component standard deviations. Later, \cite{tclust} generalized this constraint in the multivariate set-up in terms of eigenvalues to avoid singularity of the dispersion matrix estimators. 
We will impose the same eigenvalue ratio constraint in our case.
Let us denote $\lambda_{jl}$ to be the $l$-th eigenvalue of the covariance matrix $\boldsymbol{\Sigma}_{j}$ for $1 \leq j \leq k$ and $1 \leq l \leq p$,
and put  $M=\underset{1 \leq j \leq k}{\text{max}}\;\underset{1 \leq l \leq p}{\text{max}}\lambda_{jl}$ and 
$m=\underset{1 \leq j \leq k}{\text{min}}\;\underset{1 \leq l \leq p}{\text{min}}\lambda_{jl}$,
the largest and smallest eigenvalues, respectively.\\

\noindent\textbf{Eigenvalue Ratio (ER) Constraint:}  For a prespecified constant $c \geq 1$, the system satisfies the condition
\begin{equation}
\centering
\label{Eq6}
\frac{M}{m} \leq c.
\end{equation}

Along with the eigenvalue ratio constraint, we will make the following additional assumption to avoid singularity and establish the existence and consistency of our proposed estimator.\\

\noindent\textbf{Non-singularity (NS) Constraint:} We assume that the smallest eigenvalue $m$ satisfies
$
m\geq c_{1}
$
for some small positive constant $c_{1}$ which is prespecified. 
\bigskip

Under the above two constraints, characterized by constants  $C=(c,c_{1})$, 
our search for the estimator can be confined with the restricted parameter space defined as
$$
\boldsymbol{\Theta}_{C}=\left\{\boldsymbol{\theta}:\boldsymbol{\theta}= (\pi_{1}, \pi_{2},...., \pi_{k}, \boldsymbol{\mu}_{1}, \boldsymbol{\mu}_{2},...,\boldsymbol{\mu}_{k}, \boldsymbol{\Sigma}_{1}, \boldsymbol{\Sigma}_{2},..., \boldsymbol{\Sigma}_{k})\;\text{with}\; \frac{M}{m} \leq c \mbox{     and } m\geq c_{1} \right\}.
$$

Here $c=1$ provides the strongest possible restriction in case of the eigenvalue ratio constraint. But a large value of $c$ is more pragmatic in the sense that the estimation problem becomes less restrictive in this case. On the other hand a small but positive value of $c_{1}$ is preferred both in terms of theoretical aspects (such as existence and consistency) as well as practical aspects (in the sense that the constraint does not become too stringent).

The non-singularity constraint is not really a stringent assumption in the presence of the eigenvalue ratio constraint. We need the non-singularity constraint (in the presence of the eigenvalue ratio constraint) only when the sequence of the smallest eigenvalue tends to $0$ and the sequence of the largest eigenvalue is of same order as that of the sequence of smallest eigenvalues. This scenario is quite rare in practice especially under the positive definiteness of the dispersion matrices. This assumption is crucial to establish some of the theoretical properties (such as existence and consistency, see Subsection \ref{SEC:theoretical_result} for details) of the parameter estimates even in case of the above mentioned pathological case.

\subsection{Computational Algorithm for the MPLE$_{\beta}$} 

To estimate the unknown parameters, to form the clusters and to detect the outliers present in the dataset, 
we need to optimize the empirical objective function on the right hand side of Equation (\ref{Eq4}). 
We hereby propose an approximate EM like algorithm which solves this empirical problem and provides reasonable estimates of the unknown parameters. We refer to this algorithm as the MPLE$_{\beta}$ algorithm. Before describing this algorithm, we need to derive another iterative procedure to find the minimum density power divergence estimators of $\boldsymbol{\mu}$ and $\boldsymbol{\Sigma}$ based on a random sample $\boldsymbol{X}_{1}, \boldsymbol{X}_{2},...,\boldsymbol{X}_{n}$ (note that this sample is not the sample mentioned at the beginning of Section \ref{section2} which is modelled by a multivariate normal mixture distribution) which is modelled with a $N_{p}(\boldsymbol{\mu},\boldsymbol{\Sigma})$ model. Let us note that, to find the aforesaid estimators, we need to minimize
\begin{align*}
\frac{1}{1+\beta}\int \phi_{p}^{1+\beta}(\boldsymbol{x},\boldsymbol{\mu},\boldsymbol{\Sigma}) \;d\boldsymbol{x} - \frac{1}{n\beta}\sum_{i=1}^{n}\phi_{p}^{\beta}(\boldsymbol{X}_{i},\boldsymbol{\mu},\boldsymbol{\Sigma})
\end{align*}
with respect to $\boldsymbol{\mu}$ and $\boldsymbol{\Sigma}$ (as in Equation (\ref{Eq2}) which is same as the maximization of $\beta$-likelihood with respect to $\boldsymbol{\mu}$ and $\boldsymbol{\Sigma}$). This minimization can be done by solving the following system of equations.

\begin{equation}
\centering
\label{eqq_irls}
\begin{array}{l}
\frac{1}{n}\sum_{i=1}^{n}e^{-\frac{\beta}{2}(\boldsymbol{X}_{i}-\boldsymbol{\mu})^{'}\boldsymbol{\Sigma}^{-1}(\boldsymbol{X}_{i}-\boldsymbol{\mu})}(\boldsymbol{X}_{i}-\boldsymbol{\mu})=0,\\
\frac{1}{n}\sum_{i=1}^{n}e^{-\frac{\beta}{2}(\boldsymbol{X}_{i}-\boldsymbol{\mu})^{'}\boldsymbol{\Sigma}^{-1}(\boldsymbol{X}_{i}-\boldsymbol{\mu})}(\boldsymbol{\Sigma}-(\boldsymbol{X}_{i}-\boldsymbol{\mu})(\boldsymbol{X}_{i}-\boldsymbol{\mu})^{'})=\frac{\beta}{(1+\beta)^{\frac{p}{2}+1}}\boldsymbol{\Sigma}.
\end{array}
\end{equation}

The mathematical derivations of the aforesaid system of equations are presented in Section $2$ of the Appendices. Here we propose an iteratively reweighted least squares (IRLS)\label{IRLS} algorithm to solve the aforesaid system of equations as follows.
\begin{algorithm}[IRLS]
	\label{algo1}
	~~~~~~~~~~~~~~~~~~~~~~~~~~~~~~~~~~~~~~~~~~~~~~~~~~~~~~~~~~~~~~~~~~~~~~~~~~~~~~~~~~~~~~~~~~~
	\begin{enumerate}
		\item\textbf{Starting Value:} A non-robust starting value may affect a robust algorithm severely. Thus we will use robust starting values for both the mean and the dispersion matrix. For the mean vector we use the componentwise sample medians, that is, $\hat{\boldsymbol{\mu}}^{0}=(u_{1},u_{2},...,u_{p})$ where $u_{j}=\text{median}\{X_{1j},X_{2j},...,X_{nj}\}$ for $1\leq j\leq p$. For the dispersion matrix we have chosen the starting value as $\hat{\boldsymbol{\Sigma}}^{0}$ such that,    
		\[
		\hat{\Sigma}^{0}_{ij}= 
		\begin{cases}
		1.4826^{2}\;\text{median}\{|X_{li}-u_{i}|^{2},1\leq l\leq n\},& \text{if } i=j,\\
		1.4826^{2}\;\text{median}\{(X_{li}-u_{i})(X_{lj}-u_{j}),1\leq l\leq n\},  & \text{otherwise}.
		\end{cases}
		\]
		$\hat{\boldsymbol{\Sigma}}^{0}$ can be treated as the multivariate generalization of the median absolute deviation (MAD) estimator of dispersion in one dimension.
		\item \textbf{Update:} Let, $\hat{\boldsymbol{\mu}}^{l}$ and $\hat{\boldsymbol{\Sigma}}^{l}$ be the estimates at the $l$-th step of iteration. 	Calculate the current weights,
		\begin{align*}
		w^{l}_{i}=e^{-\frac{\beta}{2}(\boldsymbol{X}_{i}-\hat{\boldsymbol{\mu}}^{l})^{'}(\hat{\boldsymbol{\Sigma}}^{l})^{-1}(\boldsymbol{X}_{i}-\hat{\boldsymbol{\mu}}^{l})}.
		\end{align*}
		Now update,
		\begin{align*}
		\hat{\boldsymbol{\mu}}^{l+1}=\frac{\sum_{i=1}^{n}w^{l}_{i}\boldsymbol{X}_{i}}{\sum_{i=1}^{n}w^{l}_{i}}
		\end{align*}
		and
		\begin{align*}
		\hat{\boldsymbol{\Sigma}}^{l+1}=\frac{\sum_{i=1}^{n}w^{l}_{i}(\boldsymbol{X}_{i}-\hat{\boldsymbol{\mu}}^{l+1})(\boldsymbol{X}_{i}-\hat{\boldsymbol{\mu}}^{l+1})^{'}}{\sum_{i=1}^{n}w^{l}_{i}-\frac{n\beta}{(1+\beta)^{\frac{p}{2}+1}}}.
		\end{align*}
		\item \textbf{Iteration:} Repeat step $2$ for a large number of times until, $||\hat{\boldsymbol{\mu}}^{l}-\hat{\boldsymbol{\mu}}^{l+1}|| \leq \epsilon$ and $||\hat{\boldsymbol{\Sigma}}^{l}-\hat{\boldsymbol{\Sigma}}^{l+1}|| \leq \epsilon$ for some small (prespecified) $\epsilon > 0$.
	\end{enumerate}
\end{algorithm}
Now, let us introduce the MPLE$_\beta$ algorithm for clustering, the main focus of the present proposal.

\begin{algorithm}[MPLE$_{\beta}$]
	
	~~~~~~~~~~~~~~~~~~~~~~~~~~~~~~~~~~~~~~~~~~~~~~~~~~~~~~~~~~~~~~~~~~~~~~~~~~~~~~~~~~~~~~~~~~~~~~~~~~~~~~~
	\begin{enumerate}
		\item \textbf{Initialization:}\label{MPLE} Initially, $k$ many random observations from the dataset are chosen as initial cluster centers, identity matrices of proper dimensions as initial dispersion matrices and the vector $(\frac{1}{k},\frac{1}{k},...,\frac{1}{k})$ as initial weights. Then the initial clusters  $C^{0}_{1}\;,C^{0}_{2}\;,...\;,C^{0}_{k}$ are constructed by the maximum likelihood principle which assigns a particular data point to the cluster which maximizes its likelihood. 
		(See subsequent Remark \ref{REM:initialization} for the effects of different initialization schemes on our algorithm).
		
		\item \textbf{Update:} Let, $C^{l}_{1}\;,C^{l}_{2}\;,...\;,C^{l}_{k}$ be the clusters at the $l$-th step of the algorithm 
		($l=0, 1, \ldots$). 
		\begin{enumerate}
			\item For each $1 \leq j \leq k$, obtain $n^{l}_{j}=|C^{l}_{j}|$,  $\hat{\pi}_{j}^{l}=\frac{n^{l}_{j}}{n}$
			(see Theorem \ref{THM2} for the justification).

			\item For each $1 \leq j \leq k$, given $n^{l}_{j}$, obtain $\hat{\boldsymbol{\mu}}_{j}^{l}$ and $\hat{\boldsymbol{\Sigma}}_{j}^{l}$ 
			by maximizing the $\beta$-likelihood of the observations which are currently assigned to the $j$-th cluster $C^{l}_{j}$. 
			Specifically,
			\begin{equation}
			\centering 
			\label{Eq7}
			(\hat{\boldsymbol{\mu}}_{j}^{l},\hat{\boldsymbol{\Sigma}}_{j}^{l})=\underset{\boldsymbol{\mu}_{j},\boldsymbol{\Sigma}_{j}}{\text{argmax}}\;\left [ \frac{1}{n^{l}_{j}\beta}\sum_{i \in C^{l}_{j}}\phi_{p}^{\beta}(\boldsymbol{X}_{i},\boldsymbol{\mu}_{j},\boldsymbol{\Sigma}_{j})-\frac{1}{1+\beta }\int \phi_{p}^{1+\beta}(\boldsymbol{x},\boldsymbol{\mu}_{j},\boldsymbol{\Sigma}_{j}) \;d\boldsymbol{x}\right].
			\end{equation}
			(by the above mentioned IRLS Algorithm \ref{algo1}.)
			\item If the full set of eigenvalues $\boldsymbol{\Lambda}=(\boldsymbol{\Lambda}_1,\ldots,\;\boldsymbol{\Lambda}_k)$ of the component dispersion matrix estimates does not satisfy the either of the ER or NS constraints, we replace $\boldsymbol{\Lambda}$ by another vector $\boldsymbol{\tilde{\Lambda}}$ which minimizes $||\boldsymbol{\tilde{\Lambda}}-\boldsymbol{\Lambda}||^{2}$ subject to the aforesaid constraints. That is, we need to obtain another vector \enquote{closest} to the existing set of eigenvalues which satisfies both the constraints. The ER constrain can be mathematically rephrased as $\lambda_{jl}-c\lambda_{su}\leq 0$ for all $(j,l)\neq(s,u)$. The NS constraint can mathematically be rephrased as $\lambda_{jl}\geq c_1$ for all $(j,l)$. Both of these constraints are linear constraints. Dykstra's algorithm (\cite{Dykstra}) can be used to solve the aforesaid constrained minimization problem.
			
			\item The estimates of the current step are then given by
			\begin{equation}
			\centering
			\label{Eq8}
			\hat{\boldsymbol{\theta}}^{l}=  (\hat{\pi}_{1}^{l}, \hat{\pi}_{2}^{l},...., \hat{\pi}_{k}^{l}, \hat{\boldsymbol{\mu}}_{1}^{l}, \hat{\boldsymbol{\mu}}_{2}^{l},...,\hat{\boldsymbol{\mu}}_{k}^{l}, \hat{\boldsymbol{\Sigma}}_{1}^{l}, \hat{\boldsymbol{\Sigma}}_{2}^{l},..., \hat{\boldsymbol{\Sigma}}_{k}^{l}).
			\end{equation}
			
			\item Construct the updated clusters $C^{l+1}_{1}\;,C^{l+1}_{2}\;,...\;,C^{l+1}_{k}$ as follows. 
			For each $1 \leq i \leq n$, assign $\boldsymbol{X}_{i}$ to $C^{l+1}_{j}$ if
			$
			D_{j}^{l}(\boldsymbol{X}_{i},\hat{\boldsymbol{\theta}}^{l})=	D^{l}(\boldsymbol{X}_{i},\hat{\boldsymbol{\theta}}^{l}),
			$
			where 
			$$D_{j}^{l}(\boldsymbol{X}_{i},\hat{\boldsymbol{\theta}}^{l})=\hat{\pi}_{1}^{l}\phi_{p}(\boldsymbol{X}_{i},\hat{\boldsymbol{\mu}}_{j}^{l},\hat{\boldsymbol{\Sigma}}_{j}^{l})
			~~\mbox{ and }~~
			D^{l}(\boldsymbol{X}_{i},\hat{\boldsymbol{\theta}}^{l})=\underset{1 \leq j \leq k}{\text{max}}D_{j}^{l}(\boldsymbol{X}_{i},\hat{\boldsymbol{\theta}}^{l}).
			$$
		\end{enumerate}
		
		\item \textbf{Stopping Rule:} Repeat step $2$ for a large (preassigned) number of times or until the cluster configurations become stable.
		
		\item \textbf{Outlier Detection:} After the process terminates and the final clusters $C_{1}, C_{2},...,C_{k}$ 
		and their configurations are available, perform the exercise in  Step 2 to obtain the estimate $
		\hat{\boldsymbol{\theta}}=  (\hat{\pi}_{1}, \hat{\pi}_{2},...., \hat{\pi}_{k}, \hat{\boldsymbol{\mu}}_{1}, \hat{\boldsymbol{\mu}}_{2},...,\\ \hat{\boldsymbol{\mu}}_{k}, \hat{\boldsymbol{\Sigma}}_{1}, \hat{\boldsymbol{\Sigma}}_{2},..., \hat{\boldsymbol{\Sigma}}_{k}).$ Now, for each $1 \leq i \leq n$, if $\boldsymbol{X}_{i}$ is assigned to $C_{j}$ for some $1 \leq j \leq k$, calculate $
		D_{j}(\boldsymbol{X}_{i},\hat{\boldsymbol{\theta}})= \hat{\pi}_{j}\phi_{p}(\boldsymbol{X}_{i},\hat{\boldsymbol{\mu}}_{j},\hat{\boldsymbol{\Sigma}}_{j}).$
		If $D_{j}(\boldsymbol{X}_{i},\hat{\boldsymbol{\theta}})\leq T $ for some small positive prespecified constant $T$, classify $\boldsymbol{X}_{i}$ as an outlier.
	\end{enumerate}
\end{algorithm}
\bigskip  
\begin{remark}\label{REM:initialization}
	To initialize the aforesaid algorithm, we have applied a random initialization scheme as stated in the initialization step of the algorithm. But randomly selected data points can produce very good as well as very bad estimators after completing the iterations. Hence, the algorithm should be repeated several times using different choices of initialization and then provide the solution which leads to the maximum value of the objective function. This initialization scheme has also been proposed in \cite{tclust}, although an improved version has been proposed in \cite{tclustpackage_2013}. Non-robust initial choices produce spurious maxima and the misclassification rates as well as bias and mean squared errors of the estimates increase drastically. This problem also arises in case of other robust clustering algorithms like the TCLUST and trimmed K-means. A good initialization is also essential to ensure that the empirical objective function monotonically increases to its maximum value after reasonable amount of iterations. 
\end{remark}     		

\begin{remark}
	For positive $\beta$, our method smoothly discounts the ill effects of anomalous observations without physically deleting them. On the other hand, TCLUST forcefully trims those anomalous observations along with others that may not be anomalous.
\end{remark}

\subsection{Selection of Tuning Parameters}
\label{sec2.3}
To implement the algorithm, we have to choose the tuning parameters $\beta$, $c$, $c_{1}$ and $T$ with appropriate justification. 
In general these choices are not straightforward; some comments are provided in the following remarks.

The tuning parameter $\beta$ in the DPD balances robustness and asymptotic efficiency of the resulting MDPDE; 
a small value (close to $0$) of $\beta$ is appropriate to achieve higher asymptotic efficiency  
whereas a large value is appropriate for higher stability (see \cite{basubook}, Section $9.6$, for further details). 	
This is also the case for our proposed estimators under the normal mixture models.
In fact, from the limit of our objective function as $\beta\rightarrow 0$ (discussed in Section \ref{SEC:theoretical_method}), 
it is evident that our DPD based algorithm leads to the usual (non-robust but most efficient) likelihood based estimation procedure. 
Hence, as $\beta$ approaches $0$, our algorithm becomes non-robust.
To achieve robustness as well as high asymptotic efficiency, we use a small but positive value of $\beta$ (in the interval $(0,0.5]$). However, sophisticated theoretical techniques for choosing an optimal value of $\beta$ have been described by \cite{warwick} and \cite{basak}. What these methods essentially do is that they create an empirical estimate of the true mean square estimator as a function of the tuning parameter (and a suitable pilot estimator) which can then be appropriately minimized over the tuning parameter to obtain an \enquote{optimal} estimate of the unknown tuning parameter. But, in our algorithm, parameter estimation has to be done separately for each cluster, in each iteration. Each of these cases would,
ideally, require different optimal choices of the parameter $\beta$ and thus a single optimal choice of $\beta$ is not reasonable to derive for the full clustering problem. Further mathematical details of these procedures can be found in Appendix \ref{appen5}.

\begin{remark}
	\label{rem28}
	As we have already observed, a large value of the tuning parameter $c$ makes the optimization problem almost unrestricted. For some rational choices of $c$, we refer again to \cite{tclust} and \cite{punzo2}. It is natural to choose the value of $c_{1}$ to be close to $0$ as the value of this constant is needed to be positive to serve certain theoretical purposes but a small value of the same is less restrictive. To find an approximate optimal choice of the tuning parameter $T$, we have followed a \enquote{maximal-gap} idea. Actually, we declare those observations as outliers whose estimated discriminant values (the realized value of the discriminant function) ($\hat{\pi}_{j}\phi_{p}(\cdot,\hat{\boldsymbol{\mu}}_{j},\hat{\boldsymbol{\Sigma}}_{j})$ if the observation belongs to the $j$-th cluster) drop below the threshold value of $T$. So, we first rearrange these discriminant values of all the $n$ observations in an increasing order. This rearrangement should bring the discriminant values corresponding to the outliers (if any) at the first few positions of the sequence. Since the distant outliers are expected to have very low discriminant values in comparison with those of the regular observations, the sorted vector of discriminant values should contain a (possibly big) jump between the region containing the outliers and the region containing the regular observations. We have chosen $T$ to correspond to this \enquote{jump} region, approximated to powers of $10$ (a pictorial illustration can be found in Appendix \ref{appen6}). This maximal-gap idea can also be implemented with the TCLUST methodology (and possibly to some other robust clustering methodologies) to improve outlier detection. It could be the subject of a future research to investigate how this type of refinement might improve the performances of these methods. However, although the maximal gap strategy has worked satisfactorily in practically all of our numerical studies, we do not expect that choosing $T$ by this philosophy will work perfectly in every possible situation. For example, when we have a not-so-remote background contamination together with an additional extremely remote outlying
		observation, only the remotest data point would be discarded by the maximal-gap strategy, without discarding the not-so-remote background noise. There might be a way to generalize the idea of maximal
		(largest) gap by choosing the value of $T$ around the position of the $t$-th largest gap ($t = 2$ will serve
		the purpose for this example) depending on the situation. But choosing the optimal value of $t$ in this refinement may be a very difficult problem in itself.
		On the whole, the choice of the tuning parameter $T$ is, at best, a complicated proposition and substantial future research will be needed for a completely satisfactory solution.
\end{remark}

Finally note that, although we have performed robust estimation in the multivariate normal mixture model, the primary focus of our proposed algorithm has been on robust clustering. We need to choose the number of clusters $k$ appropriately. 
In many of the well-known clustering techniques, performance improves with increasing the value of $k$, 
but increasing the value of $k$ indefinitely may be inappropriate. Hence an optimal choice of $k$ is needed. 
Rate distortion theory gives nice insights into the problem of detecting optimal number of clusters. 
It applies the \enquote{jump} method which detects $k$ by maximizing efficiency and minimizing error using information based measurements. 
We refer to \cite{sugarjames} for details. A novel penalized likelihood based method was also proposed by \cite{kc} to optimally select the pair $(k,c)$ ($c$ of ER constraint).

\section{Properties of the Proposed Algorithm}
\label{SEC:theoretical_result}
\subsection{Theoretical Results}
\label{TR}
\begin{theorem}[MDPD Estimators]\label{THM1}
	Suppose $\boldsymbol{X}_{1}, \boldsymbol{X}_{2},...,\boldsymbol{X}_{n}$ be a random sample drawn from an unknown PDF $g$ 
	which is modelled by a family of $p$-dimensional normal distributions with mean vector $\boldsymbol{\mu}$ and dispersion matrix $\boldsymbol{\Sigma}$. 
	Then the MDPDEs of $\boldsymbol{\mu}$ and $\boldsymbol{\Sigma}$ can be obtained by solving the system of equations in (\ref{eqq_irls}).

\end{theorem}
The detailed derivation of the aforesaid system of equations can be found in Appendix \ref{appen2}. 

Our next theorem provides the mathematical justification behind the update of the estimate $\hat{\pi}_{j}$
in Step 2(a) of our proposed clustering algorithm.

\begin{theorem}\label{THM2}
	Given a particular cluster assignment $C_{1},\;C_{2},...,\;C_{k}$ and the estimates $\hat{\boldsymbol{\mu}}_{j}$, $\hat{\boldsymbol{\Sigma}}_{j}$, the optimal value of $\pi_{j}$ that maximizes (\ref{Eq4}) is given by, $\hat{\pi}_{j}=\frac{n_{j}}{n}$ where $n_{j}=|C_{j}|$.
\end{theorem}
The proof is discussed in Appendix \ref{appen2}.
%

We have also studied further theoretical and asymptotic properties, such as, $(i)$ existence of a solution to the  optimization problems in Equations (\ref{Eq4}) (sample version) and (\ref{Eq5}) (population version) for the proposed procedures and $(ii)$ consistency of the resulting parameter estimates, as defined in Equation (\ref{Eq4}), which also yields the consistency of the estimated cluster centers and proportions. A detailed discussion (along with the mathematical proofs) can be found in \cite{arxiv}. 

\subsection{Robustness: Influence Function}
\label{SEC:IF}
To justify the robustness of our proposed estimators of cluster proportions, means and dispersion matrices, 
we will study the behaviour of their influence functions. 
Deriving these influence functions in higher dimensions is substantially difficult from the computational point of view. 
So, we will focus on the one dimensional case with two components ($p=1,\; k=2$); the implications will be in the same direction for higher dimensions. 
\cite{robustif1}, the only existing literature (as per the knowledge of the authors) for studying the robustness of TCLUST, also studies this special case only. We make the following assumption in order to circumvent cumbersome calculations.  
\begin{figure}[h!]
	\centering
	\begin{tabular}{cc}
		\begin{subfigure}{0.4\textwidth}\centering\includegraphics[height=2 in,width=\columnwidth]{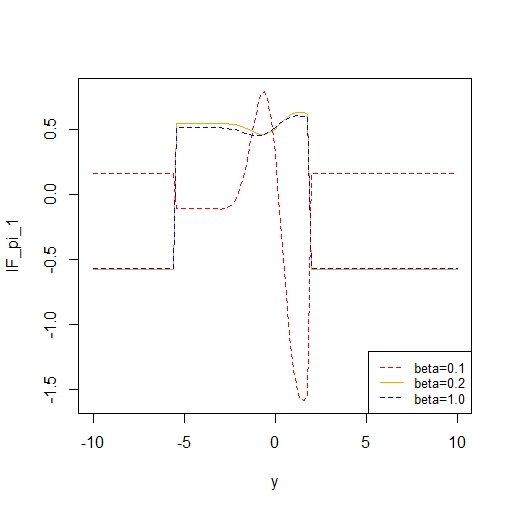}\caption{$IF(\pi_{1},P,y)$}\end{subfigure}&
		\begin{subfigure}{0.4\textwidth}\centering\includegraphics[height=2 in,width=\columnwidth]{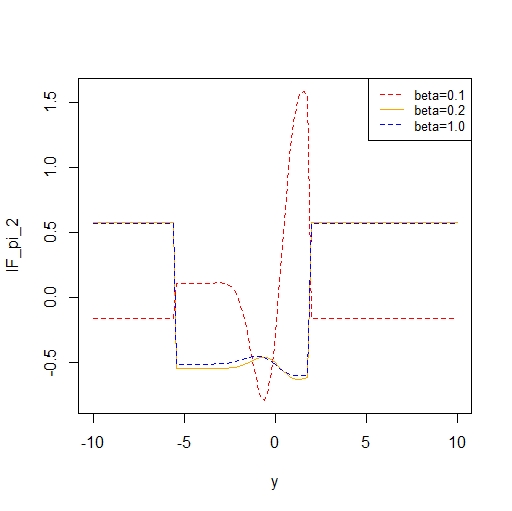}\caption{$IF(\pi_{2},P,y)$}\end{subfigure}\\
		\newline
		\begin{subfigure}{0.4\textwidth}\centering\includegraphics[height=2 in,width=\columnwidth]{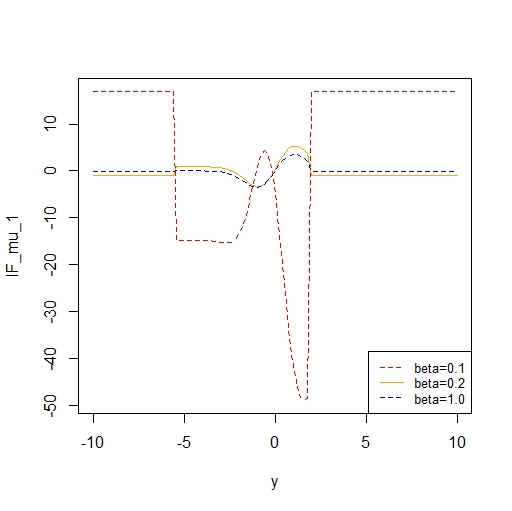}\caption{$IF(\mu_{1},P,y)$}\end{subfigure}&
		\begin{subfigure}{0.4\textwidth}\centering\includegraphics[height=2 in,width=\columnwidth]{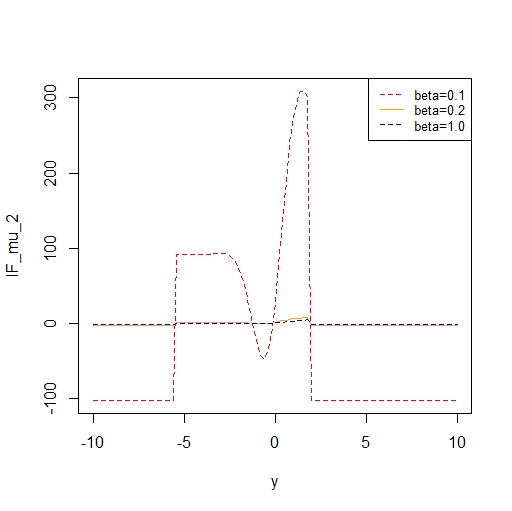}\caption{$IF(\mu_{2},P,y)$}\end{subfigure}\\
		\newline
		\begin{subfigure}{0.4\textwidth}\centering\includegraphics[height=2 in,width=\columnwidth]{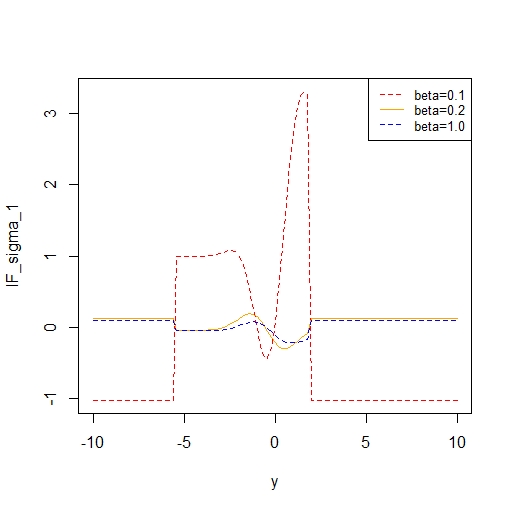}\caption{$IF(\sigma^{2}_{1},P,y)$}\end{subfigure}&
		\begin{subfigure}{0.4\textwidth}\centering\includegraphics[height=2 in,width=\columnwidth]{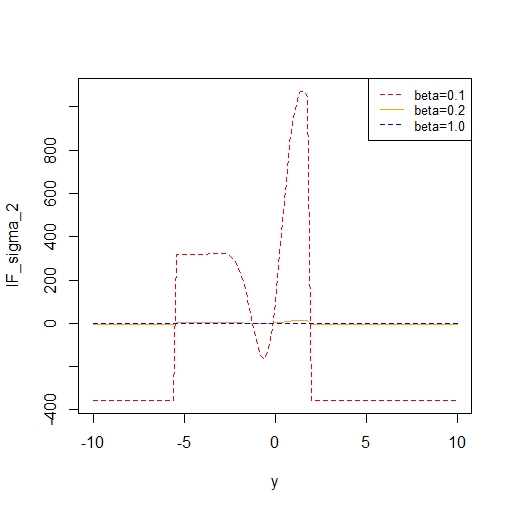}\caption{$IF(\sigma^{2}_{2},P,y)$}\end{subfigure}
	\end{tabular}
	\caption{Influence Functions of Different Functionals.}
	\label{table2}
\end{figure}

\begin{enumerate}
	\item[(\textbf{IF})] We assume that $\boldsymbol{\theta} \in \text{interior}(\boldsymbol{\Theta}_{C})$, that is, $\frac{M}{m}<c$ and $m>c_{1}$.
\end{enumerate}

To derive the influence function under our simple case ($p=1, k=2$), we first have to present our MPLF$_\beta$ (maximum pseudo $\beta$ likelihood functional) of the parameter $\boldsymbol{\theta}=(\pi_1, \pi_2, \mu_1, \mu_2, \sigma^2_1, \sigma^2_2)$ as a functional of the true probability distribution $P$ (rather CDF and $p(\cdot)$, the corresponding true probability density function), namely, 
\begin{align*}
\boldsymbol{\theta}_{\beta}(P)=(\pi_{1}(P),\pi_{2}(P),\mu_{1}(P),\mu_{2}(P),\sigma_{1}^{2}(P),\sigma_{2}^{2}(P)).
\end{align*}
This functional $\boldsymbol{\theta}_{\beta}(P)$ can be implicitly described through the following system of equations (see Appendix \ref{appen3.1}, for detailed the proof).
\begin{equation}
\centering
\label{Eq*}
\begin{array}{l}
\pi_{1}(P)=\int_{a(P)}^{b(P)}p(x)\;dx,\\
\pi_{1}(P)+\pi_{2}(P)=1,\\
D_{1}(c,\boldsymbol{\theta}(P))=D_{2}(c,\boldsymbol{\theta}(P))\; \text{for}\; c=a(P)\; \text{and}\; b(P),\\
\int_{a}^{b}f^{\beta}(x,\mu_{1},\sigma^{2}_{1})(x-\mu_{1})p(x)\;dx=0,\\
\int_{x \notin (a,b)}f^{\beta}(x,\mu_{2},\sigma^{2}_{2})(x-\mu_{2})p(x)\;dx=0,\\
\int_{a}^{b}f^{\beta}(x,\mu_{1},\sigma^{2}_{1})\left(\frac{(x-\mu_{1})^{2}}{2\sigma^{2}_{1}}-1\right)p(x)\;dx+\frac{\beta(P(b)-P(a))}{2(2\pi)^{\frac{\beta}{2}}(\sigma^{2}_{1})^{1+\frac{\beta}{2}}(1+\beta)^{\frac{3}{2}}}=0,\\
\int_{x \notin (a,b)}f^{\beta}(x,\mu_{2},\sigma^{2}_{2})\left(\frac{(x-\mu_{2})^{2}}{2\sigma^{2}_{2}}-1\right)p(x)\;dx+\frac{\beta(1-P(b)+P(a))}{2(2\pi)^{\frac{\beta}{2}}(\sigma^{2}_{2})^{1+\frac{\beta}{2}}(1+\beta)^{\frac{3}{2}}}=0,
\end{array}
\end{equation}
where $f(\cdot,\mu,\sigma^{2})$ is the PDF of univariate normal distribution with mean $\mu$ and variance $\sigma^{2}$.

The aforesaid system of equations will lead us to the influence functions of the necessary functionals through simple differentiation. The influence function $IF(T, P, y)$ of the functional $T$ at the point $y$, under the true distribution $P$, is obtained through a simple differentiation of the functional at a contaminated version of the distribution function $P$. Thus  
\begin{align*}
IF(T,P,y)=\underset{\epsilon \rightarrow 0}{\text{Lim}}\frac{T(P_{\epsilon})-T(P)}{\epsilon}=\frac{\partial T(P_{\epsilon})}{\partial \epsilon}\Big|_{\epsilon=0},
\end{align*} 
where $P_{\epsilon}=(1-\epsilon)P+\epsilon\wedge_{y}$ is the contaminated distribution where $\epsilon$ is the contamination proportion and $\wedge_{y}$ is the probability distribution degenerate at $y$. The assumption (\textbf{IF}) ensures that the eigenvalue ratio constraint and the non-singularity constraint also hold in case of the contaminated distribution for all $\epsilon$ close to zero. 

Let 
$ \boldsymbol{IF}(\boldsymbol{\theta}_\beta,P,y)=(IF(\pi_{1},P,y),IF(\pi_{2},P,y), IF(a,P,y), IF(b,P,y),
IF(\mu_{1},P,y), IF(\mu_{2},P,y),
\\
IF(\sigma^{2}_{1},P,y), IF(\sigma^{2}_{2},P,y))^{'}$ be the vector of influence functions of the aforesaid functionals in $\boldsymbol{\theta}_\beta$. Now differentiating the contaminated version of the system (\ref{Eq*}), we have the following system of linear equations,
\begin{equation}
\label{Eq**}
\centering
\boldsymbol{A}_{\beta}(\boldsymbol{\theta}_{0},a_{0},b_{0})\boldsymbol{IF}(\boldsymbol{\theta}_\beta,P,y)=\boldsymbol{B}_{\beta}(y,\boldsymbol{\theta}_{0},a_{0},b_{0}),
\end{equation}
where $\boldsymbol{\theta}_{0},a_{0},b_{0}$ are the true values of $\boldsymbol{\theta}, a$ and $b$ respectively, 
$\boldsymbol{A}_{\beta}(\boldsymbol{\theta}_{0},a_{0},b_{0})$ is a $8\times8$ coefficient matrix whose entries are independent of the contamination point $y$ and $\boldsymbol{B}_{\beta}(y,\boldsymbol{\theta}_{0},a_{0},b_{0})$ is an element in $\mathbb{R}^{8}$ 
which depends on $y$ only through $I(y \in (a_{0},b_{0}))$, $(y-\mu_{j0})exp(-\frac{\beta(y-\mu_{j0})^{2}}{2\sigma^{2}_{j0}})$ for $j=1,2$ and $p(y)$, the true density function corresponding to $P$.
Detailed formulas for $A_\beta$ and $B_\beta$ are given in Appendix \ref{appen3.2} along with the derivation of Equation (\ref{Eq**}).

The functions $I(y \in (a_{0},b_{0}))$, $(y-\mu_{j0})e^{-\frac{\beta(y-\mu_{j0})^{2}}{2\sigma^{2}_{j0}}}$ for $j=1,2$ are bounded while the PDF $p$ is not in general. This observation leads to the following result regarding the boundedness of the influence functions so that our estimators (and hence the clustering) are robust against outliers.
\begin{theorem}
	The influence function vector $\boldsymbol{IF}(P,y)$ exists if the coefficient matrix $\boldsymbol{A}_{\beta}(\boldsymbol{\theta}_{0},a_{0},b_{0})$ is invertible. Moreover, in case $\beta>0$, it is componentwise bounded as a function of $y$, if the true probability density function $p$ is bounded.
\end{theorem}
The proof is trivial from the above discussion and the form of $B_\beta$ given in the Appendices. Note that, the density $p$ is always bounded for a non-singular normal mixture model. 

Now, let us study the behaviour of the influence functions graphically in a special case. Let us take the true distribution $P$ as a mixture of $N(0,1)$ and $N(5,4)$ with mixing proportions $\pi_{1}=\pi_{2}=0.5$. We have taken $\beta=0.1,\; 0.2 \; \text{and}\; 1$ and the true values of the boundaries $a$ and $b$ are found to be $-5.5$ and $1.95$ respectively. We have taken $c$ and $c_{1}$ to be $5$ and $0.1$ so that the restrictions $\dfrac{M}{m}=4<5=c$ and $m=1>0.1=c_{1}$ are satisfied. The influence functions of the functionals are plotted in Figure \ref{table2}.

The boundedness of the curves in Figure \ref{table2} indicate the stability and the robustness of our estimators. 
Additionally, the respective ranges of each of the influence functions shrink drastically as $\beta$ increases. 
It may also be noted that the influence functions are practically identical in case of $\beta=0.2$ and $\beta=1$. 
This observation indicates that very strong levels of stability have been already attained, at least in this example, 
for very small values of $\beta$. In case of $\beta=0$, the influence functions are unbounded; 
in fact at $\beta = 0$ the curve increases so fast, that a proper depiction of this case together with the positive $\beta$ cases 
in the same frame is not informative. 

\section{Simulation Studies}
\label{secsim}
\subsection{Simulation Set-up}

We now present some simulation studies to investigate the finite sample performance of our algorithm 
in terms of the properties of the  obtained estimators and the subsequent clustering, and 
compare it with the TCLUST (\cite{tclust}), trimmed K-means (TKMEANS) (\cite{trimmedkmeans}), K-medoids (with data matrix and the Manhattan distance, abbreviated as KMEDOIDS)\\ (\cite{kmedoid}) and the MCLUST (\cite{mclust}) algorithms, some of the well-known and/or state-of-the-art robust clustering methods. 
To carry out the simulation study, we have generated samples of size $n=1000$ from  $k=3$-component 
and $p$-dimensional normal mixtures (stochastic mixtures with certain assignment probabilities to each of the clusters; since the clusters are simulated stochastically, the cluster weights are approximately equal to the assignment probabilities) with component means $\boldsymbol{\mu}_{1}=(0,0,...,0)^{t}$,  $\boldsymbol{\mu}_{2}=(5,5,...,5)^{t}$ and $\boldsymbol{\mu}_{3}=(-5,-5,...,-5)^{t}$ 
and identical covariance matrices $\boldsymbol{\Sigma}$. Different choices of $p$ and $\boldsymbol{\Sigma}$ are taken to cover a reasonable range of data shapes. 
To study the robustness and efficiency of our algorithm, both pure (contamination free) and contaminated datasets are used. 
Three types of data contamination are used for this purpose: $(i)$ uniform noise contamination from the $p$-dimensional cuboid $[-10,10]^{p}$, where only those data points whose Mahalanobis distance from any of the cluster centers are more than the $97.5$-th percentile of the $\chi^{2}(p)$ distribution are chosen; this type of contamination will be referred to as \enquote{uniform (chi-squared method) contamination} in the following and the Appendices; $(ii)$ uniform noise contamination from the $p$-dimensional annulus (centered at the origin with the inner and the outer radii $15$ and $20$, respectively), and $(iii)$ outlying cluster contamination with the outlying cluster center at $(20,20,...,20)^{t}$ and identity dispersion matrices 
with $10\%$ (approximately) contamination in each case (the contaminating part have assignment probability $0.1$). A similar motivational example is provided in Appendix \ref{appen1}. For the pure datasets, the cluster assignment probabilities are taken as $0.33$, $0.33$ and $0.34$ 
	and in case of contaminated datasets, the cluster assignment probabilities are $0.3$ for each of them and the rest of the observations (approximately $10\%$ of the sample) are outlying observations. 
As accuracy measures, we have focused on the estimated misclassification rates and the proportions of regular observations misclassified as outliers in case of the pure datasets; smaller values of these measures indicate greater accuracy. In each type of contaminated datasets the estimated misclassification rates of the regular observations (which are not outliers) and the estimated proportion of undetected outliers are the accuracy measures considered; again smaller values of these measures indicate greater accuracy. Estimated bias and mean squared errors of the estimated cluster means are also presented in the Appendix \ref{appen4}.
Datasets of five different dimensions, namely, $p=2,\;4,\;6,\;8\;\text{and}\;10$, have been generated 
with three choices of the common dispersion matrix $\boldsymbol{\Sigma}$, namely, $\boldsymbol{I}_{p},\;3\boldsymbol{I}_{p}\;\text{and}\;5\boldsymbol{I}_{p}$ ($\boldsymbol{I}_{p}$ is the $p\times p$ identity matrix). Together with the above, another simulation set-up with differentially dispersed clusters has also been considered where the component means are again $\boldsymbol{\mu}_{1}=(0,0,...,0)^{t}$,  $\boldsymbol{\mu}_{2}=(5,5,...,5)^{t}$ and $\boldsymbol{\mu}_{3}=(-5,-5,...,-5)^{t}$ with data dimensions $p=2,\;6$. But the component dispersion matrices are no longer identical and are taken to be $\boldsymbol{I}_{p},\;3\boldsymbol{I}_{p}$ and equicorrelation matrix of order $p$ with common correlation $\rho=0.5$, respectively, and the cluster assignment probabilities are taken as $0.30$, $0.35$, $0.35$  for pure datasets. For the contaminated datasets, the cluster assignment probabilities are $0.25$, $0.30$, $0.35$; the remaining $10\%$ (apprximately) observations are outliers generated by either of the three respective contamination schemes. The contamination schemes are the same as those defined above in $(i)$, $(ii)$ and $(iii)$. 
Tables \ref{table3.1}, \ref{table3.15}, \ref{table3.2}, \ref{table3.3} and \ref{table3.4} exhibit the estimated (mean) misclassification rates in case of pure datasets, uniformly (chi-squared method) contaminated datasets, 
uniformly (from annulus) contaminated datasets, outlying cluster contaminated datasets, and datasets with differentially dispersed clusters, respectively, based on $100$ replications. 

Further, different tuning parameters are chosen as follows in all our simulation experiments. In particular, we have taken $c=5$, $c_{1}=0.1$ and $T=10^{-3},\;10^{-5},\;10^{-8},\;10^{-18}\;\text{and}\;10^{-24}$ 
for $p=2,\;4,\;6,\;8\;\text{and}\;10$, respectively, in our proposed algorithm (as discussed in Remark \ref{rem28}). For the TCLUST method also we have used $c = 5$.
The trimming proportion $\alpha$ in TCLUST and trimmed K-means methods are taken as $0.0$ and $0.05$ for pure datasets and $0.10$ and $0.15$ for contaminated cases. The MCLUST method is used with a fixed number of clusters $(G=3)$ and uniform noise component (only in case of contaminated datasets) to make it resistant against outliers. Several values of $\beta$ are taken in the range $[0,0.5]$ for our method; 
values of $\beta$ larger than $0.5$ have been avoided in order to limit the loss in model efficiency. The R packages $\sf{tclust}$ (\cite{tclustpackage}), $\sf{trimcluster}$ (\cite{trimcluster}), $\sf{cluster}$ (\cite{clusterpack}) and $\sf{mclust}$ (\cite{mclust}) are used to carry out the simulations for the TCLUST, trimmed K-means, K-medoids and MCLUST algorithms, respectively.
\subsection{Discussion of Simulation Results }
The simulations that have been performed here are quite extensive, and it is necessary to clearly pinpoint what the salient features of these numbers are. In the following these features are described.
\begin{table}
	\hspace{-0.5in}
	\small
	\begin{tabular}{c c c c c c c c c c c c c} 
		\hline
		&       \multicolumn{4}{c}{\hspace{8em}MPLE$_{\beta}$} & & \multicolumn{2}{c}{TCLUST} & \multicolumn{2}{c}{TKMEANS} & KMEDOIDS  & MCLUST \\ [1ex] 
		$p$ &  $\boldsymbol{\Sigma}$  &  $\beta=0$  & $\beta=0.1$ & $\beta=0.3$ & $\beta=0.5$ & $\alpha=0.0$ &  $\alpha=0.05$ & $\alpha=0.0$  & $\alpha=0.05$ &  & \\ [1ex] 
		\hline
		2	&	$\boldsymbol{I}_{2}$	&			0.0003		&		0.0003		&		0.0004		&		0.0004		&		0.0003		&		0.050		&		0.0012		&		0.051		&		0.0003		& 0.0004\\
		&		&		$(	0.0001	)$	&	$(	0.0001	)$	&	$(	0.0002	)$	&	$(	0.0002	)$	&	$(	0.000	)$	&	$(	0.050	)$	&	$(	0.001	)$	&	$(	0.050	)$	&	$(	0.000	)$   & (0.000)\\
		&	$3\boldsymbol{I}_{2}$	&			0.029		&		0.029		&		0.029		&		0.029		&		0.030		&		0.076		&		0.028		&		0.074		&		0.029  & 0.029 	\\
		&		&		$(	0.0002	)$	&	$(	0.0002	)$	&	$(	0.0003	)$	&	$(	0.0004	)$	&	$(	0.000	)$	&	$(	0.050	)$	&	$(	0.001	)$	&	$(	0.050	)$	&	$(	0.000	)$ & (0.000)	\\
		&	$5\boldsymbol{I}_{2}$	&			0.082		&		0.082		&		0.082		&		0.082		&		0.093		&		0.129		&		0.077		&		0.122		&		0.079	 & 0.078	\\
		&		&		$(	0.001	)$	&	$(	0.001	)$	&	$(	0.001	)$	&	$(	0.001	)$	&	$(	0.000	)$	&	$(	0.050	)$	&	$(	0.001	)$	&	$(	0.050	)$	&	$(	0.000	)$  & (0.000)	\\
		&	&	&	&		\\
		4	&	$\boldsymbol{I}_{4}$	&			0.000		&		0.000		&		0.0001		&		0.0002		&		0.000		&		0.050		&		0.001		&		0.050		&		0.000		& 0.0004\\
		&		&		$(	0.001	)$	&	$(	0.001	)$	&	$(	0.0001	)$	&	$(	0.0002	)$	&	$(	0.000	)$	&	$(	0.050	)$	&	$(	0.0001	)$	&	$(	0.050	)$	&	$(	0.000	)$ & (0.000)	\\
		&	$3\boldsymbol{I}_{4}$	&			0.003		&		0.004		&		0.004		&		0.004		&		0.003		&		0.051		&		0.0034		&		0.051		&		0.0034		 & 0.003\\
		&		&		$(	0.001	)$	&	$(	0.001	)$	&	$(	0.001	)$	&	$(	0.002	)$	&	$(	0.000	)$	&	$(	0.050	)$	&	$(	0.001	)$	&	$(	0.050	)$	&	$(	0.000	)$ & (0.000)	\\
		&	$5\boldsymbol{I}_{4}$	&			0.020		&		0.019		&		0.021		&		0.022		&		0.019		&		0.066		&		0.017		&		0.064		&		0.029	& 0.019	\\
		&		&		$(	0.002	)$	&	$(	0.002	)$	&	$(	0.003	)$	&	$(	0.003	)$	&	$(	0.000	)$	&	$(	0.050	)$	&	$(	0.001	)$	&	$(	0.050	)$	&	$(	0.000	)$ & (0.000)	\\
		&	&	&	&		\\
		6	&	$\boldsymbol{I}_{6}$	&			0.000		&		0.000		&		0.000		&		0.000		&		0.000		&		0.050		&		0.001		&		0.050		&		0.000		& 0.000\\
		&		&		$(	0.000	)$	&	$(	0.000	)$	&	$(	0.000	)$	&	$(	0.000	)$	&	$(	0.000	)$	&	$(	0.050	)$	&	$(	0.001	)$	&	$(	0.050	)$	&	$(	0.000	)$  & (0.000)	\\
		&	$3\boldsymbol{I}_{6}$	&			0.000		&		0.000		&		0.000		&		0.000		&		0.0004		&		0.051		&		0.001		&		0.050		&		0.000	& 0.0002\\
		&		&		$(	0.000	)$	&	$(	0.0001	)$	&	$(	0.0003	)$	&	$(	0.0004	)$	&	$(	0.000	)$	&	$(	0.050	)$	&	$(	0.001	)$	&	$(	0.050	)$	&	$(	0.000	)$ &  (0.000)	\\
		&	$5\boldsymbol{I}_{6}$	&			0.005		&		0.005		&		0.006		&		0.006		&		0.004		&		0.053		&		0.005		&		0.053		&		0.007	 & 0.004	\\
		&		&		$(	0.000	)$	&	$(	0.001	)$	&	$(	0.001	)$	&	$(	0.001	)$	&	$(	0.000	)$	&	$(	0.050	)$	&	$(	0.001	)$	&	$(	0.050	)$	&	$(	0.000	)$ & (0.000)	\\
		&	&	&	&		\\
		8	&	$\boldsymbol{I}_{8}$	&			0.000		&		0.000		&		0.000		&		0.000		&		0.000		&		0.050		&		0.001		&		0.050		&		0.000		& 0.000\\
		&		&		$(	0.000	)$	&	$(	0.000	)$	&	$(	0.000	)$	&	$(	0.000	)$	&	$(	0.000	)$	&	$(	0.050	)$	&	$(	0.001	)$	&	$(	0.050	)$	&	$(	0.000	)$ & (0.000)	\\
		&	$3\boldsymbol{I}_{8}$	&			0.000		&		0.000		&		0.000		&		0.000		&		0.000		&		0.052		&		0.001		&		0.05		&		0.000	& 0.000	\\
		&		&		$(	0.000	)$	&	$(	0.000	)$	&	$(	0.000	)$	&	$(	0.000	)$	&	$(	0.000	)$	&	$(	0.050	)$	&	$(	0.001	)$	&	$(	0.050	)$	&	$(	0.000	)$  & (0.000)	\\
		&	$5\boldsymbol{I}_{8}$	&			0.001		&		0.001		&		0.001		&		0.001		&		0.001		&		0.051		&		0.002		&		0.051		&		0.002		& 0.002\\
		&		&		$(	0.000	)$	&	$(	0.000	)$	&	$(	0.000	)$	&	$(	0.000	)$	&	$(	0.000	)$	&	$(	0.050	)$	&	$(	0.001	)$	&	$(	0.050	)$	&	$(	0.000	)$  & (0.000)	\\
		&	&	&	&		\\
		10	&	$\boldsymbol{I}_{10}$	&			0.000		&		0.000		&		0.000		&		0.000		&		0.000		&		0.050		&		0.001		&		0.050		&		0.000	 & 0.000		\\
		&		&		$(	0.000	)$	&	$(	0.000	)$	&	$(	0.000	)$	&	$(	0.000	)$	&	$(	0.000	)$	&	$(	0.050	)$	&	$(	0.001	)$	&	$(	0.050	)$	&	$(	0.000	)$  & (0.000)	\\
		&	$3\boldsymbol{I}_{10}$	&			0.000		&		0.000		&		0.000		&		0.000		&		0.000		&		0.050		&		0.001		&		0.050		&		0.000		& 0.000\\
		&		&		$(	0.000	)$	&	$(	0.000	)$	&	$(	0.000	)$	&	$(	0.000	)$	&	$(	0.000	)$	&	$(	0.050	)$	&	$(	0.001	)$	&	$(	0.050	)$	&	$(	0.000	)$   & (0.000)	\\
		&	$5\boldsymbol{I}_{10}$	&			0.0002		&		0.0002		&		0.0002		&		0.0003		&		0.0003		&		0.051		&		0.001		&		0.051		&		0.001		& 0.000\\
		& & $(	0.000	)$	&	$(	0.000	)$	&	$(	0.000	)$	&	$(	0.000	)$	&	$(	0.000	)$	&	$(	0.050	)$	&	$(	0.001	)$	&	$(	0.050	)$	&	$(	0.000	)$  & (0.000) 	\\[2ex] 
		\hline
	\end{tabular}
	\caption{Estimated misclassification rates of regular observations (and proportions of regular observations misclassified as outliers within parentheses) for pure datasets.}
	\label{table3.1}
\end{table}
\begin{table}[!t]
	\hspace{-0.5in}
	\small
	\begin{tabular}{c c c c c c c c c c c c} 
		\hline
		&       \multicolumn{4}{c}{\hspace{8em}MPLE$_{\beta}$} & & \multicolumn{2}{c}{TCLUST} & \multicolumn{2}{c}{TKMEANS} & KMEDOIDS  & MCLUST \\ [1ex] 
		$p$ &  $\boldsymbol{\Sigma}$  &  $\beta=0$  & $\beta=0.1$ & $\beta=0.3$ & $\beta=0.5$ & $\alpha=0.1$ &  $\alpha=0.15$ & $\alpha=0.1$  & $\alpha=0.15$ & \\ [1ex] 
		\hline
		2	&	$\boldsymbol{I}_{2}$	&		0.018		&		0.013		&		0.029		&		0.031		&		0.009		&		0.056		&		0.013		&		0.055		&		0.002		&		0.011		\\
		&		&	$(	0.123	)$	&	$(	0.081	)$	&	$(	0.007	)$	&	$(	0.005	)$	&	$(	0.085	)$	&	$(	0.000	)$	&	$(	0.080	)$	&	$(	0.000	)$	&	$(	1.000	)$	&	$(	0.068	)$	\\
		&	$3\boldsymbol{I}_{2}$	&		0.061		&		0.053		&		0.056		&		0.061		&		0.040		&		0.078		&		0.036		&		0.077		&		0.028		&		0.053		\\
		&		&	$(	0.184	)$	&	$(	0.158	)$	&	$(	0.057	)$	&	$(	0.011	)$	&	$(	0.096	)$	&	$(	0.0002	)$	&	$(	0.096	)$	&	$(	0.000	)$	&	$(	1.000	)$	&	$(	0.002	)$	\\
		&	$5\boldsymbol{I}_{2}$	&		0.135		&		0.127		&		0.129		&		0.134		&		0.108		&		0.133		&		0.085		&		0.127		&		0.084		&		0.117		\\
		&		&	$(	0.213	)$	&	$(	0.198	)$	&	$(	0.146	)$	&	$(	0.060	)$	&	$(	0.107	)$	&	$(	0.001	)$	&	$(	0.085	)$	&	$(	0.000	)$	&	$(	1.000	)$	&	$(0.000	)$	\\
		&	&	&				&				&							\\
		4	&	$\boldsymbol{I}_{4}$	&		0.008		&		0.007		&		0.009		&		0.010		&		0.007		&		0.057		&		0.008		&		0.056		&		0.000		&		0.001		\\
		&		&	$(	0.060	)$	&	$(	0.008	)$	&	$(	0.006	)$	&	$(	0.006	)$	&	$(	0.041	)$	&	$(	0.000	)$	&	$(	0.036	)$	&	$(	0.000	)$	&	$(	1.000	)$	&	$(	0.022	)$	\\
		&	$3\boldsymbol{I}_{4}$	&		0.009		&		0.007		&		0.009		&		0.010		&		0.010		&		0.058		&		0.009		&		0.056		&		0.003		&		0.008		\\
		&		&	$(	0.271	)$	&	$(	0.150	)$	&	$(	0.064	)$	&	$(	0.060	)$	&	$(	0.072	)$	&	$(	0.000	)$	&	$(	0.067	)$	&	$(	0.000	)$	&	$(	1.000	)$	&	$(	0.063	)$	\\
		&	$5\boldsymbol{I}_{4}$	&		0.029		&		0.028		&		0.033		&		0.036		&		0.028		&		0.073		&		0.028		&		0.071		&		0.022		&		0.028		\\
		&		&	$(	0.286	)$	&	$(	0.190	)$	&	$(	0.072	)$	&	$(	0.053	)$	&	$(	0.091	)$	&	$(	0.0002	)$	&	$(	0.079	)$	&	$(	0.000	)$	&	$(	1.000	)$	&	$(	0.061	)$	\\
		&	&	&				&				&							\\
		6	&	$\boldsymbol{I}_{6}$	&		0.003		&		0.0005		&		0.0006		&		0.0008		&		0.006		&		0.055		&		0.005		&		0.056		&		0.000		&		0.0004		\\
		&		&	$(	0.061	)$	&	$(	0.004	)$	&	$(	0.004	)$	&	$(	0.004	)$	&	$(	0.034	)$	&	$(	0.000	)$	&	$(	0.032	)$	&	$(	0.000	)$	&	$(	1.000	)$	&	$(	0.004	)$	\\
		&	$3\boldsymbol{I}_{6}$	&		0.008		&		0.006		&		0.010		&		0.011		&		0.005		&		0.053		&		0.004		&		0.053		&		0.001		&		0.002		\\
		&		&	$(	0.079	)$	&	$(	0.026	)$	&	$(	0.016	)$	&	$(	0.014	)$	&	$(	0.051	)$	&	$(	0.000	)$	&	$(	0.053	)$	&	$(	0.000	)$	&	$(	1.000	)$	&	$(	0.003	)$	\\
		&	$5\boldsymbol{I}_{6}$	&		0.010		&		0.007		&		0.009		&		0.011		&		0.011		&		0.060		&		0.009		&		0.056		&		0.007		&		0.008		\\
		&		&	$(	0.093	)$	&	$(	0.017	)$	&	$(	0.083	)$	&	$(	0.077	)$	&	$(	0.073	)$	&	$(	0.001	)$	&	$(	0.067	)$	&	$(	0.000	)$	&	$(	1.000	)$	&	$(	0.074	)$	\\
		&	&	&				&				&							\\
		8	&	$\boldsymbol{I}_{8}$	&		0.002		&		0.000		&		0.000		&		0.000		&		0.005		&		0.056		&		0.004		&		0.055		&		0.000		&		0.000		\\
		&		&	$(	0.672	)$	&	$(	0.007	)$	&	$(	0.008	)$	&	$(	0.008	)$	&	$(	0.022	)$	&	$(	0.000	)$	&	$(	0.035	)$	&	$(	0.000	)$	&	$(	1.000	)$	&	$(	0.000	)$	\\
		&	$3\boldsymbol{I}_{8}$	&		0.003		&		0.000		&		0.000		&		0.000		&		0.005		&		0.058		&		0.005		&		0.057		&		0.001		&		0.001		\\
		&		&	$(	0.711	)$	&	$(	0.000	)$	&	$(	0.000	)$	&	$(	0.058	)$	&	$(	0.041	)$	&	$(	0.000	)$	&	$(	0.030	)$	&	$(	0.000	)$	&	$(	1.000	)$	&	$(	0.013	)$	\\
		&	$5\boldsymbol{I}_{8}$	&		0.009		&		0.007		&		0.011		&		0.014		&		0.007		&		0.055		&		0.007		&		0.057		&		0.003		&		0.003		\\
		&		&	$(	0.839	)$	&	$(	0.048	)$	&	$(	0.025	)$	&	$(	0.022	)$	&	$(	0.052	)$	&	$(	0.000	)$	&	$(	0.049	)$	&	$(	0.000	)$	&	$(	1.000	)$	&	$(	0.048	)$	\\
		&	&	&				&				&				&							\\
		10	&	$\boldsymbol{I}_{10}$	&		0.005		&		0.000		&		0.000		&		0.000		&		0.004		&		0.055		&		0.003		&		0.054		&		0.000		&		0.000		\\
		&		&	$(	0.479	)$	&	$(	0.004	)$	&	$(	0.003	)$	&	$(	0.003	)$	&	$(	0.048	)$	&	$(	0.000	)$	&	$(	0.037	)$	&	$(	0.000	)$	&	$(	1.000	)$	&	$(	0.001	)$	\\
		&	$3\boldsymbol{I}_{10}$	&		0.004		&		0.000		&		0.000		&		0.000		&		0.004		&		0.055		&		0.004		&		0.055		&		0.000		&		0.000		\\
		&		&	$(	0.612	)$	&	$(	0.022	)$	&	$(	0.019	)$	&	$(	0.019	)$	&	$(	0.044	)$	&	$(	0.000	)$	&	$(	0.037	)$	&	$(	0.000	)$	&	$(	1.000	)$	&	$(	0.005	)$	\\
		&	$5\boldsymbol{I}_{10}$	&		0.007		&		0.002		&		0.004		&		0.006		&		0.005		&		0.054		&		0.004		&		0.055		&		0.001		&		0.002		\\
		& &	$(	0.997	)$	&	$(	0.032	)$	&	$(	0.018	)$	&	$(	0.016	)$	&	$(	0.043	)$	&	$(	0.000	)$	&	$(	0.039	)$	&	$(	0.000	)$	&	$(	1.000	)$	&	$(	0.025	)$	\\
		
		\\[2ex] 
		\hline
	\end{tabular}
	\caption{Estimated misclassification rates of regular observations (and proportion of undetected outliers within parentheses) for uniformly (chi-squared method) contaminated datasets.}
	\label{table3.15}
\end{table}
\begin{table}[!t]
	\hspace{-0.5in}
	\small
	\begin{tabular}{c c c c c c c c c c c c} 
		\hline
		&       \multicolumn{4}{c}{\hspace{8em}MPLE$_{\beta}$} & & \multicolumn{2}{c}{TCLUST} & \multicolumn{2}{c}{TKMEANS} & KMEDOIDS  & MCLUST \\ [1ex] 
		$p$ &  $\boldsymbol{\Sigma}$  &  $\beta=0$  & $\beta=0.1$ & $\beta=0.3$ & $\beta=0.5$ & $\alpha=0.1$ &  $\alpha=0.15$ & $\alpha=0.1$  & $\alpha=0.15$ & \\ [1ex] 
		\hline
		2	&	$\boldsymbol{I}_{2}$	&		0.004		&		0.001		&		0.001		&		0.001		&		0.004		&		0.057		&		0.006		&		0.058		&		0.001	&   0.002 	\\
		&		&	$(	0.497	)$	&	$(	0.000	)$	&	$(	0.000	)$	&	$(	0.000	)$	&	$(	0.029	)$	&	$(	0.000	)$	&	$(	0.025	)$	&	$(	0.000	)$	&	$(	1.000	)$  & (0.000)	\\
		&	$3\boldsymbol{I}_{2}$	&		0.038		&		0.040		&		0.029		&		0.029		&		0.053		&		0.079		&		0.032		&		0.079		&		0.029	 &   0.032	\\
		&		&	$(	0.632	)$	&	$(	0.280	)$	&	$(	0.000	)$	&	$(	0.000	)$	&	$(	0.034	)$	&	$(	0.000	)$	&	$(	0.039	)$	&	$(	0.000	)$	&	$(	1.000	)$  & (0.000)	\\
		&	$5\boldsymbol{I}_{2}$	&		0.092		&		0.095		&		0.096		&		0.093		&		0.146		&		0.132		&		0.081		&		0.126		&		0.080	&   0.085 	\\
		&		&	$(	0.698	)$	&	$(	0.396	)$	&	$(	0.038	)$	&	$(	0.032	)$	&	$(	0.039	)$	&	$(	0.000	)$	&	$(	0.041	)$	&	$(	0.000	)$	&	$(	1.000	)$ & (0.000)	\\
		&	&	&								\\
		4	&	$\boldsymbol{I}_{4}$	&		0.001		&		0.0002		&		0.0002		&		0.0002		&		0.004		&		0.053		&		0.003		&		0.053		&		0.0003	&   0.0002 	\\
		&		&	$(	0.071	)$	&	$(	0.000	)$	&	$(	0.000	)$	&	$(	0.000	)$	&	$(	0.036	)$	&	$(	0.000	)$	&	$(	0.043	)$	&	$(	0.000	)$	&	$(	1.000	)$ & (0.000)	\\
		&	$3\boldsymbol{I}_{4}$	&		0.010		&		0.003		&		0.003		&		0.004		&		0.007		&		0.055		&		0.007		&		0.058		&		0.003	&   0.004  	\\
		&		&	$(	0.105	)$	&	$(	0.014	)$	&	$(	0.009	)$	&	$(	0.008	)$	&	$(	0.044	)$	&	$(	0.000	)$	&	$(	0.028	)$	&	$(	0.000	)$	&	$(	1.000	)$ & (0.007)	\\
		&	$5\boldsymbol{I}_{4}$	&		0.026		&		0.021		&		0.021		&		0.021		&		0.023		&		0.072		&		0.022		&		0.070		&		0.024	&   0.021 	\\
		&		&	$(	0.147	)$	&	$(	0.065	)$	&	$(	0.038	)$	&	$(	0.037	)$	&	$(	0.046	)$	&	$(	0.004	)$	&	$(	0.045	)$	&	$(	0.002	)$	&	$(	1.000	)$  & (0.032)	\\
		&	&	&							\\
		6	&	$\boldsymbol{I}_{6}$	&		0.0001		&		0.000		&		0.000		&		0.000		&		0.006		&		0.054		&		0.005		&		0.060		&		0.002	&   0.000	\\
		&		&	$(	0.147	)$	&	$(	0.000	)$	&	$(	0.000	)$	&	$(	0.000	)$	&	$(	0.031	)$	&	$(	0.000	)$	&	$(	0.033	)$	&	$(	0.000	)$	&	$(	1.000	)$  & (0.001)	\\
		&	$3\boldsymbol{I}_{6}$	&		0.002		&		0.0005		&		0.0008		&		0.0008		&		0.004		&		0.055		&		0.005		&		0.055		&		0.0006	&   0.001 	\\
		&		&	$(	0.239	)$	&	$(	0.012	)$	&	$(	0.009	)$	&	$(	0.009	)$	&	$(	0.038	)$	&	$(	0.001	)$	&	$(	0.038	)$	&	$(	0.001	)$	&	$(	1.000	)$ &  (0.010)	\\
		&	$5\boldsymbol{I}_{6}$	&		0.007		&		0.004		&		0.005		&		0.005		&		0.010		&		0.059		&		0.010		&		0.060		&		0.008	&   0.006 	\\
		&		&	$(	0.329	)$	&	$(	0.046	)$	&	$(	0.040	)$	&	$(	0.039	)$	&	$(	0.045	)$	&	$(	0.004	)$	&	$(	0.044	)$	&	$(	0.004	)$	&	$(	1.000	)$&  (0.032)  	\\
		&	&	&						\\
		8	&	$\boldsymbol{I}_{8}$	&		0.0001		&		0.000		&		0.000		&		0.000		&		0.005		&		0.056		&		0.003		&		0.053		&		0.003	& 0.000	\\
		&		&	$(	0.893	)$	&	$(	0.003	)$	&	$(	0.003	)$	&	$(	0.002	)$	&	$(	0.021	)$	&	$(	0.000	)$	&	$(	0.043	)$	&	$(	0.000	)$	&	$(	1.000	)$ & (0.000)	\\
		&	$3\boldsymbol{I}_{8}$	&		0.0003		&		0.000		&		0.000		&		0.000		&		0.006		&		0.057		&		0.005		&		0.056		&		0.0002	& 0.001	\\
		&		&	$(	0.925	)$	&	$(	0.014	)$	&	$(	0.012	)$	&	$(	0.012	)$	&	$(	0.031	)$	&	$(	0.000	)$	&	$(	0.035	)$	&	$(	0.000	)$	&	$(	1.000	)$ & (0.005)	\\
		&	$5\boldsymbol{I}_{8}$	&		0.004		&		0.002		&		0.003		&		0.005		&		0.008		&		0.057		&		0.006		&		0.057		&		0.003	& 0.003	\\
		&		&	$(	0.997	)$	&	$(	0.020	)$	&	$(	0.017	)$	&	$(	0.015	)$	&	$(	0.032	)$	&	$(	0.002	)$	&	$(	0.038	)$	&	$(	0.002	)$	&	$(	1.000	)$ &  (0.023) 	\\
		&	&	&						\\
		10	&	$\boldsymbol{I}_{10}$	&		0.000		&		0.000		&		0.000		&		0.000		&		0.004		&		0.056		&		0.004		&		0.056		&		0.000	& 0.000	\\
		&		&	$(	1.000	)$	&	$(	0.0002	)$	&	$(	0.0002	)$	&	$(	0.0002	)$	&	$(	0.031	)$	&	$(	0.000	)$	&	$(	0.026	)$	&	$(	0.000	)$	&	$(	1.000	)$ & (0.000)	\\
		&	$3\boldsymbol{I}_{10}$	&		0.000		&		0.000		&		0.000		&		0.000		&		0.005		&		0.057		&		0.005		&		0.057		&		0.000	& 0.001	\\
		&		&	$(	1.000	)$	&	$(	0.025	)$	&	$(	0.022	)$	&	$(	0.024	)$	&	$(	0.029	)$	&	$(	0.000	)$	&	$(	0.031	)$	&	$(	0.000	)$	&	$(	1.000	)$ & (0.002)	\\
		&	$5\boldsymbol{I}_{10}$	&		0.0005		&		0.0003		&		0.0004		&		0.0008		&		0.005		&		0.055		&		0.004		&		0.054		&		0.001	& 0.001	\\
		& &	$(	1.000	)$	&	$(	0.060	)$	&	$(	0.026	)$	&	$(	0.025	)$	&	$(	0.038	)$	&	$(	0.002	)$	&	$(	0.044	)$	&	$(	0.001	)$	&	$(	1.000	)$  & (0.015)	\\[2ex] 
		\hline
	\end{tabular}
	\caption{Estimated misclassification rates of regular observations (and proportion of undetected outliers within parentheses) for uniformly (from annulus) contaminated datasets.}
	\label{table3.2}
\end{table}
\begin{table}[!t]
	\hspace{-0.5in}
	\small
	\begin{tabular}{c c c c c c c c c c c c} 
		\hline
		&       \multicolumn{4}{c}{\hspace{8em}MPLE$_{\beta}$} & & \multicolumn{2}{c}{TCLUST} & \multicolumn{2}{c}{TKMEANS} & KMEDOIDS & MCLUST \\ [1ex] 
		$p$ &  $\boldsymbol{\Sigma}$  &  $\beta=0$  & $\beta=0.1$ & $\beta=0.3$ & $\beta=0.5$ & $\alpha=0.1$ &  $\alpha=0.15$ & $\alpha=0.1$  & $\alpha=0.15$ & \\ [1ex] 
		\hline
		2	&	$\boldsymbol{I}_{2}$	&		0.042		&		0.019		&		0.019		&		0.019		&		0.041		&		0.056		&		0.030		&		0.056		&		0.413	&     0.563	\\
		&		&	$(	0.058	)$	&	$(	0.000	)$	&	$(	0.000	)$	&	$(	0.000	)$	&	$(	0.099	)$	&	$(	0.000	)$	&	$(	0.080	)$	&	$(	0.000	)$	&	$(	1.000	)$ & (0.990)	\\
		&	$3\boldsymbol{I}_{2}$	&		0.108		&		0.107		&		0.068		&		0.067		&		0.213   &   0.088		&		0.106		&		0.080		&		0.486	&     0.471 	\\
		&		&	$(	0.238	)$	&	$(	0.209	)$	&	$(	0.000	)$	&	$(	0.000	)$	&	(0.470)  &  (0.009)	&	$(	0.151	)$	&	$(	0.000	)$	&	$(	1.000	)$	& (0.999)\\
		&	$5\boldsymbol{I}_{2}$	&		0.251		&		0.260		&		0.169		&		0.156		&		0.414   &   0.248 		&		0.268		&		0.126		&		0.505	&     0.510	\\
		&		&	$(	0.354	)$	&	$(	0.401	)$	&	$(	0.068	)$	&	$(	0.000	)$	& (0.992)  &  (0.492)	&	$(	0.405	)$	&	$(	0.000	)$	&	$(	1.000	)$ & (0.995)	\\
		&	&	&				&					\\
		4	&	$\boldsymbol{I}_{4}$	&		0.218		&		0.008		&		0.009		&		0.010		&		0.033		&		0.055		&		0.062		&		0.055		&		0.376	& 0.541	\\
		&		&	$(	0.849	)$	&	$(	0.000	)$	&	$(	0.000	)$	&	$(	0.000	)$	&	$(	0.095	)$	&	$(	0.000	)$	&	$(	0.124	)$	&	$(	0.000	)$	&	$(	1.000	)$ & (0.999)	\\
		&	$3\boldsymbol{I}_{4}$	&		0.113		&		0.111		&		0.053		&		0.057		&		0.219   &   0.057		&		0.094		&		0.054		&		0.414	& 0.459   	\\
		&		&	$(	0.930	)$	&	$(	0.869	)$	&	$(	0.000	)$	&	$(	0.000	)$	&	(0.560)  &  (0.000)	&	$(	0.204	)$	&	$(	0.000	)$	&	$(	1.000	)$ & (0.999)	\\
		&	$5\boldsymbol{I}_{4}$	&		0.218		&		0.222		&		0.134		&		0.135		&		0.368   &   0.320		&		0.153		&		0.071		&		0.464	 &     0.490 	\\
		&		&	$(	0.857	)$	&	$(	0.900	)$	&	$(	0.000	)$	&	$(	0.000	)$	&(0.920)  &  (0.480)	&	$(	0.226	)$	&	$(	0.000	)$	&	$(	1.000	)$  &     (0.999)	\\
		&	&	&				&						\\
		6	&	$\boldsymbol{I}_{6}$	&		0.339		&		0.001		&		0.001		&		0.001		&		0.050   &   0.057 		&		0.021		&		0.058		&		0.369	& 0.516	\\
		&		&	$(	0.990	)$	&	$(	0.000	)$	&	$(	0.000	)$	&	$(	0.000	)$	&	$(	0.154	)$	&	$(	0.000	)$	&	$(	0.037	)$	&	$(	0.000	)$	&	$(	1.000	)$ & (0.992)	\\
		&	$3\boldsymbol{I}_{6}$	&		0.309		&		0.027		&		0.011		&		0.013		&		0.165   &   0.056		&		0.034		&		0.054		&		0.441	& 0.518	\\
		&		&	$(	1.000	)$	&	$(	0.720	)$	&	$(	0.000	)$	&	$(	0.000	)$	&	$(	0.421	)$	&	$(	0.000	)$	&	$(	0.058	)$	&	$(	0.000	)$	&	$(	1.000	)$ & (1.000)	\\
		&	$5\boldsymbol{I}_{6}$	&		0.290		&		0.076		&		0.035		&		0.036		&		0.311   &   0.192		&		0.109		&		0.058		&		0.450	& 0.533	\\
		&		&	$(	1.000	)$	&	$(	1.000	)$	&	$(	0.020	)$	&	$(	0.000	)$	&	(0.800)  &  (0.360) 	&	$(	0.173	)$	&	$(	0.000	)$	&	$(	1.000	)$ & (1.000)	\\
		&	&	&				&				&				&				&							\\
		8	&	$\boldsymbol{I}_{8}$	&		0.337		&		0.000		&		0.000		&		0.000		&		0.070   &   0.053 		&		0.047		&		0.053		&		0.382	& 0.479	\\
		&		&	$(	1.000	)$	&	$(	0.000	)$	&	$(	0.000	)$	&	$(	0.000	)$	&	$(	0.242	)$	&	$(	0.000	)$	&	$(	0.123	)$	&	$(	0.000	)$	&	$(	1.000	)$ & (0.999)	\\
		&	$3\boldsymbol{I}_{8}$	&		0.344		&		0.004		&		0.001		&		0.000		&		0.116   &   0.054		&		0.107		&		0.053		&		0.408 & 0.499	\\
		&		&	$(	1.000	)$	&	$(	0.400	)$	&	$(	0.000	)$	&	$(	0.000	)$	&	$(	0.301	)$	&	$(	0.000	)$	&	$(	0.175	)$	&	$(	0.000	)$	&	$(	1.000	)$	& (1.000)\\
		&	$5\boldsymbol{I}_{8}$	&		0.357		&		0.012		&		0.002		&		0.002		&		0.293   &   0.100		&		0.096		&		0.057		&		0.425	& 0.471	\\
		&		&	$(	1.000	)$	&	$(	0.980	)$	&	$(	0.000	)$	&	$(	0.000	)$	&	(0.760)  &  (0.120)	&	$(	0.151	)$	&	$(	0.000	)$	&	$(	1.000	)$  & (1.000)	\\
		&	&	&				&				&				&				&						\\
		10	&	$\boldsymbol{I}_{10}$	&		0.335		&		0.000		&		0.000		&		0.000		&		0.035   &   0.056		&		0.064		&		0.055		&		0.356	& 0.495	\\
		&		&	$(	1.000	)$	&	$(	0.000	)$	&	$(	0.000	)$	&	$(	0.000	)$	&	$(	0.119	)$	&	$(	0.000	)$	&	$(	0.117	)$	&	$(	0.000	)$	&	$(	1.000	)$	& (0.999)\\
		&	$3\boldsymbol{I}_{10}$	&		0.344		&		0.004		&		0.000		&		0.000		&		0.145		&		0.053		&		0.048		&		0.056		&		0.395	& 0.535	\\
		&		&	$(	1.000	)$	&	$(	0.520	)$	&	$(	0.000	)$	&	$(	0.000	)$	&	$(	0.382	)$	&	$(	0.000	)$	&	$(	0.115	)$	&	$(	0.000	)$	&	$(	1.000	)$& (1.000)	\\
		&	$5\boldsymbol{I}_{10}$	&		0.352		&		0.029		&		0.005		&		0.002		&		0.200   &   0.101		&		0.180		&		0.054		&		0.443	& 0.467	\\
		&		&	$(	1.000	)$	&	$(	0.980	)$	&	$(	0.000	)$	&	$(	0.000	)$	&(0.520)  &  (0.120)	&	$(	0.293	)$	&	$(	0.000	)$	&	$(	1.000	)$ & (1.000)	\\	[2ex] 
		\hline
	\end{tabular}
	\caption{Estimated misclassification rates of regular observations (and proportion of undetected outliers within parentheses) for outlying cluster contaminated datasets.}
	\label{table3.3}
\end{table}

\begin{table}[!t]
	\hspace{-0.5in}
	\small
	\begin{tabular}{c c c c c c c c c c c c} 
		\hline
		&       \multicolumn{4}{c}{\hspace{8em}MPLE$_{\beta}$} & & \multicolumn{2}{c}{TCLUST} & \multicolumn{2}{c}{TKMEANS} & KMEDOIDS & MCLUST \\ [1ex] 
		$Type$ &  $p$  &  $\beta=0$  & $\beta=0.1$ & $\beta=0.3$ & $\beta=0.5$ & $\alpha=0$ &  $\alpha=0.05$ & $\alpha=0$  & $\alpha=0.05$ & \\ [1ex] 
		\hline
		Pure	&	2	&	0.004	&		0.004		&	0.004	&		0.004		&	0.003	&		0.051		&	0.010	&		0.055		&	0.009	&		0.004		\\
		&		&	(0.000)	&	(	0.0002	)	&	(0.0002)	&	(	0.0002	)	&	(0.000)	&	(	0.050	)	&	(0.001)	&	(	0.050	)	&	(0.000)	&	(	0.000	)	\\
		&		&		&				&		&				&		&				&		&				&		&				\\
		&		&		&				&		&				&		&				&		&				&		&				\\
		Pure	&	6	&	0.0002	&		0.0002		&	0.0003	&		0.0003		&	0.000	&		0.052		&	0.001	&		0.001		&	0.0004	&		0.0005		\\
		&		&	(0.0001)	&	(	0.0002	)	&	(0.0002)	&	(	0.0003	)	&	(0.000)	&	(	0.050	)	&	(0.001)	&	(	0.050	)	&	(0.000)	&	(	0.000	)	\\
		&		&		&				&		&				&		&				&		&				&		&				\\
		&		&		&				&		&				&		&				&		&				&		&				\\
		&		&		&				&		&				&	$\alpha=$0.10	&		$\alpha=$0.15		&	$\alpha=$0.10	&	$\alpha=$	0.15		&		&				\\
		&		&		&				&		&				&		&				&		&				&		&				\\
		Uniform	&	2	&	0.029	&		0.028		&	0.022	&		0.024		&	0.008	&		0.058		&	0.012	&		0.057		&	0.009	&		0.021		\\
		(chi-squared)&		&	(0.180)	&	(	0.147	)	&	(0.000)	&	(	0.000	)	&	(0.030)	&	(	0.000	)	&	(0.038)	&	(	0.000	)	&	(1.000)	&	(	0.000	)	\\
		Contaminated&		&		&				&		&				&		&				&		&				&		&				\\
		&		&		&				&		&				&		&				&		&				&		&				\\
		
		Uniform	&	6	&	0.003	&		0.002		&	0.003	&		0.003		&	0.005	&		0.054		&	0.005	&		0.056		&	0.0003	&		0.001		\\
		(chi-squared)&		&	(0.054)	&	(	0.000	)	&	(0.000)	&	(	0.000	)	&	(0.032)	&	(	0.000	)	&	(0.036)	&	(	0.000	)	&	(1.000)	&	(	0.000	)	\\
		Contaminated&		&		&				&		&				&		&				&		&				&		&				\\
		&		&		&				&		&				&		&				&		&				&		&				\\
		Uniform	&	2	&	0.014	&		0.013		&	0.006	&		0.005		&	0.008	&		0.053		&	0.013	&		0.057		&	0.010	&		0.009		\\
		(Annulus)&		&	(0.000)	&	(	0.022	)	&	(0.000)	&	(	0.000	)	&	(0.038)	&	(	0.000	)	&	(0.042)	&	(	0.000	)	&	(1.000)	&	(	0.000	)	\\
		Contaminated&		&		&				&		&				&		&				&		&				&		&				\\
		&		&		&				&		&				&		&				&		&				&		&				\\
		Uniform	&	6	&	0.003	&		0.0002		&	0.0004	&		0.0004		&	0.004	&		0.057		&	0.005	&		0.056		&	0.0002	&		0.002		\\
		(Annulus) &		&	(0.002)	&	(	0.004	)	&	(0.003)	&	(	0.004	)	&	(0.044)	&	(	0.0003	)	&	(0.038)	&	(	0.000	)	&	(1.000)	&	(	0.004	)	\\
		Contaminated&		&		&				&		&				&		&				&		&				&		&				\\
		&		&		&				&		&				&		&				&		&				&		&				\\
		
		Outlying	&	2	&	0.067	&		0.068		&	0.026	&		0.025		&	0.104	&		0.057		&	0.090	&		0.058		&	0.358	&		0.340		\\
		Cluster &		&	(0.271)	&	(	0.240	)	&	(0.000)	&	(	0.000	)	&	(0.274)	&	(	0.000	)	&	(0.255)	&	(	0.000	)	&	(1.000)	&	(	0.999	)	\\
		&		&		&				&		&				&		&				&		&				&		&				\\
		Outlying	&	6	&	0.036	&		0.018		&	0.003	&		0.003		&	0.094	&		0.053		&	0.084	&		0.055		&	0.336	&		0.332		\\
		Cluster &		&	(0.996)	&	(	0.700	)	&	(0.000)	&	(	0.000	)	&	(0.266)	&	(	0.000	)	&	(0.188)	&	(	0.000	)	&	(1.000)	&	(	0.999	)	\\[2ex] 
		\hline
	\end{tabular}
	\caption{Estimated misclassification rates with proportions of regular observations misclassified as outliers (in case of pure datasts) and proportion of undetected outliers (in case of contaminated datasets) (within parentheses) for datasets with differentially dispersed clusters.}
	\label{table3.4}
\end{table}


In case of pure datasets, the estimated misclassification rates (averaged over the $100$ samples) along with the proportion of regular observations misclassified as outliers are presented in Table \ref{table3.1}. All the methods are very similar in terms of the estimated misclassification rates. As expected, the proposed method performed the best in case of $\beta=0$, the case corresponding to maximum likelihood along with the ER and NS constraints. Further, in the present proposal, the proportions of regular observations misclassified as outliers (presented in parenthesis below the misclassification rates) are very small indicating the method is doing the correct thing when the data are from pure models.

In case of uniformly (chi-squared method) contaminated datasets, the estimated misclassification rates of the regular observations (which are not outliers) along with the estimated proportion of undetected outliers are presented in Table \ref{table3.15}. In many cases, the proposed method with $\beta \approx 0.1$ or $0.3$ have lower estimated regular misclassification rates in comparison with the trimming based methods, while being competitive in other cases. The K-medoids method generates slightly lower missclassification rates than the proposed method for several of the cases. However, the K-medoid method is not adaptable for detection of outliers  and fails in this respect. The MCLUST method produces marginally better results in comparison to almost all the methods in this case.

The relative performance (presented in Table \ref{table3.2}) of the MPLE$_\beta$ method in comparison with the other methods is better in case of uniformly (from annulus) contaminated datasets (as compared to that for uniformly (chi-squared) contaminated datasets). The proposed method clearly beats trimmed K-means and has a very similar performance to that of the MCLUST method. The K-medoids have slightly lower misclassification rates but have no outlier detection capability.


The simulation outputs for the outlying cluster contaminated datasets, presented in Table \ref{table3.3}, indicate that the proposed method clearly outperforms the other methods in this case. It should be noted that the outlying cluster simulation scheme, being stochastic in nature, sometimes generates a contaminating proportion slightly larger than $10\%$ (although the assignment probability is exactly $0.1$ for the contaminating part), and consequently, some of these contaminated observations (very distant) are not trimmed by the $\alpha=0.1$ trimming level in TCLUST. This can have a potentially negative impact on the performance of TCLUST in such cases. The K-medoids performs poorly in terms of the misclassification rate of regular observations. MCLUST performs the worst in terms of classification of regular observations and also fares poorly in detecting outlying observations. 

For the differentially dispersed simulation set-up, the proposed method has performed the best more or less as compared to the other methods in terms of misclassification rates as well as bias and mean squared errors of the cluster means. The superiority is more prominent in case of outlying cluster contaminated datasets.

Although the K-medoids method performs better than some of its competitors in terms of estimated regular misclassification rates in some cases, it completely fails to detect outliers and also has higher bias and mean squared errors of the estimated cluster centers compared to those obtained by the present proposal with moderate values of $\beta$ (see Appendix \ref{appen4}).    

The estimated bias and mean squared errors of the estimated cluster means are presented in Appendix \ref{appen4}. The proposed method performs quite well on the average. The use of TCLUST with an $\alpha=0.15$ trimming level (i.e., a trimming level larger than the actual contamination rate) also performs quite well with respect to these accuracy measures. 

It should be noted that the maximal-gap idea can also be applied to the TCLUST method (and possibly to some other robust clustering methodologies) for a data-driven optimal choice of the trimming proportion. However, it could be the subject of a future study to investigate how this type of refinement might affect the TCLUST procedure. 

\section{Real Data Examples}
\label{secrealdata}
\subsection{Swiss Bank Notes Data} 

These data, originally considered in \cite{swissbank}, have been accessed from the R-cloud of datasets. The data consist of 200 old Swiss 1000-franc bank notes. It is known that the first 100 notes are genuine and the remaining are counterfeit. Our interest is in determining whether our algorithm can detect the counterfeit notes based on these data. Six measurements are made on each bank note:  (i) length of the bank note, (ii) height of the bank note (measured along the left side), (iii) height of the bank note (measured along the right side), (iv) distance of inner frame to the lower border, (v) distance of inner frame to the upper border and (vi) length of the diagonal.
\begin{figure}[h!]
	\centering
	\begin{tabular}{cc}
		\begin{subfigure}{0.45\textwidth}\centering\includegraphics[width=.6\columnwidth]{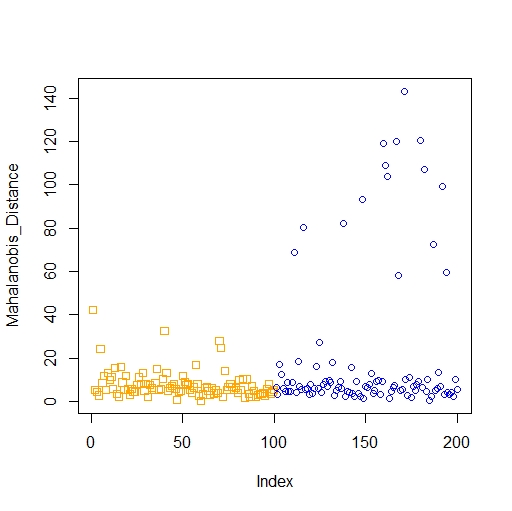}\caption{Original}
			\label{Figure 2}\end{subfigure}&
		\begin{subfigure}{0.45\textwidth}\centering\includegraphics[width=.6\columnwidth]{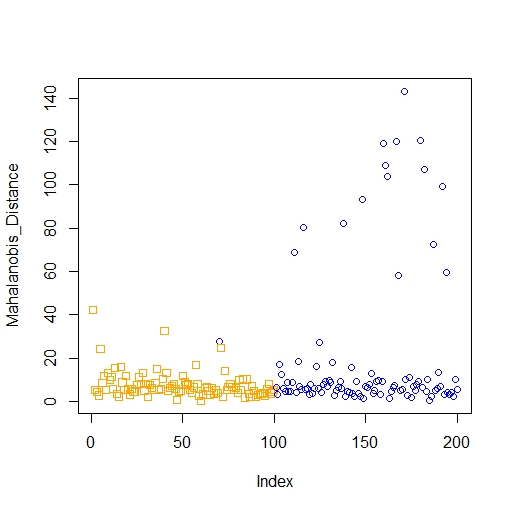}\caption{K-medoids}\end{subfigure}\\
		
		\begin{subfigure}{0.45\textwidth}\centering\includegraphics[width=.6\columnwidth]{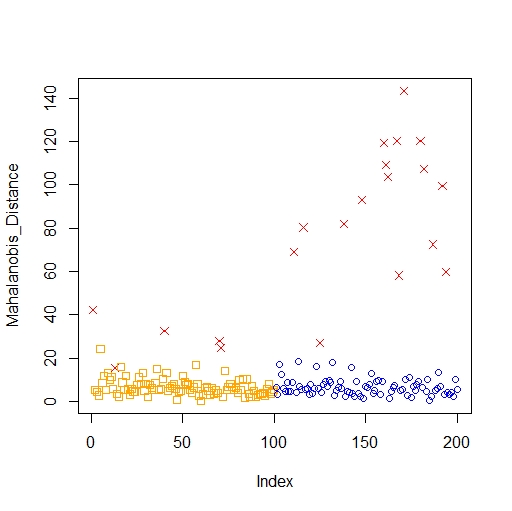}\caption{MPLE$_{\beta}$, $\beta=0.5$}\end{subfigure}&
		\begin{subfigure}{0.45\textwidth}\centering\includegraphics[width=.6\columnwidth]{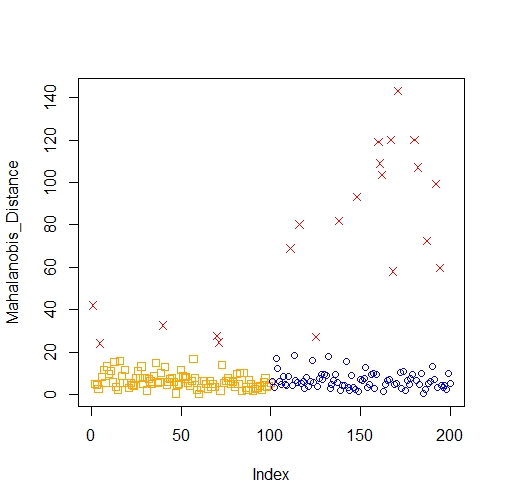}\caption{MCLUST}\end{subfigure}
		\\
		
		\begin{subfigure}{0.45\textwidth}\centering\includegraphics[width=.6\columnwidth]{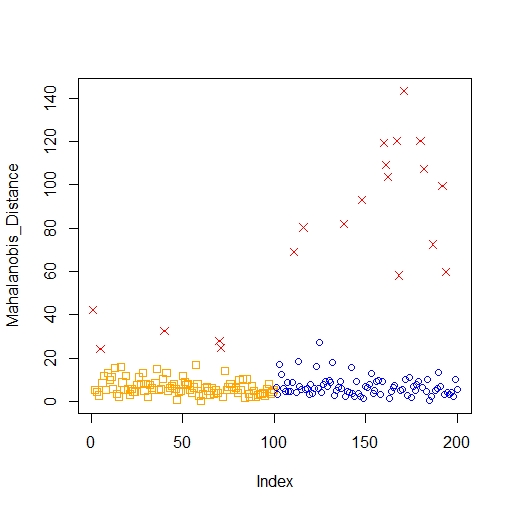}\caption{TCLUST, $\alpha=0.10$}\end{subfigure}&
		
		\begin{subfigure}{0.45\textwidth}\centering\includegraphics[width=.6\columnwidth]{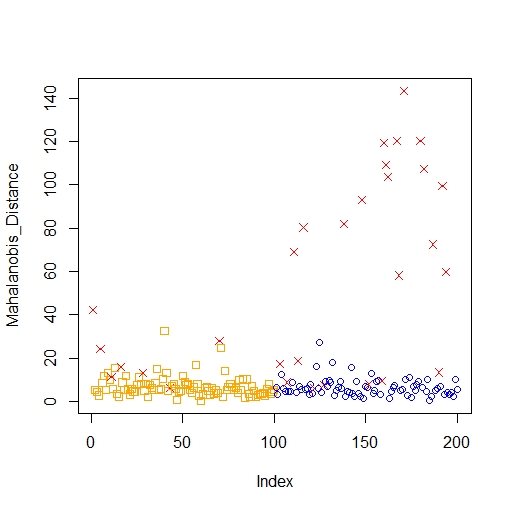}\caption{Trimmed K-means, $\alpha=0.15$}\end{subfigure}\\
		
	\end{tabular}
	\caption{Clusters derived from different methods for the Swiss Bank Notes data. 
		The vertical axis presents estimated Mahalanobis distances of the observations from their respective estimated (MCD) cluster centers.
		[Orange square: cluster $1$ of genuine notes; blue circle: cluster $2$ of counterfeit notes;  
		red cross: outliers identified by the corresponding algorithm] }
	\label{table11}
\end{figure}

To study these data, we begin with an exploratory data analysis. We split the data into two groups according to the true nature of the notes (i.e., genuine or counterfeit) and then estimate the location and dispersion of each of these groups using the minimum covariance determinant (MCD) method (\cite{mcd1}). We then estimate the Mahalanobis distances of the observations from their respective group (or cluster) centers. For the $i$-th observation $\boldsymbol{X}_{i}$, therefore, we compute,
\begin{equation}
\label{100.23}
d_{i} = 
\begin{cases}
(\boldsymbol{X}_{i}-\hat{\boldsymbol{\mu}}_{1})^{'}\hat{\boldsymbol{\Sigma}}_{1}^{-1}(\boldsymbol{X}_{i}-\hat{\boldsymbol{\mu}}_{1}),& \text{for } 1\leq i \leq 100\\
(\boldsymbol{X}_{i}-\hat{\boldsymbol{\mu}}_{2})^{'}\hat{\boldsymbol{\Sigma}}_{2}^{-1}(\boldsymbol{X}_{i}-\hat{\boldsymbol{\mu}}_{2}),            & \text{for } 101\leq i \leq 200
\end{cases}
\end{equation}
where $\hat{\boldsymbol{\mu}}_{j}$ and $\hat{\boldsymbol{\Sigma}}_{j}$ are the MCD based location and dispersion estimates of the $j$-th group, $j=1,2$.

Figure \ref{Figure 2} presents the index plot of these (robust) Mahalanobis distances.
In this plot, some points are far above the baseline with ordinates that are much larger compared to the ordinates of the general cloud of points concentrated near the horizontal axis. These observations are \enquote{far away} from their true cluster centers in terms of their estimated Mahalanobis distances and are therefore anomalous. It should also be noted that these anomalous observations are primarily from the counterfeit group of notes. In the figure, the genuine notes are represented as orange squares, and the counterfeit notes as blue circles. 

Now, we apply our robust method on these data along with the K-medoids, trimmed K-means, TCLUST and MCLUST methods. For our proposed method, we have taken $\beta=0.5$. For the TCLUST we have taken $\alpha = 0.10$, while for the  trimmed K-means $\alpha = 0.15$ has been used. The MCLUST method has been applied with $2$ clusters (and uniform noise component).  
The other panels of Figure \ref{table11} contain the clusters derived from each of these methods. 
Since the data are 6 dimensional, it is not possible to present the clusters along with the scatterplot of the data. So, in the remaining panels of Figure \ref{table11}, we present the Mahalanobis distance values (as depicted in Figure \ref{Figure 2}) for each index (through the observed magnitudes), 
the classification results according to the specified algorithm (orange squares representing observations classified as genuine and blue circles representing observations classified as counterfeit) 
and the outliers detected by the algorithm (depicted by red crosses).

In general it appears that the robust methods are all successful in doing the classifications and labeling the outliers correctly. There are rare misclassifications in the K-medoids case, but, more importantly, the latter makes no contribution to the issue of anomaly detection.  

\subsection{Seed Data} 

These data (which can be found in this \href{https://archive.ics.uci.edu/ml/datasets/seeds#}{link})  contain measurements of geometrical properties of three different varieties of wheat, namely, Kama, Rosa and Canadian. This data set consists of 210 observations on seven attributes that are all continuous and real-valued. The attributes are (i) area $A$, (ii) perimeter $P$, (iii) compactness $C=4\pi A/P^{2}$, (iv) length of kernel, (v) width of kernel, (vi) asymmetry coefficient and (vii) length of kernel groove, respectively.

\begin{figure}[h!]
	\centering
	\begin{tabular}{cc}
		\begin{subfigure}{0.45\textwidth}\centering\includegraphics[width=.6\columnwidth]{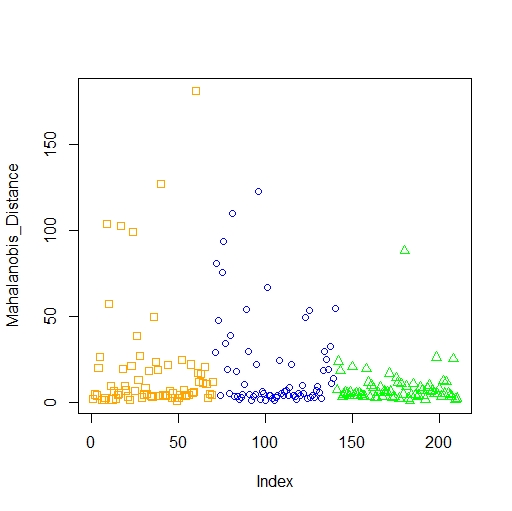}\caption{Original}
			\label{Figure 3}\end{subfigure}&
		\begin{subfigure}{0.45\textwidth}\centering\includegraphics[width=.6\columnwidth]{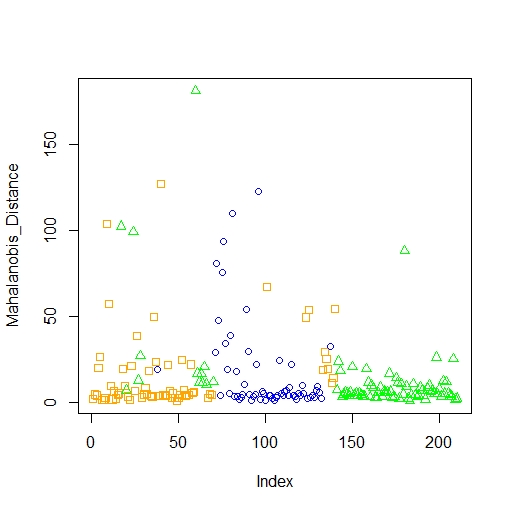}\caption{K-medoids}\end{subfigure}\\
		
		\begin{subfigure}{0.45\textwidth}\centering\includegraphics[width=.6\columnwidth]{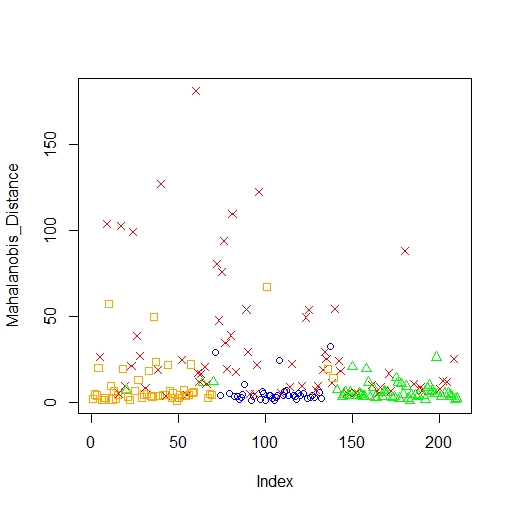}\caption{MPLE$_{\beta}$, $\beta=0.3$}\end{subfigure}&
		\begin{subfigure}{0.45\textwidth}\centering\includegraphics[width=.6\columnwidth]{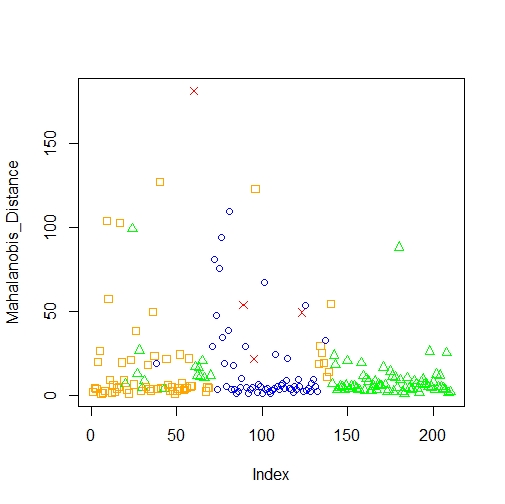}\caption{MCLUST}\end{subfigure}
		\\
		\begin{subfigure}{0.45\textwidth}\centering\includegraphics[width=.6\columnwidth]{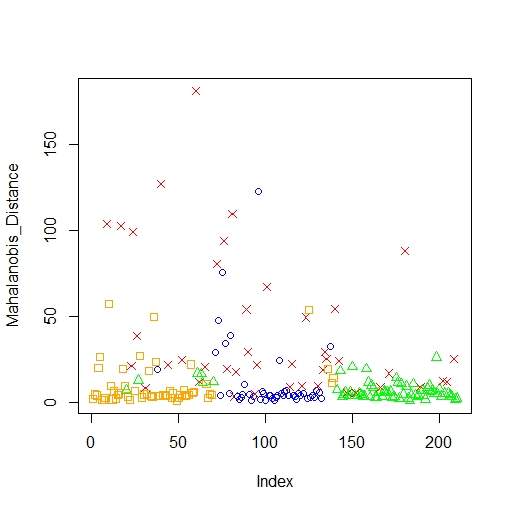}\caption{TCLUST, $\alpha=0.20$}\end{subfigure}&
		\begin{subfigure}{0.45\textwidth}\centering\includegraphics[width=.6\columnwidth]{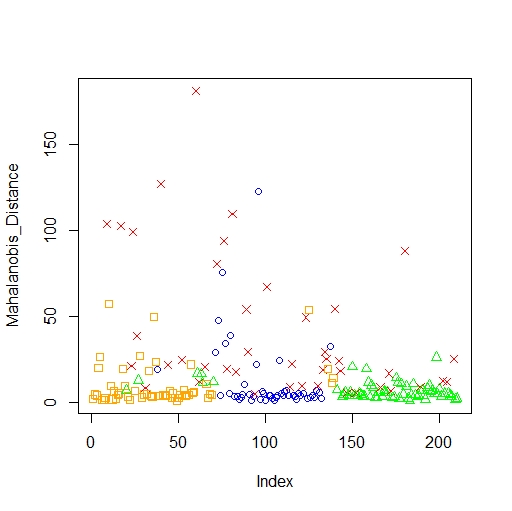}\caption{Trimmed K-means, $\alpha=0.20$}\end{subfigure}\\
	\end{tabular}
	\caption{Clusters derived from different methods for the Seed data.
		The vertical axis presents estimated Mahalanobis distances of the observations from their respective estimated (MCD) cluster centers.
		[Orange square: cluster $1$; blue circle: cluster $2$; green triangle: Cluster $3$; 
		red cross: outliers identified by the corresponding algorithm] }
	\label{table12}
\end{figure}

As we have done in case of the Swiss Bank Notes data, we first perform an exploratory data analysis by calculating the MCD estimates of cluster centers and dispersion matrices and calculate the estimated Mahalanobis distances of the observations from their respective cluster center estimates (MCD). These distances are presented in Figure \ref{Figure 3} which confirms the presence of some outlying observations. These observations are \enquote{far} away from the baseline as the Mahalanobis distances of these points from their respective estimated (MCD) cluster centers are much larger compared to majority of the points. 

The presence of such outlying observations suggests the need for robust clustering tools to analyze these data. 
We will apply our method with $\beta=0.3$ and compare it with the K-medoids, TCLUST (with $\alpha=0.2$), trimmed K-means (with $\alpha=0.2$) and MCLUST algorithm with $3$ clusters (with uniform noise component). 
Once again we represent the derived clusters and outliers, in Figure \ref{table12}, by expressing the Mahalanobis distance index curve 
using different colours and shapes, with orange square (cluster $1$), blue circle (cluster $2$) and green triangle (cluster $3$) representing the three clusters, 
and red crosses indicating the outliers.    

The classifications observed in Figure \ref{table12} indicate that the K-medoids algorithm, apart from failing in terms of outlier detection, lead to too many misclassified observations. Our proposed algorithm, the TCLUST and the trimmed K-means algorithms provide improvements through better classification and outlier detection, although neither classification nor outlier detection is done as perfectly as in case of the Swiss Bank Notes data. Our proposal does marginally better than the TCLUST and the trimmed K-means methods in terms of controlling the misclassification. The MCLUST method failed to detect the outlying observations properly. 
\section{Extension to Image Processing} 
\label{secimage}
Unsupervised methods for image analysis with anomaly detection is an important class of techniques in computer vision 
with applications in astronomy, biology, geology and many other fields. 
We analyze a satellite image presented in Figure \ref{table13} (obtained from  
\href{https://medium.com/swlh/using-deep-learning-semantic-segmentation-method-for-ship-detection-on-satellite-optical-imagery-ffeaae8c1ab}{link}, see this \cite{em}) 
with our method along with the aforesaid algorithms. 
For this purpose, we divide the original high resolution picture into $500 \times 500$ pixels (Figure \ref{Fig:imageO}). 
Each pixel consists of a combination of three different colours (red, blue and green) with different intensities ranging from $0$ to $1$. 
We observe the intensities of the colours in a pixel as a multivariate three dimensional observation. 
\begin{figure}[h!]
	\centering
	\includegraphics[height=3 in]{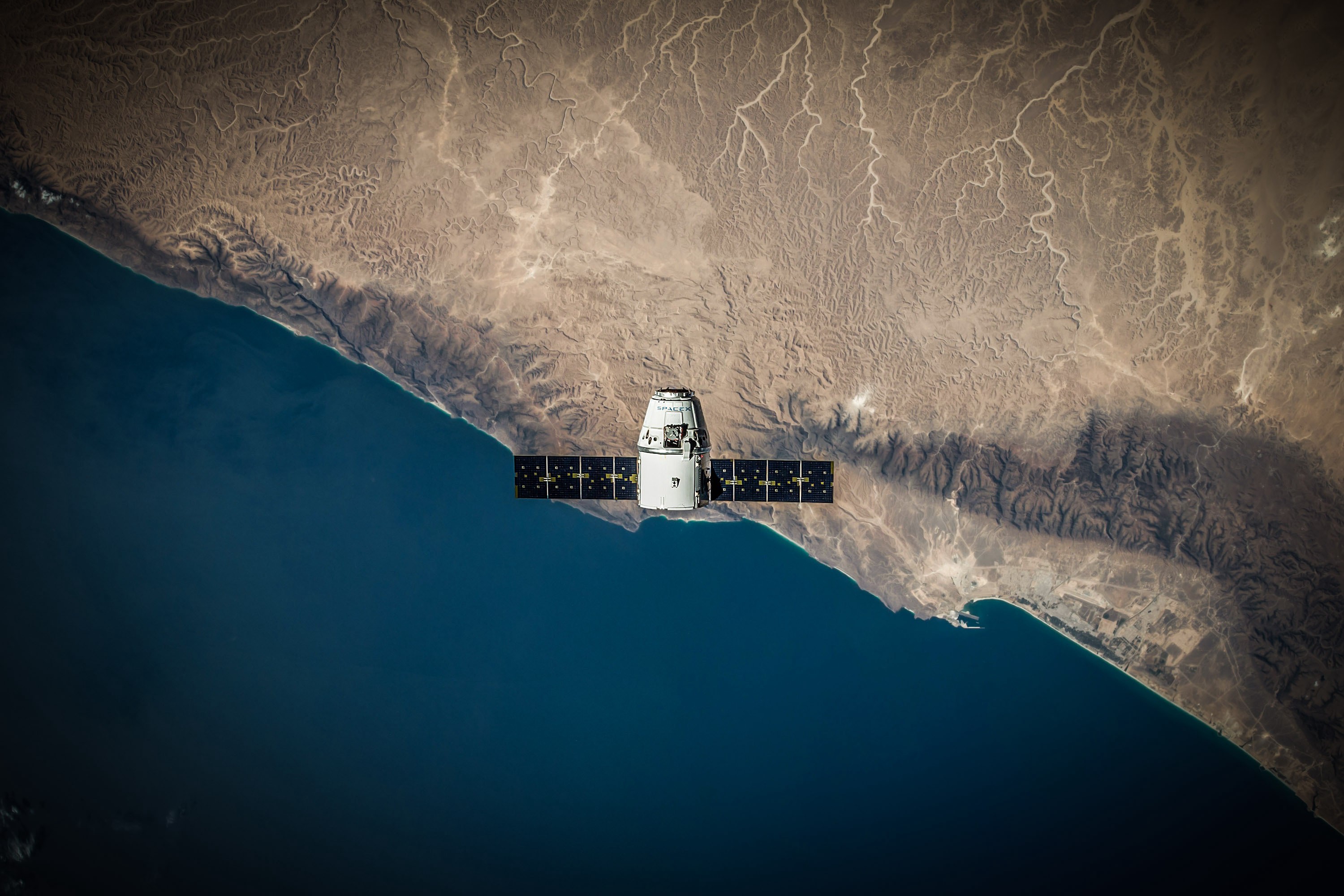}
	\caption{An Example of Satellite Image.}
	\label{table13}
\end{figure}

The two main components of the picture are the ocean body (blue region along the bottom and the left border of the image) and the coastal area (rest of the image along the top and right border) adjacent to it. But the flying body (the white body with black wings) and the shadow like regions (comparatively black regions in the coastal part, just adjacent to the ocean) should be treated as anomalies as they do not really belong to any of the two primary components. 
Thus we describe the image as a sample of size $500 \times 500=2.5 \times 10^{5}$ from a three dimensional, 
two component mixture population where the dimensions are the intensities of red, blue and green colours 
and the mixing components are the ocean body and the coastal area along with the aforesaid anomalies. 
We apply our method (along with the TCLUST, trimmed K-means, MCLUST and the ordinary K-means algorithms; K-medoids is avoided due to very large sample size of this data ($250000$) which is not permissible in the \enquote{PAM} function of the R software which is used to implement the K-medoids algorithm) with two clusters and check 
whether these methods can correctly identify those components and detect the anomalous structures (the flying object and the shadow like regions). 
In analyzing the data set and reconstructing the original image, we make the following modifications to our algorithm.   
\\
\\
\noindent
\textbf{Initialization:} The initialization and estimation methods are same as before. 
\\
\\
\noindent	
\textbf{Modification 1:} The first modification is in the assignment step. We use the minimum distance principle rather than the maximum likelihood principle in assigning the observations to different clusters. That is, we now assign $X$ (the intensity vector of red, green and blue colours for a particular pixel) to $C_{j}$ if the intensity vector of the cluster centre corresponding to $C_j$ is closest to $X$ (in the Euclidean norm) compared to all other cluster centres (instead of assigning $X$ to a cluster which maximizes its likelihood).
\\
\\
\noindent	
\textbf{Modification 2:} The second modification is in the outlier detection step. Although in our data analysis examples in Section 6, declaration of outliers has not been cluster specific, outliers may have different sources of possible anomalies, so a further classification among them may be useful. Note that prior to the identification of outliers, our algorithm assigns all data points among the regular clusters. We use these class memberships to consider a classification of the outliers; with $k$ clusters, there are $k$ different types of outliers according to their cluster assignment. It is likely that these points will end up representing different things in the reconstructed image.

It may be noted that the above modifications are not specific to the image under considerations;
they can and should be appropriately incorporated while applying our proposed clustering techniques for any image under study. 
The distance based assignment (modification 1) is common in most image processing techniques.
To see the requirement of the second modification, let us consider the example image given in Figure \ref{table13}; 
in this image the shadow like regions and the flying object both can be regarded as outliers compared to the ocean body and coastal area 
but they are actually different objects. The MPLE$_{\beta}$ as well as the TCLUST, trimmed K-means and ordinary K-means methods will recognize these regions as outliers ignoring the structural difference between them. But this will not be helpful if the clustering algorithms aim to separately identify the anomalous flying object.
Such problems are of great practical importance, e.g., in aeronautics and marine science.
Although in our example we have illustrated the classification of outliers in two groups,  
the proposed methodology can easily be extended to similarly classify more than two types of outliers,
as required, depending on the image under consideration.  

\begin{figure}[h!]
	\centering
	\begin{tabular}{cc}
		\begin{subfigure}{0.5\textwidth}\centering\includegraphics[height=1.5in,width=.8\columnwidth]{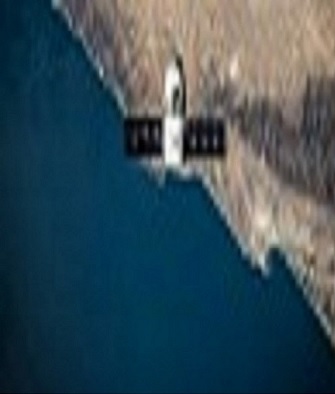}\caption{Original image with $500\times 500$ pixels }\label{Fig:imageO}\end{subfigure}&
		\begin{subfigure}{0.5\textwidth}\centering\includegraphics[height=1.5in,width=0.8\columnwidth]{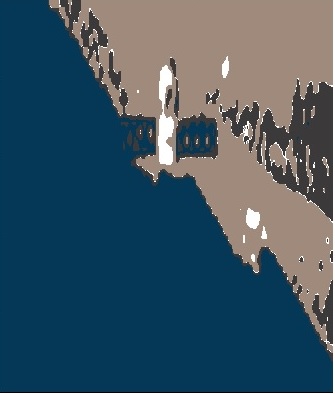}\caption{Output of MPLE$_{\beta}$ method}\end{subfigure}\\
		\begin{subfigure}{0.5\textwidth}\centering\includegraphics[height=1.5 in,width=.8\columnwidth]{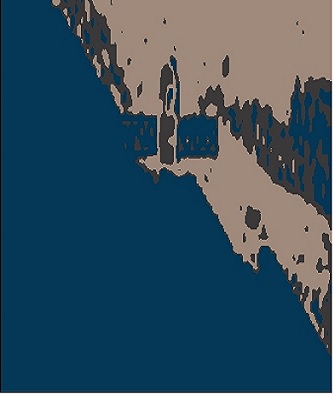}\caption{Output of trimmed K-means method}\end{subfigure}&
		\begin{subfigure}{0.5\textwidth}\centering\includegraphics[height=1.5 in,width=.8\columnwidth]{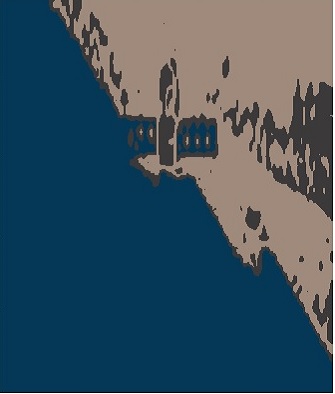}\caption{Output of TCLUST method}\end{subfigure}\\
		\begin{subfigure}{0.5\textwidth}\centering\includegraphics[height=1.5 in,width=.8\columnwidth]{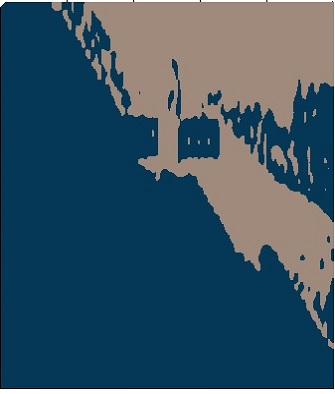}\caption{Output of K-means method}\end{subfigure}&
		\begin{subfigure}{0.5\textwidth}\centering\includegraphics[height=1.5 in,width=.8\columnwidth]{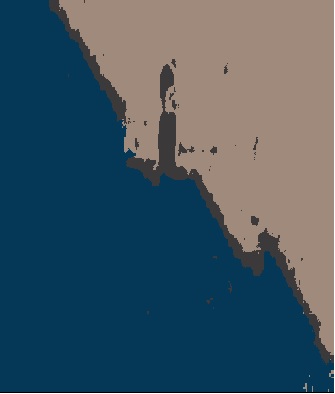}\caption{Output of MCLUST method}\end{subfigure}\\
	\end{tabular}
	\caption{Original and reconstructed images after applying different methods of clustering.}
	\label{table14}
\end{figure}

Now, we implement our proposed method to analyze the image under study in Figure \ref{Fig:imageO} (divided into $500\times 500$ pixels);
here we take $\beta=0.2$, $T=0.02$, $c=20$ and  $c_{1}=0.1$ in our proposal,
whereas for the TCLUST method $c=20$ and $\alpha=0.1$ are taken and for the trimmed K-means method $\alpha=0.1$ is taken; MCLUST is applied with $2$ clusters (with uniform noise component).
The reconstructed images with different methods are presented in Figure \ref{table14}. The water body and the coastal region can be clearly identified from the reconstructed images using all the algorithms. The brown shades indicate the possible anomalous regions in case of the TCLUST method and the trimmed K-means method. The trimmed K-means method is somewhat inefficient in detecting the outliers as some of the shadow like regions in the coastal area are misclassified as water body (blue regions within the brown outlying parts). But in our method the brown shades correspond to the shadow like regions whereas the white shades correspond to the flying object. Some regions in the coastal area which are close to white in the original image, are also detected as white outliers in our method.

So, we may conclude that the MPLE$_{\beta}$ and the TCLUST algorithms successfully point out these areas except for some areas on the wings of the flying object but the trimmed K-means method can not classify those areas perfectly. On the other hand, the MCLUST algorithm could not detect either the flying object or the shadow like anomalous region in the coastal area properly.

Additionally, our method specifically points out and distinguishes the flying object separately from the shadow like region through further classification of identified outliers. 
The existing implementations of the TCLUST or the trimmed K-means methods (in CRAN) do not distinguish between the outlier types. 
This additional refinement of our proposed methodology specifically helps to identify different small parts in the image  
making it a very useful robust clustering technique for image processing. 

For all the three clustering methods considered here, however, a little amount of misclassification occurs 
possibly due to the reduced resolution of the original image. 
\section{Conclusion and Future Plans} 
\label{secconfu}
We have proposed an algorithm which gives robust estimates of parameters of a mixture normal distribution and provides the cluster assignments of the observations along with likelihood based anomaly detection in the spirit of the density power divergence and $\beta$-likelihood. Some of the theoretical results of this procedure, including existence and consistency, has been presented in \cite{arxiv}). To explore the robustness of our algorithm, we have studied the influence functions of our estimators theoretically and have established the boundedness of the influence functions. The phenomenon has also been demonstrated graphically. Simulation studies have been presented in terms of regular misclassification error rates, proportion of undetected outliers, bias and mean squared errors of the estimated cluster means. On the whole, satisfactory results have been obtained and our method works competitively or better than the discussed methods in case of contaminated datasets in higher dimensions. Our method also exhibits satisfactory performance in analyzing the real datasets. In this work, we have compared our method with the K-medoid, TCLUST, trimmed K-means and MCLUST algorithms. In future, the performance of our method may also be compared with that of the $\gamma$-divergence based clustering algorithms (\cite{gammaclust2} and \cite{gammaclust1}) and other similar existing methods. 

We have presented a method which will provide an useful tool for applied scientists in many different domains. We have demonstrated the performance of the method in many applications and many different viewpoints. At the same time, it is important for us to be mindful about the issues involved in this research which remain unresolved (at least partially) as of now. Here we list such issues which are important in this connection. Satisfactory resolution of these issues which enrich and enhance the present research, and will hope to be able to solve them sometimes in the near future.


Firstly, most of the tuning parameters in this work have been selected either subjectively, or through ad-hoc ideas. This has been briefly discussed in Section \ref{sec2.3}. This aspect of the present work needs to be improved. In future, we will endeavour to build sophisticated statistical tools which would lead to the automatic selection of the tuning parameters taking into account the (unknown) amount of anomaly in the data.  

Secondly, some of the asymptotic results related to the proposed estimators need to be strengthened and appropriately extended. While proving the theoretical results about our estimators, it has been assumed that the true unknown distribution is a mixture of multivariate normals, i.e., it belongs to the model family. In the future, it will be of interest to establish the theoretical results under more general assumptions on the true unknown distribution to allow mixtures of other non-normal distributions including heavy tailed distributions.

Finally, the NS constraint is crucial to prove some of the theoretical results, and at the moment a proof which bypasses this condition is not available. However, it is required only for some very rare pathological cases (discussed at the end of Section \ref{SEC:theoretical_method}). In future, it will be among our primary goals to develop a proof which avoids the use of the NR constraint given the ER constraint, possibly by assuming some suitable moment conditions.   


Other possible future works may include the study of breakdown points to establish strong robustness properties of the estimators. Although, we have taken high dimensional datasets for simulation studies as well as real data applications, we have not considered the typical \enquote{high dimensional} set-up, where $n\ll p$. It will be interesting to extend this research from the present multivariate set-up to the actual high dimensional set-up, so that it can be efficiently applied to complex genomic, astronomical and other datasets. Robust classification tools may also be developed using maximum pseudo $\beta$-likelihood.
\\
\\
\noindent 
\textbf{Acknowledgement:} 
The research of AB is partially supported by the Technology Innovation Hub at Indian Statistical Institute, Kolkata under
Grant NMICPS/006/MD/2020-21 of Department of Science and Technology, Government of India, dated 16.10.2020. The research of AG is partially supported by the INSPIRE Faculty Research Grant from Department of Science and Technology, Government of India. The authors gratefully acknowledge the suggestions of two anonymous reviewers, which led to a significantly improved version of the manuscript.

 \begin{appendices}
 	\section{Robust clustering tools: A motivation} 
 	\label{appen1}
 	Let us consider the following scatter plots in Figure \ref{Figure 1}.  In the left panel of the plots, along with three main clusters, two outlying clusters (in red) are present. In the right panel, the three main clusters are contaminated with uniform noise (red points) from an annulus. If a clustering procedure with three clusters based on the estimation of the cluster mean using classical methods is run on the contaminated data, the cluster means are strongly affected and as a result the detection of the actual clusters could be highly inaccurate. The outliers affect both the estimators and the misclassification rates. 
 	Clearly, robust clustering methods may be useful in these situations to properly control the noise in the data. 
 	
 	\renewcommand{\thefigure}{S1}
 	\begin{figure}[h!]
 		\centering
 		\includegraphics[height=2.8 in,width= 6 in]{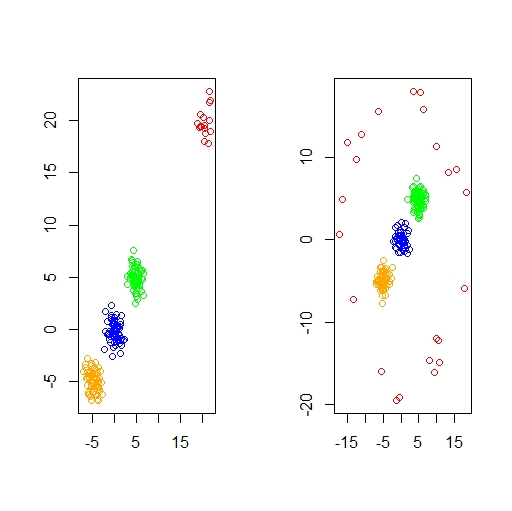}
 		\caption{Scatter plots with outlying cluster contamination (left panel) and uniform (from annulus) contamination (right panel), for a three cluster bivariate dataset.}
 		\label{Figure 1}
 	\end{figure}


 	\section{Mathematical Derivations}
 	\label{appen2}
 	\subsection{Proof of Theorem \ref{THM1}}
 	
 	To find the MDPD estimates of $\boldsymbol{\mu}$ and $\boldsymbol{\Sigma}$, it is enough to minimize,
 	\begin{align*}
 	H(\boldsymbol{\mu},\boldsymbol{\Sigma})&=\frac{1}{1+\beta}\int \phi_{p}^{1+\beta}(\boldsymbol{x},\boldsymbol{\mu},\boldsymbol{\Sigma}) \;d\boldsymbol{x} - \frac{1}{n\beta}\sum_{i=1}^{n}\phi_{p}^{\beta}(\boldsymbol{X}_{i},\boldsymbol{\mu},\boldsymbol{\Sigma})
 	\end{align*}
 	with respect to $\boldsymbol{\mu}$ and $\boldsymbol{\Sigma}$ where,
 	\begin{align*}
 	\phi_{p}(\boldsymbol{x},\boldsymbol{\mu},\boldsymbol{\Sigma}) = \frac{1}{(2\pi)^{\frac{p}{2}}|\boldsymbol{\Sigma}|^\frac{1}{2}}e^{-\frac{1}{2}(\boldsymbol{x}-\boldsymbol{\mu}){'}\boldsymbol{\Sigma}^{-1}(\boldsymbol{x}-\boldsymbol{\mu})}\; \text{for}\; \boldsymbol{x} \in \mathbb{R}^{p}.
 	\end{align*}
 	Since $N_{p}(\boldsymbol{\mu},\boldsymbol{\Sigma})$ model is a location-scale model, the first integral term in $H(\boldsymbol{\mu},\boldsymbol{\Sigma})$ $\Bigg(\text{given by} \frac{1}{(2\pi)^{\frac{p\beta}{2}}|\boldsymbol{\Sigma}|^{\frac{\beta}{2}}(1+\beta)^{\frac{p}{2}+1}}\Bigg)$ is independent of $\boldsymbol{\mu}$. So differentiating $H(\boldsymbol{\mu},\boldsymbol{\Sigma})$ with respect to $\boldsymbol{\mu}$ and equating the derivative to zero, we get the first estimating equation as,
 	\begin{equation}
 	\label{eqq*}
 	\frac{1}{n}\sum_{i=1}^{n}\phi_{p}^{\beta}(\boldsymbol{X}_{i},\boldsymbol{\mu},\boldsymbol{\Sigma})\;\frac{\partial \log\; \phi_{p}(\boldsymbol{X}_{i},\boldsymbol{\mu},\boldsymbol{\Sigma})}{\partial \boldsymbol{\mu}}=0.
 	\end{equation}
 	But, we have
 	\begin{align*}
 	\log\;\phi_{p}(\boldsymbol{x},\boldsymbol{\mu},\boldsymbol{\Sigma})& = -\frac{p}{2}\log\;(2\pi)-\frac{1}{2}\log\;|\boldsymbol{\Sigma}|-\frac{1}{2}(\boldsymbol{x}-\boldsymbol{\mu}){'}\boldsymbol{\Sigma}^{-1}(\boldsymbol{x}-\boldsymbol{\mu})
 	\end{align*}
 	and hence,
 	\begin{align*}
 	\frac{\partial \log\; \phi_{p}(\boldsymbol{X}_{i},\boldsymbol{\mu},\boldsymbol{\Sigma})}{\partial \boldsymbol{\mu}}&=\boldsymbol{\Sigma}^{-1}(\boldsymbol{X}_{i}-\boldsymbol{\mu}).
 	\end{align*}
 	So, the first estimating equation in Equation (\ref{eqq*}) simplifies to 
 	\begin{equation}
 	\centering
 	\label{Eq12}
 	\frac{1}{n}\sum_{i=1}^{n}\phi_{p}^{\beta}(\boldsymbol{X}_{i},\boldsymbol{\mu},\boldsymbol{\Sigma})(\boldsymbol{X}_{i}-\boldsymbol{\mu})=0.
 	\end{equation}
 	Next, to obtain the second estimating equation, our task is to differentiate $H(\boldsymbol{\mu},\boldsymbol{\Sigma})$ with respect to $\boldsymbol{\Sigma}$. But it will be cumbersome and thus we will carry on the calculation with respect to $\boldsymbol{\Sigma}^{-1}$ as was done to find the maximum likelihood estimates of $\boldsymbol{\mu}$ and $\boldsymbol{\Sigma}$ in Mardia et al. \cite{mardia}. To do that, we need the following lemma.  
 	
 	\begin{manuallemma}{S3.1}[Mardia et al. \cite{mardia}]\label{baz}
 		Suppose $f:\mathbb{R}^{m \times n}\rightarrow \mathbb{R}$ be a differentiable function. Then,
 		\begin{enumerate}
 			\item The derivative of $f(\boldsymbol{X})$ with respect to $\boldsymbol{X}=[[x_{ij}]]$  is given by,
 			\begin{align*}
 			\frac{\partial f(\boldsymbol{X})}{\partial \boldsymbol{X}}&=\left[\left[\frac{\partial f(\boldsymbol{X})}{\partial x_{ij}}\right]\right].
 			\end{align*}
 			\item If $\boldsymbol{X}$ is symmetric,
 			\[
 			\frac{\partial |\boldsymbol{X}|}{\partial x_{ij}}= 
 			\begin{cases}
 			\boldsymbol{X}_{ii},& \text{if } i=j,\\
 			2\boldsymbol{X}_{ij},  & \text{otherwise},
 			\end{cases}
 			\]
 			where $\boldsymbol{X}_{ij}$ being the $(i,j)$-th cofactor of $\boldsymbol{X}$.
 			\item If $\boldsymbol{X}$ is symmetric,
 			\begin{align*}
 			\frac{\partial tr(\boldsymbol{X}\boldsymbol{A})}{\partial \boldsymbol{X}}&= \boldsymbol{A}+\boldsymbol{A}^{'}-Diag(\boldsymbol{A}),
 			\end{align*}
 			where $tr(\boldsymbol{A})$ denotes the trace of the matrix $\boldsymbol{A}$ and $Diag(\boldsymbol{A})$ denotes the diagonal matrix whose diagonal elements are that of the matrix $\boldsymbol{A}$.
 		\end{enumerate}
 	\end{manuallemma}

 	Coming back to our problem, 
 	\begin{align*}
 	\frac{\partial \log\;\phi_{p}(\boldsymbol{x},\boldsymbol{\mu},\boldsymbol{\Sigma})}{\partial \boldsymbol{\Sigma}^{-1}}&=\frac{\partial }{\partial\boldsymbol{\Sigma}^{-1} }\left[\frac{1}{2}\log\;|\boldsymbol{\Sigma}^{-1}|-\frac{1}{2}tr(\boldsymbol{\Sigma}^{-1}(\boldsymbol{x}-\boldsymbol{\mu})(\boldsymbol{x}-\boldsymbol{\mu})^{'})\right]\\
 	&=\frac{\partial }{\partial \boldsymbol{V} }\left[\frac{1}{2}\log\;|\boldsymbol{V}|-\frac{1}{2}tr(\boldsymbol{V}(\boldsymbol{x}-\boldsymbol{\mu})(\boldsymbol{x}-\boldsymbol{\mu})^{'})\right] ,\;\boldsymbol{V}=\boldsymbol{\Sigma}^{-1}.
 	\end{align*}
 	Now, using Lemma \ref{baz},
 	\[
 	\frac{\partial }{\partial \boldsymbol{V} }\log\;|\boldsymbol{V}|=  \frac{1}{|\boldsymbol{V}|}\frac{\partial |\boldsymbol{V}|}{\partial \boldsymbol{V} }\; \text{whose}\; (i,j)\text{-th element is}
 	\begin{cases}
 	\frac{2\boldsymbol{V}_{ij}}{|\boldsymbol{V}|},& \text{if }i \neq j, \\
 	\frac{\boldsymbol{V}_{ii}}{|\boldsymbol{V}|},& \text{if }i = j,
 	\end{cases}
 	\]
 	where $\boldsymbol{V}_{ij}$ is the $(i,j)$-th cofactor of $\boldsymbol{V}$ and
 	\begin{align*}
 	\frac{\partial }{\partial \boldsymbol{V} }tr\left(\boldsymbol{V}(\boldsymbol{x}-\boldsymbol{\mu})(\boldsymbol{x}-\boldsymbol{\mu})^{'}\right)&=2(\boldsymbol{x}-\boldsymbol{\mu})(\boldsymbol{x}-\boldsymbol{\mu})^{'}-Diag(\boldsymbol{x}-\boldsymbol{\mu})(\boldsymbol{x}-\boldsymbol{\mu})^{'}.
 	\end{align*}
 	Since, $\boldsymbol{V}=\boldsymbol{\Sigma}^{-1}$ is symmetric, the matrix with elements $\frac{\boldsymbol{V}_{ij}}{|\boldsymbol{V}|}$ equals $\boldsymbol{V}^{-1}=\boldsymbol{\Sigma}$.
 	Hence, 
 	\begin{align*}
 	\frac{\partial \log\;\phi_{p}(\boldsymbol{x},\boldsymbol{\mu},\boldsymbol{\Sigma})}{\partial \boldsymbol{\Sigma}^{-1}}
 	&=\frac{\partial }{\partial \boldsymbol{V} }\left[\frac{1}{2}\log\;|\boldsymbol{V}|-\frac{1}{2}tr(\boldsymbol{V}(\boldsymbol{x}-\boldsymbol{\mu})(\boldsymbol{x}-\boldsymbol{\mu})^{'})\right] \;,\boldsymbol{V}=\boldsymbol{\Sigma}^{-1}\\
 	&=\frac{1}{2}\boldsymbol{M}-\frac{1}{2}(2\boldsymbol{S}-Diag(\boldsymbol{S})),
 	\end{align*}
 	where $\boldsymbol{S}=(\boldsymbol{x}-\boldsymbol{\mu})(\boldsymbol{x}-\boldsymbol{\mu})^{'}$ and $\boldsymbol{M}=2\boldsymbol{\Sigma}-Diag(\boldsymbol{\Sigma})$. 
 	Hence, differentiating $H(\boldsymbol{\mu},\boldsymbol{\Sigma})$ with respect to $\boldsymbol{\Sigma}^{-1}$ and equating the derivative to zero, we have the second estimating equation
 	
 	\begin{equation}
 	\centering
 	\label{eqq**}
 	\begin{array}{l}
 	\frac{1}{n}\sum_{i=1}^{n}\phi_{p}^{\beta}(\boldsymbol{X}_{i},\boldsymbol{\mu},\boldsymbol{\Sigma})\;\frac{\partial  \log \phi_{p}(\boldsymbol{X}_{i},\boldsymbol{\mu},\boldsymbol{\Sigma})}{\partial \boldsymbol{\Sigma}^{-1}}=e_{0}\\
 	\\
 	\implies \frac{1}{n}\sum_{i=1}^{n}\phi_{p}^{\beta}(\boldsymbol{X}_{i},\boldsymbol{\mu},\boldsymbol{\Sigma})\frac{1}{2}(\boldsymbol{M}-(2\boldsymbol{S}_{i}-Diag(\boldsymbol{S}_{i})))=e_{0},
 	\end{array}
 	\end{equation}
 	where $\boldsymbol{S}_i=(\boldsymbol{X}_i-\boldsymbol{\mu})(\boldsymbol{X}_i-\boldsymbol{\mu})^{'}$ and 
 	\begin{align*}
 	e_{0}&=\frac{\partial}{\partial \boldsymbol{\Sigma}^{-1}}\left[\frac{1}{1+\beta}\int \phi_{p}^{1+\beta}(\boldsymbol{x},\boldsymbol{\mu},\boldsymbol{\Sigma}) \;d\boldsymbol{x}\right ]\\
 	&=\int \phi_{p}^{1+\beta}(\boldsymbol{x},\boldsymbol{\mu},\boldsymbol{\Sigma})\frac{\partial \log\; \phi_{p}(\boldsymbol{x},\boldsymbol{\mu},\boldsymbol{\Sigma})} {\partial \boldsymbol{\Sigma}^{-1}}\;d\boldsymbol{x}\\
 	&=\frac{\boldsymbol{M}}{2}\int \phi_{p}^{1+\beta}(\boldsymbol{x},\boldsymbol{\mu},\boldsymbol{\Sigma}) \;d\boldsymbol{x}-\int (\boldsymbol{x}-\boldsymbol{\mu})(\boldsymbol{x}-\boldsymbol{\mu})^{'} \phi_{p}^{1+\beta}(\boldsymbol{x},\boldsymbol{\mu},\boldsymbol{\Sigma}) \;d\boldsymbol{x} + \frac{1}{2}\int Diag(\boldsymbol{x}-\boldsymbol{\mu})(\boldsymbol{x}-\boldsymbol{\mu})^{'} \phi_{p}^{1+\beta}(\boldsymbol{x},\boldsymbol{\mu},\boldsymbol{\Sigma}) \;d\boldsymbol{x}.
 	\end{align*}
 	Now using the facts that,
 	\begin{equation}
 	\centering 
 	\label{Eq13}
 	\phi_{p}^{1+\beta}(\boldsymbol{x},\boldsymbol{\mu},\boldsymbol{\Sigma})=\frac{1}{(2\pi)^{\frac{p\beta}{2}}|\boldsymbol{\Sigma}|^{\frac{\beta}{2}}(1+\beta)^{\frac{p}{2}}}\phi_{p}(\boldsymbol{x},\boldsymbol{\mu},\boldsymbol{\Sigma}_{0})
 	\end{equation}
 	with $\boldsymbol{\Sigma}_{0}=\frac{1}{1+\beta}\boldsymbol{\Sigma}$ and $(\boldsymbol{X}-\boldsymbol{\mu})(\boldsymbol{X}-\boldsymbol{\mu})^{'}$ follows a $p$-dimensional Wishart distribution with scale parameter $\boldsymbol{\Sigma}_0$ and 1 degree of freedom for a random vector $\boldsymbol{X}$ which follows $N_{p}(\boldsymbol{\mu},\boldsymbol{\Sigma}_{0})$,
 	\begin{align*}
 	e_{0}&=\frac{\boldsymbol{M}}{2}\frac{1}{(2\pi)^{\frac{p\beta}{2}}|\boldsymbol{\Sigma}|^{\frac{\beta}{2}}(1+\beta)^{\frac{p}{2}}} -\frac{\boldsymbol{M}}{2}\frac{1}{1+\beta}\frac{1}{(2\pi)^{\frac{p\beta}{2}}|\boldsymbol{\Sigma}|^{\frac{\beta}{2}}(1+\beta)^{\frac{p}{2}}}\\
 	&=\frac{\boldsymbol{M}\beta}{2(2\pi)^{\frac{p\beta}{2}}|\boldsymbol{\Sigma}|^{\frac{\beta}{2}}(1+\beta)^{\frac{p+2}{2}}}.
 	\end{align*}
 	Precisely, the second estimating equation in Equation $(\ref{eqq**})$ becomes,
 	\begin{align*} \frac{1}{n}\sum_{i=1}^{n}\phi_{p}^{\beta}(\boldsymbol{X}_{i},\boldsymbol{\mu},\boldsymbol{\Sigma})\frac{1}{2}(\boldsymbol{M}-(2\boldsymbol{S}_{i}-Diag(\boldsymbol{S}_{i})))=\frac{\boldsymbol{M}\beta}{2(2\pi)^{\frac{p\beta}{2}}|\boldsymbol{\Sigma}|^{\frac{\beta}{2}}(1+\beta)^{\frac{p+2}{2}}}\\
 	\implies \frac{1}{n}\sum_{i=1}^{n}\phi^{\beta}_{p}(\boldsymbol{X}_{i},\boldsymbol{\mu},\boldsymbol{\Sigma})(\boldsymbol{M}-(2\boldsymbol{S}_{i}-Diag(\boldsymbol{S}_{i})))&=c_{0}\boldsymbol{M}
 	\end{align*}
 	which is algebraically equivalent to,
 	\begin{equation}
 	\centering
 	\label{Eq14}
 	\frac{1}{n}\sum_{i=1}^{n}\phi^{\beta}_{p}(\boldsymbol{X}_{i},\boldsymbol{\mu},\boldsymbol{\Sigma})(\boldsymbol{\Sigma}-(\boldsymbol{X}_{i}-\boldsymbol{\mu})(\boldsymbol{X}_{i}-\boldsymbol{\mu})^{'})=c_{0}\boldsymbol{\Sigma}.
 	\end{equation}
 	where $c_{0}={\beta}{(2\pi)^{-\frac{p\beta}{2}}|\boldsymbol{\Sigma}|^{-\frac{\beta}{2}}(1+\beta)^{-\frac{p+2}{2}}}$. Now, by expanding the multivariate normal density $\phi_{p}(\boldsymbol{X}_{i},\boldsymbol{\mu},\boldsymbol{\Sigma})$ in Equations (\ref{Eq12}) and (\ref{Eq14}), we get the following simplified versions of the aforesaid estimating equations.
 	\begin{equation*}
 	\centering
 	\frac{1}{n}\sum_{i=1}^{n}e^{-\frac{\beta}{2}(\boldsymbol{X}_{i}-\boldsymbol{\mu})^{'}\boldsymbol{\Sigma}^{-1}(\boldsymbol{X}_{i}-\boldsymbol{\mu})}(\boldsymbol{X}_{i}-\boldsymbol{\mu})=0,
 	\end{equation*}
 	\begin{equation*}
 	\centering
 	\frac{1}{n}\sum_{i=1}^{n}e^{-\frac{\beta}{2}(\boldsymbol{X}_{i}-\boldsymbol{\mu})^{'}\boldsymbol{\Sigma}^{-1}(\boldsymbol{X}_{i}-\boldsymbol{\mu})}(\boldsymbol{\Sigma}-(\boldsymbol{X}_{i}-\boldsymbol{\mu})(\boldsymbol{X}_{i}-\boldsymbol{\mu})^{'})=\frac{\beta}{(1+\beta)^{\frac{p}{2}+1}}\boldsymbol{\Sigma}.
 	\end{equation*}
 	This completes derivation of the estimating equations.
 	\subsection{Proof of Theorem \ref{THM2}}
 	Since the cluster assignments and the estimates  $\hat{\boldsymbol{\mu}}_{j}$, $\hat{\boldsymbol{\Sigma}}_{j}$ are known, 
 	the values of the assignment functions $Z(\cdot,\cdot)$ are also known. 
 	Thus, maximizing the empirical objective function with respect to $(\pi_{1},\pi_{2},...,\pi_{k})$ is equivalent to the optimization problem:
 	\begin{align*}
 	~\mbox{\textbf{maximize}}~~~~
 	\sum_{j=1}^{k}n_{j}\log\;\pi_{j},
 	~~~~~\mbox{	\textbf{subject to}}~~~~
 	\sum_{j=1}^{k}\pi_{j}=1.
 	\end{align*}
 	Now the optimal choice of $\pi_{j}$ can be directly obtained by optimizing the Lagrangian,
 	\begin{align*}
 	l(\pi_{1},\pi_{2},...,\pi_{k}, \lambda)&=\sum_{j=1}^{k}n_{j}\log\;\pi_{j}-\lambda\left (\sum_{j=1}^{k}\pi_{j}-1\right ).
 	\end{align*}
 	
 	\section{Technical Derivations of Section \ref{SEC:IF}}
 	\label{appen3}
 	\subsection{Derivations of the Systems of Equations defining the Maximum Pseudo $\beta$-Likelihood Functional (MPLF$_\beta$)}
 	\label{appen3.1}
 	Consider the set-up and notation of Section \ref{SEC:IF}. 
 	We now present the derivation of the system of equations defining the MPLF$_\beta$ $\boldsymbol{\theta}(P)$, given in system (\ref{Eq*}) of equations.

 	Let us first note that, under $p=1,k=2$ and the true distribution function $P$,
 	our objective function for the MPLF$_\beta$ functional is given by
 	\begin{equation}
 	\centering
 	\label{Eq10}
 	L_{\beta}(\boldsymbol{\theta},P)=E_{P}\left[\sum_{j=1}^{2}Z_{j}(X,\boldsymbol{\theta})\left[\log\;\pi_{j}+\frac{1}{\beta}f^{\beta}(X,\mu_{j},\sigma^{2}_{j})-\frac{1}{1+\beta}\int f^{1+\beta}(x,\mu_{j},\sigma^{2}_{j}) \;dx]\right]\right],
 	\end{equation}
 	where $f(\cdot,\mu_{j},\sigma^{2}_{j})$ is the pdf of univariate normal distribution with mean $\mu_{j}$ and variance $\sigma^{2}_{j}$ for $j=1,2$, $P$ is the true but unknown distribution function and hence the parameter as a functional can be written as,
 	\begin{equation}
 	\centering 
 	\label{Eq11}
 	\boldsymbol{\theta}_{\beta}(P)=\underset{\boldsymbol{\theta} \in \boldsymbol{\Theta}_{C}}{\text{argmax}}\;L_{\beta}(\boldsymbol{\theta}, P).
 	\end{equation}
 	
 	In our case, the assignment functions can be presented as,
 	\begin{align*}
 	Z_{1}(X,\boldsymbol{\theta})&=I(D_{1}(X,\boldsymbol{\theta})\geq D_{2}(X,\boldsymbol{\theta}))\; \text{and}\\
 	Z_{2}(X,\boldsymbol{\theta})&=1-Z_{1}(X,\boldsymbol{\theta}).
 	\end{align*}
 	Note that $D_{1}(X,\boldsymbol{\theta})$ and $D_{2}(X,\boldsymbol{\theta})$ are bell-shaped convex functions. Hence,
 	\begin{align*}
 	&D_{1}(X,\boldsymbol{\theta})\geq D_{2}(X,\boldsymbol{\theta})\\
 	&\implies \frac{(X-\mu_{2})^2}{\sigma^{2}_{2}}-\frac{(X-\mu_{1})^2}{\sigma^{2}_{1}}\geq 2\log\;\frac{\pi_{2}\sigma_{1}}{\pi_{1}\sigma_{2}}\\
 	&\implies a_{1}X^{2}+a_{2}X+a_{3}\geq0
 	\end{align*}
 	for suitable values of $a_{1},a_{2}$ and $a_{3}$. Since, a quadratic equation can have at most two real roots, the aforesaid inequality leads to the following possibilities,
 	
 	\textbf{Possibility 1:} $X\in(a,b)$ if $ a_{1}<0$.
 	
 	\textbf{Possibility 2:} $X \in (-\infty,a)\cup(b,\infty)$ if $ a_{1}>0$.
 	
 	\textbf{Possibility 3:} $X \in (-\infty,a)$ or $X \in (a,\infty)$ if $ a_{1}=0$.
 	
 	Let us carefully note that, the constants $a$ and $b$ depend on the parameter $\boldsymbol{\theta}$. So, we have to consider these additional functionals $a(P)$ and $b(P)$ in order to derive the influence function of $\boldsymbol{\theta}(P)$. The above  possibilities can be graphically observed in Figure \ref{table 1}.
 	\renewcommand{\thefigure}{S2}
 	\begin{figure}[h!]
 		\centering
 		\begin{tabular}{cc}
 			\begin{subfigure}{0.4\textwidth}\centering\includegraphics[width=.6\columnwidth]{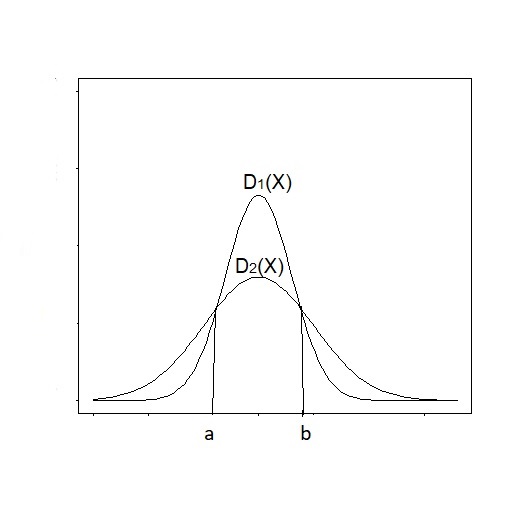}\caption{Possibility 1}\end{subfigure}&
 			\begin{subfigure}{0.4\textwidth}\centering\includegraphics[width=.6\columnwidth]{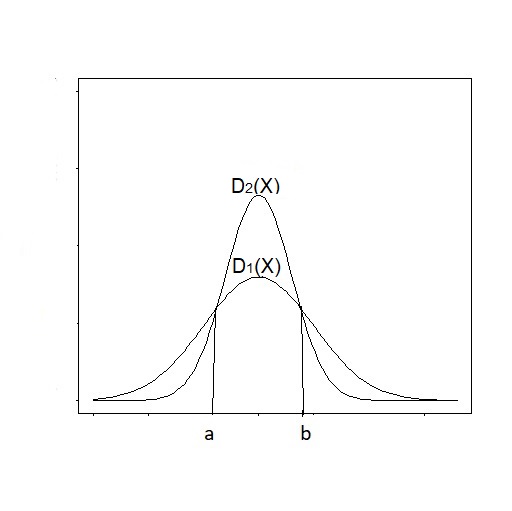}\caption{Possibility 2}\end{subfigure}\\
 			\newline
 			\begin{subfigure}{0.4\textwidth}\centering\includegraphics[width=.6\columnwidth]{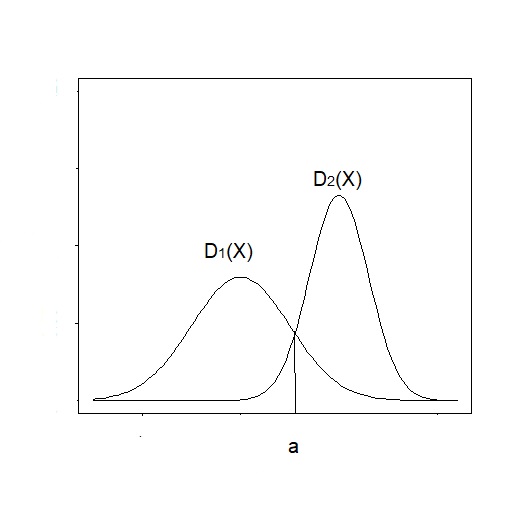}\caption{Possibility 3}\end{subfigure}&
 			\begin{subfigure}{0.4\textwidth}\centering\includegraphics[width=.6\columnwidth]{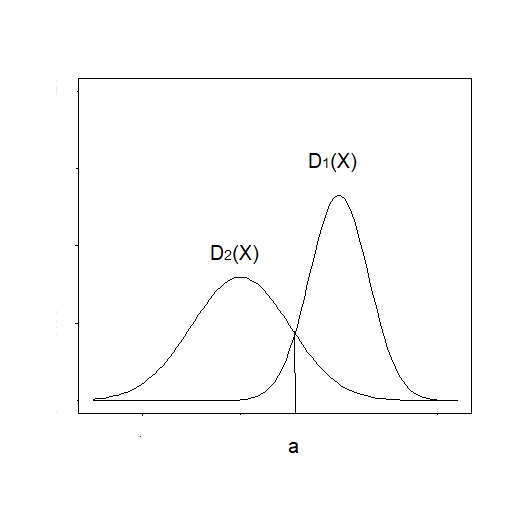}\caption{Possibility 3}\end{subfigure}\\
 		\end{tabular}
 		\caption{Different Possibilities}
 		\label{table 1}
 	\end{figure}
 	
 	We will first derive the influence function in possibility:1; the derivation of the same for the remaining possibilities are similar. To do that, let us first study the mathematical relationships among these functionals.
 	The proportion functional satisfies,
 	\begin{align*}
 	\pi_{1}(P)+\pi_{2}(P)=1.
 	\end{align*}
 	From the plots in Figure \ref{table 1}, it is intuitively clear that, the functional $\boldsymbol{\theta}_{\beta}(P) $ satisfies,
 	\begin{equation}
 	\centering
 	\label{eqq***}
 	\begin{array}{l}
 	D_{1}(a(P),\boldsymbol{\theta}_{\beta}(P))=D_{2}(a(P),\boldsymbol{\theta}_{\beta}(P)),\\
 	D_{1}(b(P),\boldsymbol{\theta}_{\beta}(P))=D_{2}(b(P),\boldsymbol{\theta}_{\beta}(P)).
 	\end{array}
 	\end{equation}
 	Next let us simplify $L_{\beta}(\boldsymbol{\theta},P)$ under $k=2$ and $p=1$.
 	\begin{align*}
 	L_{\beta}(\boldsymbol{\theta},P)&=E_{P}\left[\sum_{j=1}^{2}Z_{j}(X,\boldsymbol{\theta})\left[\log\;\pi_{j}+\frac{1}{\beta}f^{\beta}(X,\mu_{j},\sigma^{2}_{j})-\frac{1}{1+\beta}\int f^{1+\beta}(x,\mu_{j},\sigma^{2}_{j}) \;dx]\right]\right]\\
 	&=E_{P}\left[I(D_{1}(X,\boldsymbol{\theta})>D_{2}(X,\boldsymbol{\theta}))\left[\log\;\pi_{1}+\frac{1}{\beta}f^{\beta}(X,\mu_{1},\sigma^{2}_{1})-\frac{1}{1+\beta}\int f^{1+\beta}(x,\mu_{1},\sigma^{2}_{1}) \;dx]\right]\right]\\
 	&\;\;\;+E_{P}\left[I(D_{2}(X,\boldsymbol{\theta})>D_{1}(X,\boldsymbol{\theta}))\left[\log\;\pi_{2}+\frac{1}{\beta}f^{\beta}(X,\mu_{2},\sigma^{2}_{2})-\frac{1}{1+\beta}\int f^{1+\beta}(x,\mu_{2},\sigma^{2}_{2}) \;dx]\right]\right]\\
 	&=E_{P}\left[I(X \in (a,b))\left[\log\;\pi_{1}+\frac{1}{\beta}f^{\beta}(X,\mu_{1},\sigma^{2}_{1})-\frac{1}{1+\beta}\int f^{1+\beta}(x,\mu_{1},\sigma^{2}_{1}) \;dx]\right]\right]\\
 	&\;\;\;+E_{P}\left[I(X \notin (a,b))\left[\log\;\pi_{2}+\frac{1}{\beta}f^{\beta}(X,\mu_{2},\sigma^{2}_{2})-\frac{1}{1+\beta}\int f^{1+\beta}(x,\mu_{2},\sigma^{2}_{2}) \;dx]\right]\right]\\
 	&=\left[(P(b)-P(a))\log\;\pi_{1}+\frac{1}{\beta}\int_{a}^{b}f^{\beta}(x,\mu_{1},\sigma^{2}_{1})p(x)\;dx-\frac{P(b)-P(a)}{1+\beta}\int f^{1+\beta}(x,\mu_{1},\sigma^{2}_{1})\;dx\right]\\
 	&\;\;\;+\left[(1-P(b)+P(a))\log\;\pi_{2}+\frac{1}{\beta}\int_{x \notin (a,b)} f^{\beta}(x,\mu_{2},\sigma^{2}_{2})p(x)\;dx-\frac{1-P(b)+P(a)}{1+\beta}\int f^{1+\beta}(x,\mu_{2},\sigma^{2}_{2})\;dx\right]
 	\end{align*}
 	with $p(x)$ as the true density corresponding to the true distribution function $P(x)$.
 	Now using (\ref{Eq13}),
 	\begin{align*}
 	\int f^{1+\beta}(x,\mu_{j},\sigma^{2}_{j})\;dx&=\frac{1}{(2\pi)^{\frac{\beta}{2}}\sigma^{\beta}_{j}(1+\beta)^{\frac{1}{2}}},\;j=1,2.
 	\end{align*}
 	Hence,
 	\begin{align*}
 	L_{\beta}(\boldsymbol{\theta},P)&=\left[(P(b)-P(a))\log\;\pi_{1}+\frac{1}{\beta}\int_{a}^{b}f^{\beta}(x,\mu_{1},\sigma^{2}_{1})p(x)\;dx-\frac{P(b)-P(a)}{(2\pi)^{\frac{\beta}{2}}\sigma^{\beta}_{1}(1+\beta)^{\frac{3}{2}}}\right]\\
 	&\;\;\;+\left[(1-P(b)+P(a))\log\;\pi_{2}+\frac{1}{\beta}\int_{x \notin (a,b)} f^{\beta}(x,\mu_{2},\sigma^{2}_{2})p(x)\;dx-\frac{1-P(b)+P(a)}{(2\pi)^{\frac{\beta}{2}}\sigma^{\beta}_{2}(1+\beta)^{\frac{3}{2}}}\right].
 	\end{align*}
 	To find the estimator of $\mu_{j}$ and $\sigma^{2}_{j}$, we have to differentiate the above with respect to the respective parameters. Differentiating $L_{\beta}(\boldsymbol{\theta},P)$ with respect to $\mu_{1}$ and $\mu_{2}$ gives,
 	\begin{align*}
 	&\int_{a}^{b}f^{\beta}(x,\mu_{1},\sigma^{2}_{1})(x-\mu_{1})p(x)\;dx=0,\\
 	&\int_{x \notin (a,b)}f^{\beta}(x,\mu_{2},\sigma^{2}_{2})(x-\mu_{2})p(x)\;dx=0.
 	\end{align*}
 	And differentiating $L_{\beta}(\boldsymbol{\theta},P)$ with respect to $\sigma^{2}_{1}$ and $\sigma^{2}_{2}$ gives,
 	\begin{align*}
 	&\int_{a}^{b}f^{\beta}(x,\mu_{1},\sigma^{2}_{1})\left(\frac{(x-\mu_{1})^{2}}{2\sigma^{2}_{1}}-1\right)p(x)\;dx+\frac{\beta(P(b)-P(a))}{2(2\pi)^{\frac{\beta}{2}}(\sigma^{2}_{1})^{1+\frac{\beta}{2}}(1+\beta)^{\frac{3}{2}}}=0,\\
 	&\int_{x \notin (a,b)}f^{\beta}(x,\mu_{2},\sigma^{2}_{2})\left(\frac{(x-\mu_{2})^{2}}{2\sigma^{2}_{2}}-1\right)p(x)\;dx+\frac{\beta(1-P(b)+P(a))}{2(2\pi)^{\frac{\beta}{2}}(\sigma^{2}_{2})^{1+\frac{\beta}{2}}(1+\beta)^{\frac{3}{2}}}=0.
 	\end{align*}
 	Additionally,
 	\begin{align*}
 	\pi_{1}(P)=\int_{a(P)}^{b(P)}p(x)\;dx
 	\end{align*} 
 	with $p(x)$ being the density function corresponding to $P(x)$. So, the functional $\boldsymbol{\theta}_{\beta}(P)=\Big(\pi_{1}(P),\pi_{2}(P),\mu_{1}(P),
 	\\ \mu_{2}(P),$
 	$\sigma_{1}^{2}(P),\sigma_{2}^{2}(P)\Big)$ can be implicitly described through the following system of equations.
 	\begin{equation*}
 	\centering
 	\begin{array}{l}
 	\pi_{1}(P)=\int_{a(P)}^{b(P)}p(x)\;dx,\\
 	\pi_{1}(P)+\pi_{2}(P)=1,\\
 	D_{1}(c,\boldsymbol{\theta}_{\beta}(P))=D_{2}(c,\boldsymbol{\theta}_{\beta}(P))\; \text{for}\; c=a(P)\; \text{and}\; b(P),\\
 	\int_{a}^{b}f^{\beta}(x,\mu_{1},\sigma^{2}_{1})(x-\mu_{1})p(x)\;dx=0,\\
 	\int_{x \notin (a,b)}f^{\beta}(x,\mu_{2},\sigma^{2}_{2})(x-\mu_{2})p(x)\;dx=0,\\
 	\int_{a}^{b}f^{\beta}(x,\mu_{1},\sigma^{2}_{1})\left(\frac{(x-\mu_{1})^{2}}{2\sigma^{2}_{1}}-1\right)p(x)\;dx+\frac{\beta(P(b)-P(a))}{2(2\pi)^{\frac{\beta}{2}}(\sigma^{2}_{1})^{1+\frac{\beta}{2}}(1+\beta)^{\frac{3}{2}}}=0,\\
 	\int_{x \notin (a,b)}f^{\beta}(x,\mu_{2},\sigma^{2}_{2})\left(\frac{(x-\mu_{2})^{2}}{2\sigma^{2}_{2}}-1\right)p(x)\;dx+\frac{\beta(1-P(b)+P(a))}{2(2\pi)^{\frac{\beta}{2}}(\sigma^{2}_{2})^{1+\frac{\beta}{2}}(1+\beta)^{\frac{3}{2}}}=0
 	\end{array}
 	\end{equation*}
 	which is exactly the same system described in system (\ref{Eq*}) of equations.

 	\subsection{Derivation of the Influence Functions}
 	\label{appen3.2}
 	Consider the set-up and notation of Section \ref{SEC:IF}. 
 	Recall, we have assumed that, $\frac{M}{m}<c$ and $m>c_{1}$. In case of $p=1$,
 	\begin{align*}
 	\frac{M}{m}=\frac{\text{max}\{\sigma^{2}_{1}(P),\sigma^{2}_{2}(P)\}}{\text{min}\{\sigma^{2}_{1}(P),\sigma^{2}_{2}(P)\}}<c.
 	\end{align*}
 	The strict inequality is assumed to confirm that, the same constraint also holds in case of contaminated distribution, that is,
 	\begin{align*}
 	\frac{M_{\epsilon}}{m_{\epsilon}}=\frac{\text{max}\{\sigma^{2}_{1}(P_{\epsilon}),\sigma^{2}_{2}(P_{\epsilon})\}}{\text{min}\{\sigma^{2}_{1}(P_{\epsilon}),\sigma^{2}_{2}(P_{\epsilon})\}}<c.
 	\end{align*}
 	If the aforesaid eigenvalue ratio constraint does not hold for the contaminated distribution, 
 	it is not possible to derive the influence functions of our estimators which are derived under the same constraint. Similarly the second inequality $(m>c_{1})$ confirms the fact that the non-singularity constraint also holds under contamination.
 	

 	Now, to derive the influence functions of our estimators, let us introduce the following notations. 
 	Our functionals are  $\left(\pi_{1}(P),\pi_{2}(P), a(P), b(P), \mu_{1}(P),\mu_{2}(P), \sigma_{1}^{2}(P),\sigma_{2}^{2}(P)\right)$ which satisfy the system (\ref{Eq*}) of equations.
 	and  suppose 
 	$ IF(\boldsymbol{\theta}_{\beta},P,y)=\big(IF(\pi_{1},P,y),IF(\pi_{2},P,y), IF(a,P,y),\\ IF(b,P,y), IF(\mu_{1},P,y), IF(\mu_{2},P,y), IF(\sigma^{2}_{1},P,y), IF(\sigma^{2}_{2},P,y)\big)^{'}$ 
 	be the vector of influence functions of the aforesaid functionals. 
 	Also let $\boldsymbol{\theta}_{\epsilon},a_{\epsilon}$ and $b_{\epsilon}$ are the contaminated versions of $\boldsymbol{\theta}_{\beta}(P),a(P)$ and $b(P)$ respectively. 
 	Then, we have
 	\begin{align*}
 	\pi_{1\epsilon}&=\int_{a_{\epsilon}}^{b_{\epsilon}}p_{\epsilon}(x)\;dx\\
 	&=(1-\epsilon)\int_{a_{\epsilon}}^{b_{\epsilon}}p(x)\;dx+\epsilon I(y \in (a_{\epsilon},b_{\epsilon})).
 	\end{align*}
 	Hence,
 	\begin{align*}
 	IF(\pi_{1},P,y)&=\frac{\partial \pi_{1\epsilon}}{\partial \epsilon}\Big|_{\epsilon=0}\\
 	&=\left[P(a_{0})-P(b_{0})\right]+\left[p(b_{0})IF(b,P,y)-p(a_{0})IF(a,P,y)\right]+I(y \in (a_{0},b_{0})).
 	\end{align*}
 	The equation,
 	\begin{align*}
 	\pi_{1}(P)+\pi_{2}(P)=1
 	\end{align*}
 	gives 
 	\begin{align*}
 	IF(\pi_{1},P,y)+IF(\pi_{2},P,y)=0.
 	\end{align*}
 	Next let us recall (Equation (\ref{eqq***})) that,
 	\begin{align*}
 	D_{1}(a_{\epsilon},\boldsymbol{\theta}_{\epsilon})&=D_{2}(a_{\epsilon},\boldsymbol{\theta}_{\epsilon})\\
 	&\text{and}\\
 	D_{1}(b_{\epsilon},\boldsymbol{\theta}_{\epsilon})&=D_{2}(b_{\epsilon},\boldsymbol{\theta}_{\epsilon}).
 	\end{align*}
 	Differentiating the above equations with respect to $\epsilon$ at $0$ gives,
 	\begin{align*}
 	&2\left[\frac{IF(\pi_{1},P,y)}{\pi_{10}}-\frac{IF(\pi_{2},P,y)}{\pi_{20}}\right]+\left[\frac{IF(\sigma^{2}_{2},P,y)}{\sigma^{2}_{20}}\left(1-\frac{(a_{0}-\mu_{20})^{2}}{\sigma^{2}_{20}}\right)-\frac{IF(\sigma^{2}_{1},P,y)}{\sigma^{2}_{10}}\left(1-\frac{(a_{0}-\mu_{10})^{2}}{\sigma^{2}_{10}}\right)\right]\\
 	&=2IF(a,P,y)\left[\frac{(a_{0}-\mu_{10})}{\sigma^{2}_{10}}-\frac{(a_{0}-\mu_{20})}{\sigma^{2}_{20}}\right]+2\left[\frac{(a_{0}-\mu_{20})IF(\mu_{2},P,y)}{\sigma^{2}_{20}}-\frac{(a_{0}-\mu_{10})IF(\mu_{1},P,y)}{\sigma^{2}_{10}}\right]
 	\end{align*}
 	and
 	\begin{align*}
 	&2\left[\frac{IF(\pi_{1},P,y)}{\pi_{10}}-\frac{IF(\pi_{2},P,y)}{\pi_{20}}\right]+\left[\frac{IF(\sigma^{2}_{2},P,y)}{\sigma^{2}_{20}}\left(1-\frac{(b_{0}-\mu_{20})^{2}}{\sigma^{2}_{20}}\right)-\frac{IF(\sigma^{2}_{1},P,y)}{\sigma^{2}_{10}}\left(1-\frac{(b_{0}-\mu_{10})^{2}}{\sigma^{2}_{10}}\right)\right]\\
 	&=2IF(b,P,y)\left[\frac{(b_{0}-\mu_{10})}{\sigma^{2}_{10}}-\frac{(b_{0}-\mu_{20})}{\sigma^{2}_{20}}\right]+2\left[\frac{(b_{0}-\mu_{20})IF(\mu_{2},P,y)}{\sigma^{2}_{20}}-\frac{(b_{0}-\mu_{10})IF(\mu_{1},P,y)}{\sigma^{2}_{10}}\right].
 	\end{align*}
 	Let us observe that (from system (\ref{Eq*}) of equations),
 	\begin{align*}
 	&\;\;\;\;\;\;\int_{a_{\epsilon}}^{b_{\epsilon}}f^{\beta}(x,\mu_{1\epsilon},\sigma^{2}_{1\epsilon})(x-\mu_{1\epsilon})p_{\epsilon}(x)\;dx=0\\
 	&\implies (1-\epsilon)\int_{a_{\epsilon}}^{b_{\epsilon}}f^{\beta}(x,\mu_{1\epsilon},\sigma^{2}_{1\epsilon})(x-\mu_{1\epsilon})p(x)\;dx+\epsilon f^{\beta}(y,\mu_{1\epsilon},\sigma^{2}_{1\epsilon})(y-\mu_{1\epsilon})I(y \in(a_{\epsilon},b_{\epsilon}))=0.
 	\end{align*}
 	Differentiating the above with respect to $\epsilon$ at $0$ gives,
 	\begin{align*}
 	-\int_{a}^{b}f^{\beta}(x,\mu_{1},\sigma^{2}_{1})(x-\mu_{1})p(x)\;dx+\frac{\partial}{\partial \epsilon}\int_{a_{\epsilon}}^{b_{\epsilon}}f^{\beta}(x,\mu_{1\epsilon},\sigma^{2}_{1\epsilon})(x-\mu_{1\epsilon})p(x)\;dx\Big |_{\epsilon=0}+f^{\beta}(y,\mu_{1},\sigma^{2}_{1})(y-\mu_{1})I(y \in (a,b))=0.
 	\end{align*}
 	To evaluate the middle term in the above equation we use the Leibniz integral rule (Intermediate Calculus $(1985)$ \cite{Liebnitzrule}) as follows.
 	\begin{align*}
 	&\frac{\partial}{\partial \epsilon}\int_{a_{\epsilon}}^{b_{\epsilon}}f^{\beta}(x,\mu_{1\epsilon},\sigma^{2}_{1\epsilon})(x-\mu_{1\epsilon})p(x)\;dx\Big |_{\epsilon=0}\\
 	&=f^{\beta}(b,\mu_{1},\sigma^{2}_{1})(b-\mu_{1})p(b)-f^{\beta}(a,\mu_{1},\sigma^{2}_{1})(a-\mu_{1})p(a)+\int_{a}^{b}\frac{\partial}{\partial \epsilon}f^{\beta}(x,\mu_{1\epsilon},\sigma^{2}_{1\epsilon})(x-\mu_{1\epsilon})p(x)\Big |_{\epsilon=0}\;dx.
 	\end{align*}
 	The calculation of $\frac{\partial}{\partial \epsilon}f^{\beta}(x,\mu_{1\epsilon},\sigma^{2}_{1\epsilon})(x-\mu_{1\epsilon})p(x)\Big |_{\epsilon=0}$ is straightforward and it can be easily observed that,
 	this integration will be a linear combination of $IF(\mu_{1},P,y)$ and $IF(\sigma^{2}_{1},P,y)$. The rest of the linear equations can be similarly derived as of those in the system (\ref{Eq*}) of equations. The derivation of these influence function in other possibilities (Figure \ref{table 1}) are similar and hence omitted.

 	These calculations finally lead to the following system of linear equations defining the required influence function $IF(\boldsymbol{\theta}_{\beta},P,y)$.
 	\begin{align*}
 	A_{\beta}(\boldsymbol{\theta}_{0},a_{0},b_{0})IF(\boldsymbol{\theta}_{\beta},P,y)=B_{\beta}(y,\boldsymbol{\theta}_{0},a_{0},b_{0}),
 	\nonumber
 	\end{align*}
 	where $\boldsymbol{\theta}_{0}=(\pi_{10},\pi_{20},\mu_{10},\mu_{20},\sigma^{2}_{10},\sigma^{2}_{20}),a_{0},b_{0}$ are the true values of $\boldsymbol{\theta}, a$ and $b$, respectively
 	and 
 	\begin{align*}
 	&B_{\beta}(y,\boldsymbol{\theta}_{0},a_{0},b_{0})=\Bigg(P(a_0)-P(b_0)+I(a_0<y<b_0),0,0,0,\int_{a_0}^{b_0}(x-\mu_{10})f^\beta(x,\mu_{10},\sigma^{2}_{10})p(x)\;dx-\\
 	&(y-\mu_{10})f^\beta(y,\mu_{10},\sigma^{2}_{10})I(a_0<y<b_0),\int_{x \notin (a_0,b_0)}(x-\mu_{20})f^\beta(x,\mu_{20},\sigma^{2}_{20})p(x)\;dx-
 	(y-\mu_{20})f^\beta(y,\mu_{20},\sigma^{2}_{20})\\
 	&I(y \notin (a_0,b_0)),\int_{a_0}^{b_0}f^\beta(x,\mu_{10},\sigma^{2}_{10})\left(\dfrac{(x-\mu_{10})^2}{\sigma^{2}_{10}}-1\right)p(x)\;dx-f^\beta(y,\mu_{10},\sigma^{2}_{10})\left(\dfrac{(y-\mu_{10})^2}{\sigma^{2}_{10}}-1\right)I(a_0<y<b_0),\\
 	&\int_{x \notin (a_0,b_0)}f^\beta(x,\mu_{20},\sigma^{2}_{20})\left(\dfrac{(x-\mu_{20})^2}{\sigma^{2}_{20}}-1\right)p(x)\;dx-f^\beta(y,\mu_{20},\sigma^{2}_{20})\left(\dfrac{(y-\mu_{20})^2}{\sigma^{2}_{20}}-1\right)I(y \notin (a_0,b_0)\Bigg).
 	\end{align*}
 	The $8 \times 8$ matrix $A_{\beta}(\boldsymbol{\theta}_{0},a_{0},b_{0})$ has the $j$-th row as $A_{j*}$, for $j=1, \ldots, 8$, where
 	\begin{align*}
 	&A_{1*}=\left(1,0,p(a_0),-p(b_0),0,0,0,0 \right),\\
 	&A_{2*}=\left(1,1,0,0,0,0,0,0 \right),\\
 	&A_{3*}=\Bigg(\dfrac{2}{\pi_{10}},\dfrac{-2}{\pi_{20}},-2\left(\dfrac{(a_0-\mu_{10})}{\sigma^{2}_{10}}-\dfrac{(a_0-\mu_{20})}{\sigma^{2}_{20}}\right),0,
 	\dfrac{(a_0-\mu_{10})}{\sigma^{2}_{10}},-\dfrac{(a_0-\mu_{20})}{\sigma^{2}_{20}},\\
 	&\dfrac{(a_0-\mu_{10})^2}{\sigma^{4}_{10}}-\dfrac{1}{\sigma^{2}_{10}},\dfrac{1}{\sigma^{2}_{20}}-\dfrac{(a_0-\mu_{20})^2}{\sigma^{4}_{20}}
 	\Bigg),\\
 	&A_{4*}=\Bigg(\dfrac{2}{\pi_{10}},\dfrac{-2}{\pi_{20}},0,-2\left(\dfrac{(b_0-\mu_{10})}{\sigma^{2}_{10}}-\dfrac{(b_0-\mu_{20})}{\sigma^{2}_{20}}\right),
 	\dfrac{(b_0-\mu_{10})}{\sigma^{2}_{10}},-\dfrac{(b_0-\mu_{20})}{\sigma^{2}_{20}},\\
 	&\dfrac{(b_0-\mu_{10})^2}{\sigma^{4}_{10}}-\dfrac{1}{\sigma^{2}_{10}},\dfrac{1}{\sigma^{2}_{20}}-\dfrac{(b_0-\mu_{20})^2}{\sigma^{4}_{20}}
 	\Bigg),\\
 	&A_{5*}=\Bigg(0,0,-f^\beta(a_0,\mu_{10},\sigma^{2}_{10})(a_0-\mu_{10})p(a_0),f^\beta(b_0,\mu_{10},\sigma^{2}_{10})(b_0-\mu_{10})p(b_0),C_3,0,C_5,0\Bigg),\\
 	&A_{6*}=\Bigg(0,0,-f^\beta(a_0,\mu_{20},\sigma^{2}_{20})(a_0-\mu_{20})p(a_0),f^\beta(b_0,\mu_{20},\sigma^{2}_{20})(b_0-\mu_{20})p(b_0),0,C_4,0,C_6\Bigg),
 	\end{align*}
 	\begin{align*}
 	&A_{7*}=\Bigg(0,0,\dfrac{\beta p(a_0)}{C_0\sigma^{2+\beta}_{10}}-f^\beta(a_0,\mu_{10},\sigma^{2}_{10})\left(\dfrac{(a_0-\mu_{10})^2}{\sigma^{2}_{10}}-1\right)p(a_0),-\dfrac{\beta p(b_0)}{C_0\sigma^{2+\beta}_{10}}\\
 	&+f^\beta(b_0,\mu_{10},\sigma^{2}_{10}\left(\dfrac{(b_0-\mu_{10})^2}{\sigma^{2}_{10}}-1\right)p(b_0),\dfrac{-2C_1}{\sigma^{2}_{10}+2C_5},0,C_8-C_7+(P(a_0)-P(b_0))\dfrac{\beta}{C_0\sigma^{(4+2\beta)}_{10}},0\Bigg),\\
 	&A_{8*}=\Bigg(0,0,\dfrac{\beta p(a_0)}{C_0\sigma^{2+\beta}_{20}}-f^\beta(a_0,\mu_{20},\sigma^{2}_{20})\left(\dfrac{(a_0-\mu_{20})^2}{\sigma^{2}_{20}}-1\right)p(a_0),-\dfrac{\beta p(b_0)}{C_0\sigma^{2+\beta}_{20}}\\
 	&+f^\beta(b_0,\mu_{20},\sigma^{2}_{20}\left(\dfrac{(b_0-\mu_{20})^2}{\sigma^{2}_{20}}-1\right)p(b_0),0,\dfrac{2C_2}{\sigma^{2}_{10}-2C_6},0,C_9-C_{10}+(1-P(a_0)+P(b_0))\dfrac{\beta}{C_0\sigma^{(4+2\beta)}_{20}}\Bigg)
 	\end{align*}
 	and
 	\begin{align*}
 	&C_0=\dfrac{\beta}{2(2\pi)^{\beta/2}(1+\beta)^{1.5}},\\
 	&C_1=\int_{a_0}^{b_0}(x-\mu_{10})f^\beta(x,\mu_{10},\sigma^{2}_{10})p(x)\;dx,\\
 	&C_2=\int_{x \notin (a_0,b_0)}(x-\mu_{20})f^\beta(x,\mu_{20},\sigma^{2}_{20})p(x)\;dx,\\
 	&C_3=\int_{a_0}^{b_0}f^\beta(x,\mu_{10},\sigma^{2}_{10})\left(\dfrac{(x-\mu_{10})^2}{\sigma^{2}_{10}}-1\right)p(x)\;dx,\\
 	&C_4=\int_{x \notin (a_0,b_0)}f^\beta(x,\mu_{20},\sigma^{2}_{20})\left(\dfrac{(x-\mu_{20})^2}{\sigma^{2}_{20}}-1\right)p(x)\;dx,\\
 	&C_5=\int_{a_0}^{b_0}\dfrac{1}{2}p(x)f^{\beta}(x,\mu_{10},\sigma^{2}_{10})\left(\dfrac{(x-\mu_{10})^3}{\sigma^{4}_{10}}-\dfrac{x-\mu_{10}}{\sigma^{2}_{10}}\right)\;dx,\\
 	&C_6=\int_{x \notin (a_0,b_0)}\dfrac{1}{2}p(x)f^{\beta}(x,\mu_{20},\sigma^{2}_{20})\left(\dfrac{(x-\mu_{20})^3}{\sigma^{4}_{20}}-\dfrac{x-\mu_{20}}{\sigma^{2}_{20}}\right)\;dx,\\
 	&C_7=\int_{a_0}^{b_0}p(x)f^{\beta}(x,\mu_{10},\sigma^{2}_{10})\dfrac{(x-\mu_{10})^2}{\sigma^{4}_{10}}\;dx,\\
 	&C_8=\int_{a_0}^{b_0}p(x)f^{\beta}(x,\mu_{10},\sigma^{2}_{10})\left(\dfrac{(x-\mu_{10})^2}{\sigma^{2}_{10}}-1\right)^2\dfrac{1}{2\sigma^{2}_{10}}\;dx,\\
 	&C_9=\int_{x \notin (a_0,b_0)}p(x)f^{\beta}(x,\mu_{20},\sigma^{2}_{20})\dfrac{(x-\mu_{20})^2}{\sigma^{4}_{20}}\;dx,\\
 	&C_{10}=\int_{x \notin (a_0,b_0)}p(x)f^{\beta}(x,\mu_{20},\sigma^{2}_{20})\left(\dfrac{(x-\mu_{20})^2}{\sigma^{2}_{20}}-1\right)^2\dfrac{1}{2\sigma^{2}_{20}}\;dx.\\
 	\end{align*}
 	\subsection{Influence of a single contamination on the MDPDEs}
 	
 	To illustrate the effect of moving a single sample observation too far from the original data cloud on the resulting MDPDEs of a mixture normal model, we have taken a random sample of size $99$ and dimension $2$ from a $3$-component normal mixture with component means $(0,0)^{t}$, $(5,5)^{t}$, $(-5,-5)^{t}$, dispersions $I_{2}$ (for all the components) and weights $0.33$, $0.33$ and $0.34$. Now, we contaminate this sample with a single additional observation $(\delta,\delta)^{t}$, with $\delta$ running from $-15$ to $15$ in intervals of $0.5$, and calculate the (summed) squared $L_2$ norms of the biases of the component mean estimates (corresponding to $\beta=0.1$). This summed squared norm bias is plotted in Figure \ref{ifc} against $\delta$. It may be observed that as $\delta$ becomes too large or too small, in which case the additional observation becomes a clearly incongruent observation in relation to the rest of the data, the effect of this outlier actually vanishes, and the summed squared norm of the biases settles down on what one would have had (represented by the horizontal dashed line) if that observation was just not there in the sample. This is a consequence of the fact that for very distant observations, the strong density-power downweighting 
 	applied by the minimum DPD method makes the contribution of the corresponding term practically vanish. 
 	
 	\renewcommand{\thefigure}{S3}
 	\begin{figure}[h!]
 		\centering
 		\includegraphics[height=4 in,width= 4 in]{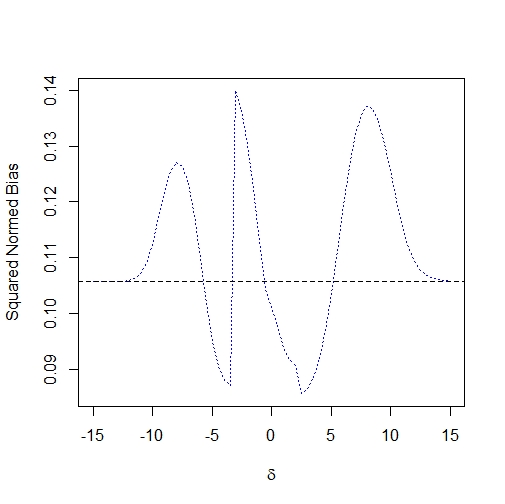}
 		\caption{The summed squared bias of the MDPDEs of the component means from contaminated samples (blue) and the horizontal red line corresponds to the bias of the MDPDEs of the component means based on the original sample.}
 		\label{ifc}
 	\end{figure}
 	
 	\newpage 
 	\section{Bias and Mean Squared Errors of the Cluster Means }
 	\label{appen4}
 	The bias ($L_{2}$ norm) and the mean squared errors of the cluster means for pure, uniformly (from chi-squared) contaminated, uniformly (from annulus) contaminated, outlying cluster contaminated and datasets with differentially dispersed clusters are provided in Tables \ref{table3.1}, \ref{table3.15}, \ref{table3.2}, \ref{table3.3} and \ref{table3.4}, respectively.
 	\begin{table}[h!]
 		\hspace{-0.5in}
 		\small
 		\begin{tabular}{c c c c c c c c c c c c c} 
 			\hline
 			&       \multicolumn{4}{c}{\hspace{8em}MPLE$_{\beta}$} & & \multicolumn{2}{c}{TCLUST} & \multicolumn{2}{c}{TKMEANS} & KMEDOIDS   & MCLUST\\ [1ex] 
 			$p$ &  $\boldsymbol{\Sigma}$  &  $\beta=0$  & $\beta=0.1$ & $\beta=0.3$ & $\beta=0.5$ & $\alpha=0.0$ &  $\alpha=0.05$ & $\alpha=0.0$  & $\alpha=0.05$ &  & \\ [1ex] 
 			\hline
 			2	&	$\boldsymbol{I}_{2}$	&		0.014		&		0.016		&		0.019		&		0.021		&		0.014		&		0.019		&		0.019		&		0.021		&		0.014		& 0.030\\
 			&		&	$(	0.018	)$	&	$(	0.018	)$	&	$(	0.019	)$	&	$(	0.022	)$	&	$(	0.018	)$	&	$(	0.021	)$	&	$(	0.019	)$	&	$(	0.024	)$	&	$(	0.043	)$ & (0.018)	\\
 			&	$3\boldsymbol{I}_{2}$	&		0.030		&		0.037		&		0.054		&		0.063		&		0.032		&		0.094		&		0.048		&		0.058		&		0.043		& 0.029\\
 			&		&	$(	0.076	)$	&	$(	0.077	)$	&	$(	0.081	)$	&	$(	0.087	)$	&	$(	0.078	)$	&	$(	0.091	)$	&	$(	0.068	)$	&	$(	0.087	)$	&	$(	0.133	)$ & (0.060)	\\
 			&	$5\boldsymbol{I}_{2}$	&		0.086		&		0.110		&		0.163		&		0.194		&		0.181  &  0.142		&		0.168		&		0.060		&		0.126		& 0.082\\
 			&		&	$(	0.279	)$	&	$(	0.379	)$	&	$(	0.296	)$	&	$(	0.310	)$	&	(0.219) & (0.320)	&	$(	0.154	)$	&	$(	0.178	)$	&	$(	0.288	)$ & (0.146)\\
 			&	&	&								\\
 			4	&	$\boldsymbol{I}_{4}$	&		0.024		&		0.022		&		0.021		&		0.023		&		0.031		&		0.029		&		0.023		&		0.026		&		0.097	&  0.028	\\
 			&		&	$(	0.035	)$	&	$(	0.037	)$	&	$(	0.043	)$	&	$(	0.053	)$	&	$(	0.037	)$	&	$(	0.044	)$	&	$(	0.034	)$	&	$(	0.041	)$	&	$(	0.444	)$	&   (0.036)\\
 			&	$3\boldsymbol{I}_{4}$	&		0.040		&		0.040		&		0.043		&		0.046		&		0.042		&		0.057		&		0.053		&		0.059		&		0.131 & 0.044	\\
 			&		&	$(	0.105	)$	&	$(	0.110	)$	&	$(	0.132	)$	&	$(	0.163	)$	&	$(	0.110	)$	&	$(	0.121	)$	&	$(	0.113	)$	&	$(	0.132	)$	&	$(	1.395	)$	&   (0.117)\\
 			&	$5\boldsymbol{I}_{4}$	&		0.067		&		0.071		&		0.081		&		0.092		&		0.075		&		0.050		&		0.094		&		0.077		&		0.205	 & 0.043	\\
 			&		&	$(	0.230	)$	&	$(	0.230	)$	&	$(	0.257	)$	&	$(	0.307	)$	&	$(	0.222	)$	&	$(	0.225	)$	&	$(	0.227	)$	&	$(	0.252	)$	&	$(	2.423	)$  & (0.197)	\\
 			&	&	&							\\
 			6	&	$\boldsymbol{I}_{6}$	&		0.025		&		0.026		&		0.028		&		0.035		&		0.038		&		0.031		&		0.035		&		0.040		&		0.193	&  0.030	\\
 			&		&	$(	0.055	)$	&	$(	0.057	)$	&	$(	0.069	)$	&	$(	0.090	)$	&	$(	0.054	)$	&	$(	0.061	)$	&	$(	0.055	)$	&	$(	0.061	)$	&	$(	1.612	)$	 &  (0.056)\\
 			&	$3\boldsymbol{I}_{6}$	&		0.066		&		0.067		&		0.069		&		0.072		&		0.056		&		0.063		&		0.066		&		0.073		&		0.307	 & 0.066	\\
 			&		&	$(	0.164	)$	&	$(	0.169	)$	&	$(	0.203	)$	&	$(	0.167	)$	&	$(	0.184	)$	&	$(	0.195	)$	&	$(	0.174	)$	&	$(	0.201	)$	&	$(	4.923	)$ &  (0.170)	\\
 			&	$5\boldsymbol{I}_{6}$	&		0.046		&		0.048		&		0.053		&		0.063		&		0.057  &  0.091		&		0.059		&		0.073		&		0.386	 & 0.062	\\
 			&		&	$(	0.253	)$	&	$(	0.281	)$	&	$(	0.341	)$	&	$(	0.443	)$	&	(0.281) & (0.309)	&	$(	0.285	)$	&	$(	0.311	)$	&	$(	7.821	)$	 & (0.270)\\
 			&	&	&								\\
 			8	&	$\boldsymbol{I}_{8}$	&		0.040		&		0.039		&		0.043		&		0.055		&		0.044		&		0.045		&		0.044		&		0.051		&		0.162	 & 0.043\\
 			&		&	$(	0.073	)$	&	$(	0.077	)$	&	$(	0.100	)$	&	$(	0.141	)$	&	$(	0.074	)$	&	$(	0.083	)$	&	$(	0.075	)$	&	$(	0.085	)$	&	$(	3.202	)$  &  (0.075)	\\
 			&	$3\boldsymbol{I}_{8}$	&		0.058		&		0.059		&		0.068		&		0.081		&		 0.073  &  0.074		&		0.054		&		0.056		&		0.490		& 0.086\\
 			&		&	$(	0.207	)$	&	$(	0.217	)$	&	$(	0.277	)$	&	$(	0.384	)$	&	(0.218) & (0.233)	&	$(	0.217	)$	&	$(	0.242	)$	&	$(	10.196	)$  &  (0.210)\\
 			&	$5\boldsymbol{I}_{8}$	&		0.087		&		0.084		&		0.087		&		0.101		&		0.088  &  0.078		&		0.085		&		0.088		&		0.563		& 0.100\\
 			&		&	$(	0.381	)$	&	$(	0.396	)$	&	$(	0.504	)$	&	$(	0.711	)$	&	(0.382) & (0.431)	&	$(	0.369	)$	&	$(	0.415	)$	&	$(	16.049	)$  & (0.344)\\
 			&	&	&							\\
 			10	&	$\boldsymbol{I}_{10}$	&		0.054		&		0.054		&		0.057		&		0.065		&		0.043		&		0.041		&		0.040		&		0.048		&		0.324		& 0.041\\
 			&		&	$(	0.088	)$	&	$(	0.094	)$	&	$(	0.129	)$	&	$(	0.201	)$	&	$(	0.089	)$	&	$(	0.100	)$	&	$(	0.090	)$	&	$(	0.100	)$	&	$(	5.240	)$ & (0.095)\\
 			&	$3\boldsymbol{I}_{10}$	&		0.058		&		0.058		&		0.067		&		0.089		&		0.062  &  0.085		&		0.081		&		0.085		&		0.626	 & 0.062	\\
 			&		&	$(	0.258	)$	&	$(	0.272	)$	&	$(	0.375	)$	&	$(	0.600	)$	&(0.272) & (0.310)	&	$(	0.260	)$	&	$(	0.299	)$	&	$(	15.654	)$  &  (0.259)	\\
 			&	$5\boldsymbol{I}_{10}$	&		0.09		&		0.090		&		0.103		&		0.128		&		0.116  &  0.115		&		0.071		&		0.070		&		0.720	&   0.082	\\
 			&		&	$(	0.498	)$	&	$(	0.523	)$	&	$(	0.704	)$	&	$(	1.112	)$	&	(0.492) & (0.498)	&	$(	0.482	)$	&	$(	0.531	)$	&	$(	26.962	)$ &  (0.418)	\\[2ex] 
 			\hline
 		\end{tabular}
 		\caption{Estimated bias and mean squared errors (within parentheses) for pure datasets.}
 		\label{table3.1}
 	\end{table}
 	
 	\begin{table}[h!]
 		\hspace{-0.5in}
 		\small
 		\begin{tabular}{c c c c c c c c c c c c} 
 			\hline
 			&       \multicolumn{4}{c}{\hspace{8em}MPLE$_{\beta}$} & & \multicolumn{2}{c}{TCLUST} & \multicolumn{2}{c}{TKMEANS} & KMEDOIDS & MCLUST \\ [1ex] 
 			$p$ &  $\boldsymbol{\Sigma}$  &  $\beta=0$  & $\beta=0.1$ & $\beta=0.3$ & $\beta=0.5$ & $\alpha=0.1$ &  $\alpha=0.15$ & $\alpha=0.1$  & $\alpha=0.15$ & \\ [1ex] 
 			\hline
 			2	&	$\boldsymbol{I}_{2}$	&		0.101		&		0.041		&		0.013		&		0.014		&		0.022		&		0.041		&		0.016		&		0.018		&		0.077		&		0.016		\\
 			&		&	$(	0.076	)$	&	$(	0.036	)$	&	$(	0.023	)$	&	$(	0.024	)$	&	$(	0.024	)$	&	$(	0.026	)$	&	$(	0.020	)$	&	$(	0.022	)$	&	$(	0.107	)$	&	$(	0.025	)$	\\
 			&	$3\boldsymbol{I}_{2}$	&		0.106		&		0.076		&		0.039		&		0.039		&		0.044		&		0.027		&		0.068		&		0.078		&		0.093		&		0.056		\\
 			&		&	$(	0.154	)$	&	$(	0.129	)$	&	$(	0.095	)$	&	$(	0.091	)$	&	$(	0.100	)$	&	$(	0.089	)$	&	$(	0.094	)$	&	$(	0.097	)$	&	$(	0.501	)$	&	$(	0.086	)$	\\
 			&	$5\boldsymbol{I}_{2}$	&		0.159		&		0.110		&		0.061		&		0.051		&		0.126		&		0.144		&		0.129		&		0.066		&		0.056		&		0.194		\\
 			&		&	$(	0.309	)$	&	$(	0.275	)$	&	$(	0.245	)$	&	$(	0.261	)$	&	$(	0.352	)$	&	$(	0.310	)$	&	$(	0.253	)$	&	$(	0.249	)$	&	$(	1.193	)$	&	$(	0.228	)$	\\
 			&	&	&				&				&				&				&				&				&				&				&							\\
 			4	&	$\boldsymbol{I}_{4}$	&		0.058		&		0.028		&		0.027		&		0.030		&		0.020		&		0.033		&		0.033		&		0.028		&		0.120		&		0.028		\\
 			&		&	$(	0.060	)$	&	$(	0.037	)$	&	$(	0.042	)$	&	$(	0.052	)$	&	$(	0.046	)$	&	$(	0.051	)$	&	$(	0.039	)$	&	$(	0.044	)$	&	$(	0.471	)$	&	$(	0.040	)$	\\
 			&	$3\boldsymbol{I}_{4}$	&		0.061		&		0.031		&		0.033		&		0.043		&		0.044		&		0.046		&		0.065		&		0.060		&		0.166		&		0.051		\\
 			&		&	$(	0.206	)$	&	$(	0.127	)$	&	$(	0.119	)$	&	$(	0.143	)$	&	$(	0.124	)$	&	$(	0.145	)$	&	$(	0.119	)$	&	$(	0.132	)$	&	$(	1.464	)$	&	$(	0.119	)$	\\
 			&	$5\boldsymbol{I}_{4}$	&		0.114		&		0.092		&		0.081		&		0.080		&		0.054		&		0.073		&		0.062		&		0.079		&		0.295		&		0.071		\\
 			&		&	$(	0.352	)$	&	$(	0.272	)$	&	$(	0.256	)$	&	$(	0.301	)$	&	$(	0.252	)$	&	$(	0.274	)$	&	$(	0.253	)$	&	$(	0.265	)$	&	$(	2.685	)$	&	$(	0.256	)$	\\
 			&	&	&				&				&				&				&				&				&				&				&							\\
 			6	&	$\boldsymbol{I}_{6}$	&		0.049		&		0.037		&		0.039		&		0.046		&		0.037		&		0.033		&		0.037		&		0.036		&		0.152		&		0.032		\\
 			&		&	$(	0.173	)$	&	$(	0.065	)$	&	$(	0.078	)$	&	$(	0.099	)$	&	$(	0.069	)$	&	$(	0.072	)$	&	$(	0.066	)$	&	$(	0.072	)$	&	$(	1.697	)$	&	$(	0.061	)$	\\
 			&	$3\boldsymbol{I}_{6}$	&		0.066		&		0.054		&		0.057		&		0.063		&		0.067		&		0.059		&		0.048		&		0.054		&		0.394		&		0.061		\\
 			&		&	$(	0.310	)$	&	$(	0.216	)$	&	$(	0.247	)$	&	$(	0.307	)$	&	$(	0.187	)$	&	$(	0.207	)$	&	$(	0.187	)$	&	$(	0.222	)$	&	$(	5.234	)$	&	$(	0.186	)$	\\
 			&	$5\boldsymbol{I}_{6}$	&		0.097		&		0.094		&		0.102		&		0.116		&		0.088		&		0.073		&		0.068		&		0.067		&		0.367		&		0.074		\\
 			&		&	$(	0.438	)$	&	$(	0.364	)$	&	$(	0.421	)$	&	$(	0.532	)$	&	$(	0.314	)$	&	$(	0.367	)$	&	$(	0.312	)$	&	$(	0.335	)$	&	$(	7.889	)$	&	$(	0.315	)$	\\
 			&	&	&				&				&				&				&				&				&				&				&							\\
 			8	&	$\boldsymbol{I}_{8}$	&		0.094		&		0.030		&		0.033		&		0.038		&		0.038		&		0.039		&		0.044		&		0.034		&		0.284		&		0.051		\\
 			&		&	$(	0.220	)$	&	$(	0.088	)$	&	$(	0.112	)$	&	$(	0.156	)$	&	$(	0.092	)$	&	$(	0.088	)$	&	$(	0.088	)$	&	$(	0.089	)$	&	$(	3.513	)$	&	$(	0.092	)$	\\
 			&	$3\boldsymbol{I}_{8}$	&		0.086		&		0.062		&		0.061		&		0.080		&		0.050		&		0.073		&		0.089		&		0.085		&		0.452		&		0.082		\\
 			&		&	$(	0.363	)$	&	$(	0.260	)$	&	$(	0.336	)$	&	$(	0.467	)$	&	$(	0.257	)$	&	$(	0.280	)$	&	$(	0.278	)$	&	$(	0.282	)$	&	$(	10.232	)$	&	$(	0.257	)$	\\
 			&	$5\boldsymbol{I}_{8}$	&		0.138		&		0.113		&		0.118		&		0.133		&		0.109		&		0.197		&		0.103		&		0.103		&		0.655		&		0.085		\\
 			&		&	$(	0.525	)$	&	$(	0.411	)$	&	$(	0.510	)$	&	$(	0.715	)$	&	$(	0.411	)$	&	$(	0.444	)$	&	$(	0.429	)$	&	$(	0.454	)$	&	$(	17.489	)$	&	$(	0.413	)$	\\
 			&	&	&				&				&				&				&				&				&				&				&							\\
 			10	&	$\boldsymbol{I}_{10}$	&		0.066		&		0.040		&		0.046		&		0.055		&		0.042		&		0.046		&		0.054		&		0.051		&		0.267		&		0.046		\\
 			&		&	$(	0.308	)$	&	$(	0.107	)$	&	$(	0.140	)$	&	$(	0.213	)$	&	$(	0.111	)$	&	$(	0.113	)$	&	$(	0.108	)$	&	$(	0.110	)$	&	$(	5.576	)$	&	$(	0.104	)$	\\
 			&	$3\boldsymbol{I}_{10}$	&		0.086		&		0.068		&		0.078		&		0.103		&		0.089		&		0.086		&		0.088		&		0.088		&		0.474		&		0.075		\\
 			&		&	$(	0.437	)$	&	$(	0.306	)$	&	$(	0.401	)$	&	$(	0.613	)$	&	$(	0.318	)$	&	$(	0.319	)$	&	$(	0.329	)$	&	$(	0.345	)$	&	$(	16.821	)$	&	$(	0.291	)$	\\
 			&	$5\boldsymbol{I}_{10}$	&		0.122		&		0.115		&		0.129		&		0.168		&		0.085		&		0.090		&		0.119		&		0.118		&		0.819		&		0.112		\\
 			& & $(	0.645	)$	&	$(	0.537	)$	&	$(	0.720	)$	&	$(	1.101	)$	&	$(	0.545	)$	&	$(	0.570	)$	&	$(	0.497	)$	&	$(	0.554	)$	&	$(	27.373	)$	&	$(	0.511	)$	\\	[2ex] 
 			\hline
 		\end{tabular}
 		\caption{Estimated bias and mean squared errors (within parentheses) for uniformly (chi-squared method) contaminated datasets.}
 		\label{table3.15}
 	\end{table}
 	\begin{table}[h!]
 		\hspace{-0.5in}
 		\small
 		\begin{tabular}{c c c c c c c c c c c c} 
 			\hline
 			&       \multicolumn{4}{c}{\hspace{8em}MPLE$_{\beta}$} & & \multicolumn{2}{c}{TCLUST} & \multicolumn{2}{c}{TKMEANS} & KMEDOIDS & MCLUST \\ [1ex] 
 			$p$ &  $\boldsymbol{\Sigma}$  &  $\beta=0$  & $\beta=0.1$ & $\beta=0.3$ & $\beta=0.5$ & $\alpha=0.1$ &  $\alpha=0.15$ & $\alpha=0.1$  & $\alpha=0.15$ & \\ [1ex] 
 			\hline
 			2	&	$\boldsymbol{I}_{2}$	&		0.794		&		0.035		&		0.038		&		0.040		&		0.072		&		0.022		&		0.051		&		0.021		&		0.135	&    0.014	\\
 			&		&	$(	0.810	)$	&	$(	0.021	)$	&	$(	0.023	)$	&	$(	0.026	)$	&	$(	0.059	)$	&	$(	0.023	)$	&	$(	0.030	)$	&	$(	0.027	)$	&	$(	0.125	)$ &   (0.020)	\\
 			&	$3\boldsymbol{I}_{2}$	&		0.697		&		0.174		&		0.061		&		0.076		&		0.514		&		0.060		&		0.123		&		0.061		&		0.237	&    0.064	\\
 			&		&	$(	0.749	)$	&	$(	0.255	)$	&	$(	0.085	)$	&	$(	0.093	)$	&(1.741) & (0.109)	&	$(	0.098	)$	&	$(	0.082	)$	&	$(	0.448	)$	&   (0.077)\\
 			&	$5\boldsymbol{I}_{2}$	&		0.778		&		0.311		&		0.254		&		0.293		&		 1.962   &  0.247		&		0.300		&		0.053		&		0.429	&    0.070	\\
 			&		&	$(	1.151	)$	&	$(	0.513	)$	&	$(	0.393	)$	&	$(	0.402	)$	&	$(	37.672	)$	&	$(	0.439	)$	&	$(	0.319	)$	&	$(	0.212	)$	&	$(	0.863	)$ &   (0.173)	\\
 			&	&	&				&						\\
 			4	&	$\boldsymbol{I}_{4}$	&		0.057		&		0.026		&		0.028		&		0.030		&		0.048		&		0.026		&		0.054		&		0.019		&		0.118	&    0.031	\\
 			&		&	$(	0.225	)$	&	$(	0.040	)$	&	$(	0.045	)$	&	$(	0.054	)$	&	$(	0.046	)$	&	$(	0.045	)$	&	$(	0.053	)$	&	$(	0.051	)$	&	$(	0.516	)$	&   (0.044)\\
 			&	$3\boldsymbol{I}_{4}$	&		0.131		&		0.078		&		0.069		&		0.076		&		0.093		&		0.069		&		0.047		&		0.043		&		0.141	&    0.045	\\
 			&		&	$(	0.317	)$	&	$(	0.132	)$	&	$(	0.145	)$	&	$(	0.171	)$	&	$(	0.132	)$	&	$(	0.134	)$	&	$(	0.129	)$	&	$(	0.137	)$	&	$(	1.485	)$ &   (0.122)	\\
 			&	$5\boldsymbol{I}_{4}$	&		0.219		&		0.094		&		0.083		&		0.088		&		0.105		&		0.067		&		0.086		&		0.059		&		0.237	&    0.084	\\
 			&		&	$(	0.524	)$	&	$(	0.269	)$	&	$(	0.246	)$	&	$(	0.277	)$	&	$(	0.248	)$	&	$(	0.292	)$	&	$(	0.241	)$	&	$(	0.262	)$	&	$(	2.275	)$ &   (0.248)	\\
 			&	&	&				&							\\
 			6	&	$\boldsymbol{I}_{6}$	&		0.050		&		0.045		&		0.049		&		0.056		&		0.036		&		0.041		&		0.036		&		0.036		&		0.136	&    0.037	\\
 			&		&	$(	0.250	)$	&	$(	0.061	)$	&	$(	0.073	)$	&	$(	0.095	)$	&	$(	0.066	)$	&	$(	0.073	)$	&	$(	0.065	)$	&	$(	0.069	)$	&	$(	1.545	)$ &   (0.063)	\\
 			&	$3\boldsymbol{I}_{6}$	&		0.096		&		0.059		&		0.069		&		0.082		&		0.059		&		0.068		&		0.069		&		0.080		&		0.322	&    0.062	\\
 			&		&	$(	0.357	)$	&	$(	0.185	)$	&	$(	0.229	)$	&	$(	0.297	)$	&	$(	0.188	)$	&	$(	0.207	)$	&	$(	0.185	)$	&	$(	0.213	)$	&	$(	4.988	)$  &   (0.195)	\\
 			&	$5\boldsymbol{I}_{6}$	&		0.127		&		0.069		&		0.089		&		0.107		&		0.082		&		0.080		&		0.110		&		0.107		&		0.357	&    0.095	\\
 			&		&	$(	0.435	)$	&	$(	0.287	)$	&	$(	0.344	)$	&	$(	0.453	)$	&	$(	0.342	)$	&	$(	0.352	)$	&	$(	0.346	)$	&	$(	0.383	)$	&	$(	8.597	)$ &   (0.319)	\\
 			&	&	&				&					\\
 			8	&	$\boldsymbol{I}_{8}$	&		0.059		&		0.045		&		0.047		&		0.050		&		0.034		&		0.032		&		0.049		&		0.035		&		0.191	&    0.042	\\
 			&		&	$(	0.305	)$	&	$(	0.081	)$	&	$(	0.105	)$	&	$(	0.146	)$	&	$(	0.087	)$	&	$(	0.091	)$	&	$(	0.089	)$	&	$(	0.086	)$	&	$(	3.279	)$ &   (0.083)	\\
 			&	$3\boldsymbol{I}_{8}$	&		0.091		&		0.072		&		0.076		&		0.075		&		0.081		&		0.067		&		0.090		&		0.092		&		0.489	&    0.065	\\
 			&		&	$(	0.394	)$	&	$(	0.247	)$	&	$(	0.304	)$	&	$(	0.415	)$	&	$(	0.248	)$	&	$(	0.277	)$	&	$(	0.248	)$	&	$(	0.266	)$	&	$(	10.136	)$ &   (0.243)	\\
 			&	$5\boldsymbol{I}_{8}$	&		0.088		&		0.093		&		0.107		&		0.127		&		0.102		&		0.100		&		0.084		&		0.083		&		0.512	&    0.097	\\
 			&		&	$(	0.548	)$	&	$(	0.429	)$	&	$(	0.543	)$	&	$(	0.771	)$	&	$(	0.392	)$	&	$(	0.458	)$	&	$(	0.411	)$	&	$(	0.465	)$	&	$(	16.497	)$&   (0.426)	\\
 			&	&	&				&					\\
 			10	&	$\boldsymbol{I}_{10}$	&		0.076		&		0.042		&		0.051		&		0.064		&		0.048		&		0.046		&		0.055		&		0.051		&		0.305	&    0.048	\\
 			&		&	$(	0.279	)$	&	$(	0.109	)$	&	$(	0.148	)$	&	$(	0.231	)$	&	$(	0.119	)$	&	$(	0.119	)$	&	$(	0.113	)$	&	$(	0.116	)$	&	$(	5.410	)$ &   (0.010)	\\
 			&	$3\boldsymbol{I}_{10}$	&		0.093		&		0.064		&		0.073		&		0.098		&		0.069		&		0.079		&		0.093		&		0.100		&		0.539	&    0.080	\\
 			&		&	$(	0.451	)$	&	$(	0.289	)$	&	$(	0.384	)$	&	$(	0.592	)$	&	$(	0.330	)$	&	$(	0.331	)$	&	$(	0.304	)$	&	$(	0.343	)$	&	$(	16.047	)$	&   (0.303)\\
 			&	$5\boldsymbol{I}_{10}$	&		0.101		&		0.090		&		0.088		&		0.115		&		0.098		&		0.115		&		0.099		&		0.102		&		0.836	 &    0.112 	\\
 			&		&	$(	0.612	)$	&	$(	0.492	)$	&	$(	0.641	)$	&	$(	0.998	)$	&	$(	0.507	)$	&	$(	0.560	)$	&	$(	0.571	)$	&	$(	0.601	)$	&	$(	27.712	)$ &   (0.513)	\\[2ex] 
 			\hline
 		\end{tabular}
 		\caption{Estimated bias and mean squared errors (within parentheses) for uniformly (from annulus) contaminated datasets.}
 		\label{table3.2}
 	\end{table}
 	
 	\begin{table}[h!]
 		\hspace{-0.5in}
 		\small
 		\begin{tabular}{c c c c c c c c c c c c} 
 			\hline
 			&       \multicolumn{4}{c}{\hspace{8em}MPLE$_{\beta}$} & & \multicolumn{2}{c}{TCLUST} & \multicolumn{2}{c}{TKMEANS} & KMEDOIDS & MCLUST \\ [1ex] 
 			$p$ &  $\boldsymbol{\Sigma}$  &  $\beta=0$  & $\beta=0.1$ & $\beta=0.3$ & $\beta=0.5$ & $\alpha=0.1$ &  $\alpha=0.15$ & $\alpha=0.1$  & $\alpha=0.15$ & \\ [1ex] 
 			\hline
 			2	&	$\boldsymbol{I}_{2}$	&		4.870		&		0.028		&		0.029		&		0.030		&		2.232   &  0.025		&		1.424		&		0.021		&		20.888	&       22.152	\\
 			&		&	$(	23.97	)$	&	$(	0.018	)$	&	$(	0.021	)$	&	$(	0.024	)$	&	(47.236) & (0.028)	&	$(	27.977	)$	&	$(	0.027	)$	&	$(	457.109	)$ &      (497.0417)	\\
 			&	$3\boldsymbol{I}_{2}$	&		4.876		&		4.487		&		0.062		&		0.070		&		10.536  &  0.427		&		3.187		&		0.053		&		21.928	&      21.762	\\
 			&		&	$(	24.141	)$	&	$(	20.599	)$	&	$(	0.087	)$	&	$(	0.094	)$	&	(233.780)& (10.392)	&	$(	68.016	)$	&	$(	0.093	)$	&	$(	482.016	)$ &      (482.281)	\\
 			&	$5\boldsymbol{I}_{2}$	&		9.712		&		10.096		&		0.990		&		0.566		&		21.899  &  20.053		&		8.837		&		0.064		&		26.925	&       21.938	\\
 			&		&	$(	145.236	)$	&	$(	158.673	)$	&	$(	28.353	)$	&	$(	0.685	)$	&	(485.105)& (461.605)	&	$(	191.516	)$	&	$(	0.189	)$	&	$(	491.365	)$ &      (486.658)	\\
 			&	&	&				&						\\
 			4	&	$\boldsymbol{I}_{4}$	&		6.860		&		0.025		&		0.030		&		0.037		&		2.661   &  0.030		&		3.249		&		0.030		&		29.363	&      31.221	\\
 			&		&	$(	47.878	)$	&	$(	0.040	)$	&	$(	0.045	)$	&	$(	0.054	)$	&	(77.778) & (0.050)	&	$(	98.839	)$	&	$(	0.052	)$	&	$(	906.572	)$ &      (987.839)	\\
 			&	$3\boldsymbol{I}_{4}$	&		6.881		&		6.310		&		0.055		&		0.059		&		17.412  &  0.051		&		5.772		&		0.048		&		30.596	&       30.730	\\
 			&		&	$(	47.980	)$	&	$(	44.003	)$	&	$(	0.131	)$	&	$(	0.155	)$	&	(547.829)& (0.145)	&	$(	172.106	)$	&	$(	0.144	)$	&	$(	941.274	)$&      (959.393)  	\\
 			&	$5\boldsymbol{I}_{4}$	&		6.974		&		7.079		&		0.089		&		0.093		&		28.472  &  21.841		&		6.983		&		0.047		&		30.767	&     30.856	\\
 			&		&	$(	49.457	)$	&	$(	51.252	)$	&	$(	0.239	)$	&	$(	0.283	)$	&	(892.367)& (689.382)	&	$(	217.568	)$	&	$(	0.243	)$	&	$(	951.243	)$ &      (968.882)	\\
 			&	&	&				&						\\
 			6	&	$\boldsymbol{I}_{6}$	&		8.757		&		0.033		&		0.038		&		0.041		&		5.511   &  0.033		&		0.993		&		0.036		&		36.762	&       38.104  	\\
 			&		&	$(	109.060	)$	&	$(	0.064	)$	&	$(	0.076	)$	&	$(	0.097	)$	&	(204.188)& (0.067)	&	$(	31.582	)$	&	$(	0.068	)$	&	$(	1393.215	)$   &      (1472.076)	\\
 			&	$3\boldsymbol{I}_{6}$	&		8.343		&		6.886		&		0.063		&		0.065		&		15.997  &  0.067		&		1.784		&		0.054		&		37.819	 &       38.054 	\\
 			&		&	$(	75.015	)$	&	$(	66.517	)$	&	$(	0.208	)$	&	$(	0.275	)$	&	(617.990)& (0.195)	&	$(	61.955	)$	&	$(	0.196	)$	&	$(	1444.099	)$  &      (0.172)	\\
 			&	$5\boldsymbol{I}_{6}$	&		7.859		&		9.443		&		0.320		&		0.128		&		30.467  &  13.672 		&		6.331		&		0.076		&		37.696	&       0.061	\\
 			&		&	$(	75.932	)$	&	$(	90.913	)$	&	$(	6.551	)$	&	$(	0.495	)$	&	(976.222)& (526.805)	&	$(	238.451	)$	&	$(	0.356	)$	&	$(	1432.805	)$ &      (1470.286)	\\
 			&	&	&				&						\\
 			8	&	$\boldsymbol{I}_{8}$	&		14.642		&		0.042		&		0.054		&		0.059		&		9.995   &  0.044		&		4.623		&		0.042		&		42.513	&      43.710	\\
 			&		&	$(	215.859	)$	&	$(	0.086	)$	&	$(	0.108	)$	&	$(	0.150	)$	&	(430.428)& (0.088)	&	$(	188.793	)$	&	$(	0.087	)$	&	$(	1872.28	)$ &      (1937.060)	\\
 			&	$3\boldsymbol{I}_{8}$	&		14.509		&		4.952		&		0.095		&		0.117		&		 13.260  &  0.061		&		7.295		&		0.069		&		43.359	&       43.846	\\
 			&		&	$(	214.399	)$	&	$(	61.377	)$	&	$(	0.307	)$	&	$(	0.429	)$	&(591.215)& (0.272)	&	$(	321.262	)$	&	$(	0.249	)$	&	$(	1903.009	)$ &      (1951.052)	\\
 			&	$5\boldsymbol{I}_{8}$	&		14.555		&		11.787		&		0.099		&		0.114		&		 33.417  &  5.136		&		6.360		&		0.103		&		43.435	&       43.567	\\
 			&		&	$(	213.653	)$	&	$(	143.740	)$	&	$(	0.551	)$	&	$(	0.758	)$	&(1191.519)& (225.505)	&	$(	279.669	)$	&	$(	0.571	)$	&	$(	1909.009	)$ &      (1926.959)	\\
 			&	&	&				&						\\
 			10	&	$\boldsymbol{I}_{10}$	&		16.282		&		0.047		&		0.057		&		0.070		&		5.188   &  0.045		&		5.211		&		0.048		&		45.457	&       48.981	\\
 			&		&	$(	267.572	)$	&	$(	0.103	)$	&	$(	0.138	)$	&	$(	0.218	)$	&	(242.153)& (0.117)	&	$(	248.064	)$	&	$(	0.109	)$	&	$(	2232.301	)$&      (2432.958)	\\
 			&	$3\boldsymbol{I}_{10}$	&		16.293		&		7.749		&		0.076		&		0.092		&		18.751  &  0.084		&		5.091		&		0.078		&		48.279	&       49.325	\\
 			&		&	$(	268.224	)$	&	$(	116.158	)$	&	$(	0.396	)$	&	$(	0.616	)$	&	(935.354)& (0.328)	&	$(	236.564	)$	&	$(	0.408	)$	&	$(	2370.414	)$ &      (2466.566)	\\
 			&	$5\boldsymbol{I}_{10}$	&		16.357		&		14.035		&		0.119		&		0.163		&		 25.640  &  5.964		&		14.062		&		0.110		&		48.664	&       48.745	\\
 			&		&	$(	272.519	)$	&	$(	231.636	)$	&	$(	0.620	)$	&	$(	1.092	)$	&	 (1281.557)& (298.761)	&	$(	698.405	)$	&	$(	0.550	)$	&	$(	2408.094	)$ &      (2411.771)	\\[2ex] 
 			\hline
 		\end{tabular}
 		\caption{Estimated bias and mean squared errors (within parentheses) for outlying cluster contaminated datasets.}
 		\label{table3.3}
 	\end{table}
 	
 	\begin{table}[!h]
 		\hspace{-0.9in}
 		\small
 		\begin{tabular}{c c c c c c c c c c c c} 
 			\hline
 			&       \multicolumn{4}{c}{\hspace{8em}MPLE$_{\beta}$} & & \multicolumn{2}{c}{TCLUST} & \multicolumn{2}{c}{TKMEANS} & KMEDOIDS & MCLUST \\ [1ex] 
 			$Type$ &  $p$  &  $\beta=0$  & $\beta=0.1$ & $\beta=0.3$ & $\beta=0.5$ & $\alpha=0$ &  $\alpha=0.05$ & $\alpha=0$  & $\alpha=0.05$ & \\ [1ex] 
 			\hline
 			Pure	&	2	&	0.018	&	0.019		&	0.026	&	0.032		&	0.025	&	0.027		&	0.108	&	0.025		&	0.064	&	0.015	\\
 			&		&	(0.027)	&	(0.027	)	&	(0.030)	&	(0.034	)	&	(0.027)	&	(0.034	)	&	(0.044)	&	(0.054	)	&	(0.075)	&	(0.030	)	\\
 			&		&		&			&		&			&		&			&		&			&		&			\\
 			Pure	&	6	&	0.057	&	0.057		&	0.058	&	0.061		&	0.049	&	0.058		&	0.044	&	0.045		&	0.170	&	0.037		\\
 			&		&	(0.084)	&	(0.087	)	&	(0.107)	&	(0.141	)	&	(0.105)	&	(0.123	)	&	(0.098)	&	(0.115	)	&	(2.422)	&	(0.086	)	\\
 			& & & & & & &\\
 			& & & & & & &\\
 			&		&		&			&		&			&	$\alpha=$0.10	&	$\alpha=$0.15		&	$\alpha=$0.10	&	$\alpha=$0.15		&		&			\\
 			Uniform	&	2	&	0.136	&	0.094		&	0.060	&	0.042		&	0.030	&	0.047		&	0.135	&	0.062		&	0.102	&	0.042		\\
 			(chi-squared)&		&	(0.134)	&	(0.098	)	&	(0.041)	&	(0.043	)	&	(0.042)	&	(0.040	)	&	(0.064)	&	(0.061	)	&	(0.353)	&	(0.047	)	\\
 			Contaminated&		&		&			&		&			&		&			&		&			&		&			\\
 			&		&		&			&		&			&		&			&		&			&		&			\\
 			Uniform	&	6	&	0.077	&	0.114		&	0.053	&	0.066		&	0.073	&	0.055		&	0.049	&	0.044		&	0.249	&	0.023		\\
 			(chi-squared)&		&	(0.251)	&	(0.115	)	&	(0.121)	&	(0.154	)	&	(0.125)	&	(0.123	)	&	(0.113)	&	(0.130	)	&	(2.639)	&	(0.101	)	\\
 			Contaminated&		&		&			&		&			&		&			&		&			&		&			\\
 			&		&		&			&		&			&		&			&		&			&		&			\\
 			Uniform	&	2	&	0.586	&	0.820		&	0.018	&	0.025		&	0.142	&	0.048		&	0.185	&	0.030		&	0.290	&	0.032		\\
 			(Annulus)&		&	(0.966)	&	(0.940	)	&	(0.034)	&	(0.037	)	&	(0.093)	&	(0.047	)	&	(0.106)	&	(0.062	)	&	(0.285)	&	(0.040	)	\\
 			Contaminated&		&		&			&		&			&		&			&		&			&		&			\\
 			&		&		&			&		&			&		&			&		&			&		&			\\
 			Uniform	&	6	&	0.087	&	0.037		&	0.044	&	0.049		&	0.040	&	0.045		&	0.057	&	0.050		&	0.193	&	0.061		\\
 			(Annulus)&		&	(0.357)	&	(0.107	)	&	(0.126)	&	(0.156	)	&	(0.107)	&	(0.116	)	&	(0.112)	&	(0.123	)	&	(2.653)	&	(0.102	)	\\
 			Contaminated&		&		&			&		&			&		&			&		&			&		&			\\
 			&		&		&			&		&			&		&			&		&			&		&			\\
 			
 			Outlying	&	2	&	4.911	&	4.538		&	0.018	&	0.017		& 6.0677  &  0.037 	&	5.318	&	0.064		&	21.599	&	21.505		\\
 			Cluster&		&	(24.337)	&	(20.865	)	&	(0.040)	&	(0.044	)	&(129.004)& (0.050)	&	(110.232)	&	(0.057	)	&	(467.394)	&	(462.341	)	\\
 			&		&		&			&		&			&		&			&		&			&		&			\\
 			Outlying	&	6	&	8.622	&	6.970		&	0.039	&	0.047		&	  9.707   &  0.048		&	6.810	&	0.074		&	36.805	&	37.312		\\
 			Cluster&		&	(83.850)	&	(70.038	)	&	(0.127)	&	(0.158	)	&	(358.676)& (0.129)	&	(257.010)	&	(0.128	)	&	(1389.584	&	(1392.408	)	\\[2ex] 
 			\hline
 		\end{tabular}
 		\caption{Estimated bias and mean squared errors (within parentheses) for datasets with differentially dispersed clusters.}
 		\label{table3.4}
 	\end{table}
 	\newpage 
 	\section{Optimal choice of $\beta$}
 	\label{appen5}
 	The optimal selection of the $\beta$ parameter in the minimum Density Power Divergence estimation has been explored primarily in simple estimation scenarios by, among others, Warwick and Jones $(2005)$. They essentially evaluated the performance of the estimator through its asymptotic summed mean squared error (MSE)  
 	\begin{equation} \label{cc} 
 	E\left\{\left(\hat{\theta}_{\beta}-\theta^*\right)^T\left(\hat{\theta}_{\beta}-\theta^*\right)\right\}=n^{-1}{\rm 
 		tr}\left\{J_{\beta}^{-1}\left(\theta^g_{\beta}\right)K_{\beta}\left(\theta^g_{\beta}\right)J_{\beta}^{-1}\left(\theta^g_{\beta}\right)\right\}+\left(\theta^g_{\beta}-\theta^*\right)^T\left(\theta^g_{\beta}-\theta^*\right), 
 	\end{equation}  
 	where $\theta^g_{\beta}$ is the best fitting parameter value (i.e., $\theta^g_{\beta}=\underset{\theta \in \Theta}{\text{argmin}}\;D_{\beta}(g,f_{\theta})$) and $\theta^*$ is the actual target parameter. If the true unknown density belongs to the model family, $\theta^*=\theta^g_{\beta}$. The detailed expressions of $J_{\beta}\left(\theta\right)$ and $K_{\beta}\left(\theta\right)$, the matrices involved in the asymptotic variance of the minimum DPD estimator with tuning parameter $\beta$, can be found in Basu et al. $(1998)$, Warwick and Jones $(2005)$ or Basak et al. $(2021)$.  Now, there are two unknown quantities in Equation (\ref{cc}), namely, $\theta^g_{\beta}$ and $\theta^*$. In practice, Warwick and Jones $(2005)$ suggested $\theta^g_{\beta}$ to be replaced by $\hat{\theta}_{\beta}$, the MDPDE of $\theta$ and $\theta^*$ to be replaced by some suitable robust pilot estimator $\hat{\theta}^{P}$. Later, Basak et al. $(2021)$ suggested an iterative procedure to find the optimal choice of the tuning parameter $\beta$. This is basically a refinement of the procedure proposed by Warwick and Jones $(2005)$ which chooses the optimal parameter estimate at any particular stage as the pilot estimate for the next stage and continues the process till convergence.
 	
 	Any of the approaches above can be adapted to the clustering situation considered by us. However, there is one issue involved here. The literature has so far considered the estimation of an optimal $\beta$ parameter for one sample of real data. Our situation is a little more complex than that. Here one has to estimate the parameters of each of the $k$ populations, separately, and then again this process has to be repeated in each iteration. Thus, tuning parameter selection has to be performed separately for each population at each iteration. Therefore it does not make sense to talk about a single optimal value of $\beta$, as these cases may result in distinct optimals. Otherwise there is no conceptual difficulty in optimal tuning parameter selection in our clustering problem. However, we believe that the implementation of this should be reserved for a sequel paper, given that the level of refinement is quite intricate. 
 	
 	\newpage 
 	\section{Determination of $T$}
 	\label{appen6}
 	We choose the optimal value of $T$ following the maximal-gap philosophy. To understand this, we present the plots of the negative log likelihood values of the sample observations along with the choice of $-\text{log}(T)$ in Figure \ref{figure2}. The figure clearly indicates how the cut-offs correspond to maximal gaps.
 	\renewcommand{\thefigure}{S4}
 	\begin{figure}
 		\centering
 		\begin{tabular}{cc}
 			\begin{subfigure}{0.5\textwidth}\centering\includegraphics[height=3in,width=.8\columnwidth]{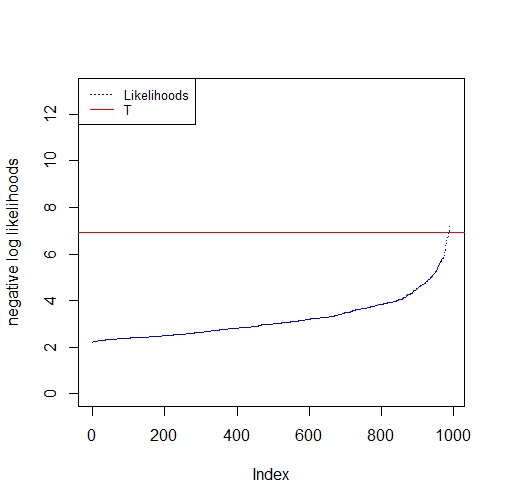}\caption{Pure Data }\end{subfigure}&
 			\begin{subfigure}{0.5\textwidth}\centering\includegraphics[height=3in,width=0.8\columnwidth]{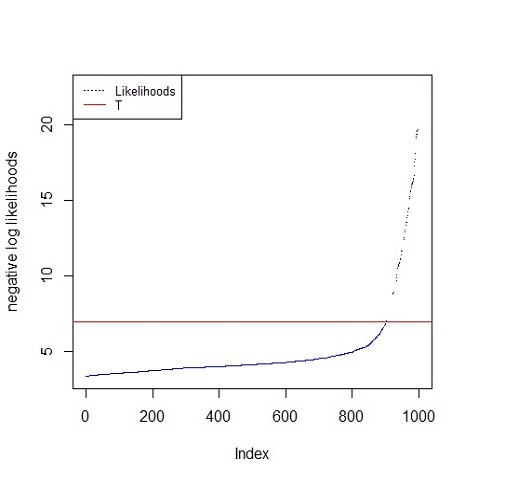}\caption{Uniformly  contaminated (chi-squared type) data}\end{subfigure}\\
 			\begin{subfigure}{0.5\textwidth}\centering\includegraphics[height=3 in,width=.8\columnwidth]{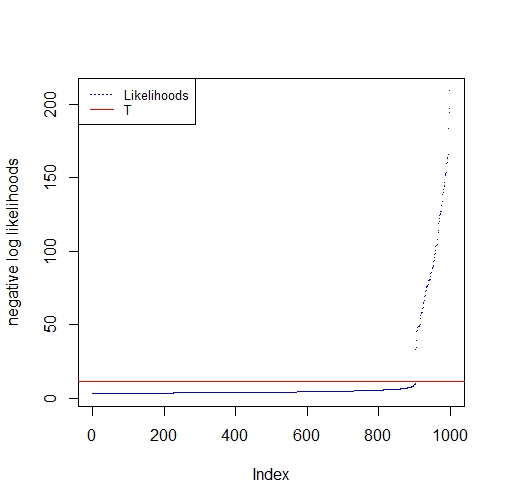}\caption{Annulus contaminated Data}\end{subfigure}&
 			\begin{subfigure}{0.5\textwidth}\centering\includegraphics[height=3 in,width=.8\columnwidth]{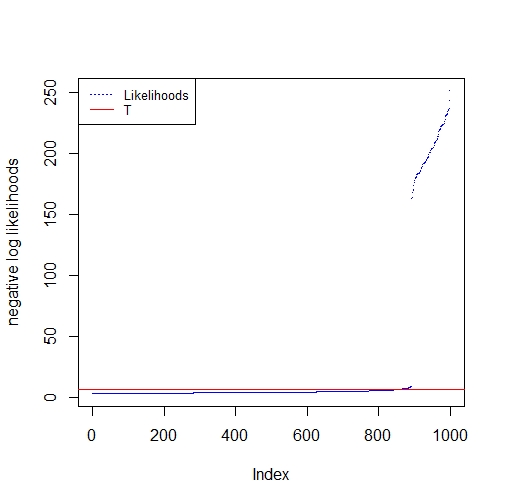}\caption{Outlying cluster Contaminated Data}\end{subfigure}
 		\end{tabular}
 		\caption{Plots of the negative log likelihoods of sample observations with the optimal choice of $-\text{log}(T)$ as outlier thresholds.}
 		\label{figure2}
 	\end{figure}
 	To implement the maximal-gap procedure in the R software, the following codes can be useful. 
 	\newpage 
 	\begin{verbatim}
 	l=sort(L) % sorted vertor of likelihoods.
 	gaps=c()
 	for(i in 1:(n-1))
 	{
 	gaps[i]=l[i+1]-l[i]  % Construction of successive gaps.
 	}
 	u=sort.list(gaps)[n-1]  % Position of the maximal-gap.
 	T=l[u]
 	\end{verbatim}
 	If there is no contamination in the data, normality of the clusters suggests that the concentration of the sample observations will be more in the central region and less in the low probability zones which are distant from the cluster means. Thus, the maximum gap is expected to be found in the low probability zones containing a few extreme observations. Hence, in case of no contamination, the maximal-gap strategy will, in most cases, find the threshold far away from the central region and the number of declared outliers will be small, as one would ideally expect. 
 	
 	As we have mentioned in the main manuscript (Remark $2.3$), there could be situations where the performance of the maximal-gap strategy may not be fully satisfactory. Using the $t$-th largest gap (rather than the overall largest) may provide a way around in such cases but the optimal choice of this $t$ may itself pose a major challenge. All in all, the choice of the tuning parameter $T$ remains a difficult problem for which a fully satisfactory solution, covering all different cases, is not yet available.

 \end{appendices}

 \end{document}